\documentclass[acmsmall]{acmart}

\usepackage[capitalise,noabbrev]{cleveref}
\usepackage[T1]{fontenc}
\usepackage[utf8]{inputenc}
\usepackage{microtype}
\usepackage{booktabs,multirow,array,makecell}
\usepackage{enumitem}
\setlist{nosep,leftmargin=1.4em}
\usepackage{xcolor}
\usepackage{tikz}
\usetikzlibrary{arrows.meta,positioning,shapes.geometric,calc,fit}
\usepackage{algorithm}
\usepackage{algpseudocode}

\theoremstyle{plain}
\newtheorem{theorem}{Theorem}
\newtheorem{proposition}[theorem]{Proposition}
\newtheorem{lemma}[theorem]{Lemma}
\theoremstyle{definition}
\newtheorem{definition}[theorem]{Definition}
\theoremstyle{remark}

\newcommand{\R}{\mathbb{R}}
\newcommand{\B}{\{0,1\}}
\newcommand{\nnz}{\mathrm{nnz}}
\renewcommand{\lg}{\log_2}
\newcommand{\Vv}{\mathbf{v}}
\newcommand{\xx}{\mathbf{x}}
\newcommand{\yy}{\mathbf{y}}

\newcommand{\ww}{\mathbf{w}}
\newcommand{\Oh}{\mathcal{O}}

\newcommand{\T}{^{\top}}
\newcommand{\gonzalo}[1]{\textcolor{red}{Gonzalo says: #1}}
\newcommand{\chg}[1]{\textcolor{black!80!black}{#1}}
\newcommand{\chgbegin}{\color{black!80!black}}
\newcommand{\chgend}{\color{black}}
\newcommand{\no}[1]{}

\title[Computation-Friendly Compression of Matrices]{Computation-Friendly Compression of Matrices:\\
Speeding up Linear and Boolean Algebra}
\author{Gonzalo Navarro}
\affiliation{%
  \institution{Universidad de Chile}
  \department{Department of Computer Science}   % opcional
  \city{Santiago}
  \country{Chile}
}
\email{gnavarro@uchile.cl}

\begin{document}

\begin{abstract}
Large sparse matrices---adjacency matrices of web and social graphs, feature matrices of machine-learning
pipelines, the relations of graph databases---are routinely compressed to save space. Classical compressors,
however, force the matrix to be decompressed before it can be used, so the space savings often yield a
computational overhead. A recent line of work studies \emph{computation-friendly} compression: formats
that are not only small, but let us carry out matrix--vector products, and in some cases matrix--matrix
products, sums and closures, directly on the compressed representation {\em in time proportional to that compressed size}.
This survey organizes that literature around the kind of redundancy each format exploits: 
spatial clustering of the nonzeros ($k^2$-trees and their variants), similarity between
rows (Boldi--Vigna-style referencing and set-difference compression), repeated substrings (grammar
compression), and dense bipartite substructures (bicliques). For each family we describe the representation,
the algorithms, the theoretical guarantees and the empirical performance.
Further, we show that the last three families are instances of a single
general scheme, which we call a {\em sparse
straight-line factorization} of the matrix, whose forward and backward substitutions evaluate right and left products in time proportional to the size of the representation, and that the product of two such factorizations is again one, which gives matrix--matrix products with a compressed output. \chg{We also relate the model to linear circuits, give simple recompression steps, and study the size of the smallest factorization: worst-case bounds, separations from each family, and parallelism.} \chg{An appendix contains the proofs of the new results, and another one} reviews formats that must be decompressed (at least
partially) before operating, for completeness. We close with a cross-family comparison and a list of open
problems.
\end{abstract}




\maketitle

\section{Introduction}
\label{sec:intro}

Matrices are the workhorse of scientific computing, data analysis and machine learning, and the matrices that
matter in practice are often huge and sparse. The adjacency matrix of a web crawl has billions of rows but
only a few dozen nonzeros per row~\cite{boldi2004webgraph,newman2010networks}; the labeled graph behind a
knowledge base such as Wikidata decomposes into thousands of Boolean matrices, one per predicate~\cite{arroyuelo2025evaluating};
the design matrices of machine-learning pipelines have millions of rows with few distinct values per
column~\cite{elgohary2018compressed,ferragina2022improving}. Such matrices are routinely compressed, either to
fit in main memory, to reduce the bandwidth between memory and processor, or to cut storage and
transmission costs.

The price is that compression and computation have traditionally been at odds. A general-purpose
compressor (\texttt{gzip}, \texttt{xz}, \texttt{zstd}, \dots) or an entropy-coded format saves space on disk, but
a matrix--vector product on it requires decompressing the matrix, at least piece by piece (say, row by row), so the
decompression time is \emph{added} to the time of the standard operation instead of reducing it; and if the matrix is expanded as a whole, the peak memory is also that of the uncompressed matrix. The notion of \emph{computation-friendly} compression is more ambitious:
a representation that is as small as possible \emph{and} on which the operation of interest can
be run directly, with time (and working space) {\em proportional to the compressed} size rather than to the size
of the matrix it stands for. If this can be done, compression becomes an algorithmic accelerator, not just a
storage device.

Is this feasible? For the sparsest possible redundancy, yes, and this is what every sparse-matrix format does:
the product with a vector can be computed in time proportional to the number of nonzeros $\nnz(A)$, instead of
the total $n\cdot m$ cells. The question is whether \emph{more powerful} forms of compression, which go beyond sparsity and
exploit {\em the structure} of the nonzeros, can also be made to speed up computation. There are good reasons for
pessimism. \citet{abboud2020impossibility} show that, under standard complexity assumptions, some of the most
effective compression schemes (grammar compressors, here including the Lempel--Ziv family \cite{ziv1977universal}) cannot support even
inner products in time proportional to their compressed size in the worst case, and
\citet{chakraborty2018tight} prove cell-probe lower bounds for Boolean matrix--vector
multiplication on any data structure that uses little space. The positive results surveyed here therefore
are not worst-case theorems about arbitrary matrices: they are guarantees stated in terms of the size of
a particular compressed representation, which is small only on applications where matrices have the structure the
representation exploits.

\paragraph{Scope.}
This survey covers \emph{lossless} matrix representations that (i)~can be significantly smaller than the plain
sparse representation on matrices with suitable structure, and (ii)~support matrix--vector multiplication,
or other algebraic operations such as matrix--matrix sums, products or transitive closure, on the compressed
form. We are primarily interested in matrix--vector and matrix--matrix products, with sums and other
pointwise operations as secondary operations (because they are easier and/or less frequently used). We organize the 
literature into four families according to the type of redundancy they exploit:
\begin{enumerate}
\item \textbf{$k^2$-tree based representations} (\cref{sec:quadtree}), which exploit the spatial
  clustering of the nonzeros and, by recursively partitioning the matrix, support a complete Boolean algebra
  \cite{brisaboa2014compact,arroyuelo2025evaluating,farina2025cache,arroyuelo2026boosting};
\item \textbf{Referencing / set-difference formats} (\cref{sec:reference}), where a row is described in terms
  of another, similar row (the idea behind the WebGraph framework of \citet{boldi2004webgraph}
  and, more generally, behind the compression of ``sets of sets'' by symmetric
  differences~\cite{bookstein1991compression,gagie2026set});
\item \textbf{Grammar-compressed matrices} (\cref{sec:grammar}), where repeated sequences of
  (value, column) pairs across rows are replaced by nonterminals of a context-free
  grammar~\cite{ferragina2022improving};
\item \textbf{Biclique decompositions} (\cref{sec:biclique}), where complete bipartite subgraphs of the
  adjacency matrix are extracted and shared~\cite{hernandez2014compressed,francisco2022graph}, a representation that also supports Boolean matrix--matrix multiplication directly on the compressed matrices~\cite{araya2026compressed}.
\end{enumerate}

We present first the $k^2$-trees, which offer good space on clustered matrices and the most general algebra, but 
also the least gains in time performance. We then present the other three families,
which achieve considerable speedups and still yield relevant compression
ratios. Experimental numbers reported in \cref{sec:quadtree,sec:reference,sec:grammar,sec:biclique} come from the original papers, obtained on different machines, datasets and implementations; they are not comparable across families, and 
we use them only to illustrate trends.
For completeness, Appendix~\ref{sec:decompress} briefly reviews the formats that must be decoded, at least in part,
before they are operated upon.

\paragraph{A unifying view.}
Beyond describing each format, we stress in \cref{sec:factorization} one observation: referencing, grammar and biclique formats are
all instances of a single scheme, a \emph{sparse straight-line factorization} of the matrix,
$A=P\,(I-Q)^{-1}\,T+T'$ with four sparse matrices. A right product is then a forward substitution and a left
product the transposed backward substitution, each costing time proportional to the number of nonzeros in the factors. This
single fact explains why all these formats multiply in time proportional to their size, why left and right products cost the same,
why the saving in arithmetic operations equals the compression ratio relative to the plain sparse format (and the observed speedups stay below it), and what really distinguishes the families: the shape
and depth of the sharing structure $Q$ and, above all, the (heuristic, expensive) way in which the factorization is found. We use this view throughout the survey,
specializing it in each of \cref{sec:reference,sec:grammar,sec:biclique}, and we revisit it when comparing the families (\cref{sec:comparison}) and in the open problems
(\cref{sec:open}). The same view shows that the product of two factorized
matrices is again a factorization, computed from the compressed pieces
(\cref{sec:factprod}), which generalizes the compositional output of the
biclique algorithms; the size of that output cannot be bounded in general, so
recompression may be needed. The view is ours and is not stated in the
original papers; we also mark other contributions of our own throughout.
\chg{\Cref{sec:factorization} also contains new results on this model: its equivalence with linear circuits (so that circuit bounds apply), the cost of sums, products and Kronecker products, linear-time recompression steps, the minimum size of a factorization (worst-case bounds and separations from each family), and the depth as a measure of parallelism. The proofs are in \cref{app:proofs}.}

\no{
\paragraph{Contributions and organization.}
Beyond summarizing each paper, we (a)~give a common notation and a precise notion of computation-friendly
compression (\cref{sec:background}); (b)~show in \cref{sec:factorization}, in detail and with a worked example, that the referencing, grammar and biclique
formats are all instances of a single scheme, the evaluation of a product through a
sparse, acyclic factorization of the matrix, and that the product of two such factorizations is again one (\cref{sec:factprod}), whereas $k^2$-trees follow a different principle based on
recursive partitioning; (c)~derive a few simple facts that the original papers leave implicit, such as the
\emph{left} product for referencing formats; (d)~compare the families in \cref{sec:comparison}; and
(e)~collect open problems in \cref{sec:open}. We point out clearly where a statement is our own observation
rather than a result of the surveyed papers.
}

\section{Background and Problem Statement}
\label{sec:background}

\subsection{Notation and operations}

Let $A$ be an $n\times m$ matrix over a semiring $(\mathbb{S},+,\cdot)$. We mostly consider two
cases: real matrices on the semiring $(\mathbb{R},+,\cdot)$, and Boolean matrices on
$(\B,\vee,\wedge)$. Boolean matrices model the adjacency matrices of directed graphs
($A_{uv}=1$ iff $(u,v)$ is an edge) and, more generally, of binary relations. We write $\nnz(A)$ for the number of
nonzero entries, $A_{i\ast}$ for the $i$-th row, and $\Vv_i\in\mathbb{S}^m$ for the row viewed as a vector. When $A$ is Boolean it is
often convenient to regard each row as the \emph{set} of columns in which it has a $1$, so that $A$ becomes a
\emph{set of sets} (e.g., the adjacency lists of the graph).

The operations that we study, in decreasing order of interest, are:
\begin{description}[style=nextline,font=\normalfont\bfseries]
\item[Right product:] $\yy = A\,\xx$, with $\yy\in\mathbb{S}^n$ and $\xx\in\mathbb{S}^m$ (usually dense).
\item[Left product:] $\xx\T = \yy\T A$, the product by a row vector, i.e., the product by $A\T$.
\item[Matrix--matrix product:] $C=A\times B$, with $C_{ij}=\sum A_{ik}\cdot B_{kj}$.
\item[Pointwise operations:] sum $A+B$, and (for Boolean matrices) intersection, difference, symmetric difference.
\item[Transposition and closure:] $A\T$ and the transitive closure $A^{+}=A+A^2+\cdots+A^n$ of Boolean matrices.
\end{description}

In this survey we focus mostly on the \emph{right and left products}. They are the basic step of the power
method, of conjugate gradient and other iterative solvers, of the training and inference of linear
models, and of the propagation steps of graph neural networks (GNNs). A typical running example is
PageRank~\cite{page1999pagerank,langville2006google}: given the row-normalized adjacency matrix $M=D^{-1}A$ of a
graph, a teleport probability $0<\alpha<1$ and a start distribution $\mathbf{p}_0$, one iterates
\begin{equation}\label{eq:pagerank}
\mathbf{p}_t=\alpha\,\mathbf{p}_0+(1-\alpha)\,\mathbf{p}_{t-1}\,M ,
\end{equation}
so that the cost of the computation is dominated by repeated products of one fixed matrix by changing
vectors~\cite{francisco2022graph,chung2007heat}. Because the same matrix is multiplied many times, it pays off to
spend a (possibly large) one-time effort to compress it.

Matrix--matrix products and transitive closure are the core of Boolean matrix algebra, and arise in context-free
and regular path queries over labeled graphs~\cite{arroyuelo2025evaluating,losemann2013complexity,angles2017foundations}, in
triangle counting, and in join-project queries in databases~\cite{amossen2009faster}.

\subsection{Plain sparse formats}
\label{sec:plain}

All the formats below are measured against the plain \emph{sparse} formats~\cite{saad2003iterative}.
\emph{Coordinate} (COO) stores the triples $(i,j,A_{ij})$ of the nonzeros. \emph{Compressed sparse row} (CSR)
sorts the nonzeros by row and stores three arrays: the nonzero values, their column indices, and, for each row,
the position where its nonzeros start; \emph{CSC} does the same by column. For real matrices with few
distinct values, \emph{CSR-IV}~\cite{kourtis2008optimizing} and its variant CSRV~\cite{ferragina2022improving} replace the values by indices into a table of distinct
values. With $w$-bit words, CSR takes about
$\nnz(A)\,(\lg m+w)+n\lg\nnz(A)$ bits, and both products with a vector take $\Theta(\nnz(A))$ time. Boolean
CSR drops the values, and takes essentially $\nnz(A)\lg m$ bits. Those figures are the baseline: the
\emph{compressed size} $|\mathcal{C}(A)|$ of a representation $\mathcal{C}$ will be compared with
$\nnz(A)$, and the \emph{speedup potential} of a computation-friendly format will be the ratio
$\nnz(A)/|\mathcal{C}(A)|$, measured in the appropriate unit.

\paragraph{Product of sparse matrices (Schoor's algorithm).}
For the product $C=A\times B$ of two sparse Boolean matrices, the classical algorithm of \citet{schoor1982fast} works best if $A$ is sorted by columns (CSC) and $B$ by rows (CSR). It intersects the set of nonempty columns of $A$ with the set of nonempty rows of $B$. For each $t$ in the intersection, every $A_{it}=1$ and every $B_{tj}=1$ contribute to $C_{ij}=1$: the index $t$ generates the Cartesian product of the rows $i$ of column $t$ of $A$ and the columns $j$ of row $t$ of $B$, which is a \emph{biclique} of $C$. Different values of $t$ may generate the same cell, and the duplicates must be removed (for example, with a heap if the aim is to produce the output in CSR or CSC order). The cost is dominated by the number $\sum_t c_tr_t$ of elementary products, where $c_t$ and $r_t$ are the numbers of nonzeros of column $t$ of $A$ and of row $t$ of $B$. This is $\Oh(\nnz(A)\nnz(B)/n)$ on average if the 1s are uniformly distributed \cite{schoor1982fast}.

\subsection{What is computation-friendly compression?}

\begin{definition}[Computation-friendly compression]\label{def:cf}
A representation $\mathcal{C}(A)$ of a matrix $A$ is \emph{computa\-tion-friendly} for an operation $\textsc{op}$ if
$\textsc{op}$ can be computed from $\mathcal{C}(A)$, in time and working space $\Oh(|\mathcal{C}(A)|)$
plus the size of the output. 
%, and with working space proportional to the compressed
%representation, or at least smaller than the size of $A$ in plain form.
\end{definition}

It follows from Definition~\ref{def:cf} that \textsc{op} cannot be computed by first materializing
$A$ (unless $\mathcal{C}$ never obtains asymptotic compression). The definition is deliberately loose about the unit in which $|\mathcal{C}(A)|$ is measured. For $k^2$-tree compression, the space is measured in \emph{bits} and 
the times are proportional to their amount. For the other three representation families considered in the survey 
the size is measured in \emph{symbols} (integers, signatures), and the time is linear in the number of those 
symbols. Those symbols can be further compressed with variable-length codes; the cost of decoding those is a
constant per symbol in the cases we discuss, but it matters in practice (\cref{sec:grammar}).

An important aspect of these compression formats is what operations
\textsc{op} they can support efficiently, which can range from just
matrix--vector multiplications to a complete matrix algebra. To enable 
a matrix algebra, compositionality is desirable: can the output be directly
expressed in the same compressed format, or does one need to recompress it?
We also consider the basic problem of efficiently accessing random cells
$A_{ij}$.

\no{
Three further properties are worth distinguishing when comparing the formats:
\begin{itemize}
\item \emph{Which operations} are supported on the compressed form (products only, or also sums,
  intersections, transposition, closure)?
\item \emph{Closure under the operation} (\emph{compositionality}): is the output of a matrix--matrix operation
  again in the same compressed format, so that it can be fed to the next operation without decoding? This is
  crucial for algebras of matrices~\cite{arroyuelo2025evaluating,araya2026compressed}, but does not apply to products by a dense vector, whose output is a dense vector.
\item \emph{Random access} (e.g., access to a cell or to a row), which is useful, say, in graph traversals, but is not
  required by products.
\gonzalo{Revisar si esto vale la pena}
\end{itemize}
}

We note that the term \emph{compressed matrix multiplication} is used by \citet{pagh2013compressed} for an unrelated notion: approximating a product by a randomized sketch of the output. It is lossy and approximate, and not covered here.

\subsection{Lower bounds and what they imply}
\label{sec:limits}

Two families of lower bounds apply to our problem. First, \citet{abboud2020impossibility}
show that, for vectors compressed by applying a grammar compressor (a class where they include the Lempel--Ziv family)
to the plain sequence of their values, even the inner product cannot be computed
in time linear in the compressed size, in the worst case and under standard fine-grained complexity
assumptions; in contrast, for simpler schemes such as run-length encoding it can. Second,
\citet{chakraborty2018tight} prove that any data structure that stores a Boolean $n\times n$ matrix
in $r$ bits, for $n<r<n^2$, requires time $t\in\tilde\Omega(n^3/r)$ for matrix--vector
multiplication, a trade-off that no compression can circumvent in the worst case. Related hardness arises
in the \emph{online} matrix--vector (OMV) setting~\cite{henzinger2015unifying,larsen2017faster}. Finally,
the fastest algorithms for general matrix multiplication~\cite{strassen1969gaussian,coppersmith1990matrix,alman2018further}
are asymptotic and need rings, whereas the Boolean case is only a semiring; the practical alternatives are
combinatorial algorithms~\cite{arlazarov1970economical}, algorithms for sparse
inputs~\cite{yuster2005fast,amossen2009faster,schoor1982fast,gustavson1978two}, and the output-sensitive
ones.

The surveyed papers escape these results for two complementary reasons. The first is that their guarantees are
\emph{instance-dependent}: the bounds are in terms of the size of a particular representation (number of
nonzeros of a residual matrix, size of a grammar, number of nodes of a tree), which is small only on
compressible inputs. For example, a sparse matrix $A$ represented in COO, CSR,
or a similar format in $\Oh(\nnz(A))$ space is easily multiplied by a dense vector in $\Oh(n+\nnz(A))$
time. The second, which we believe is crucial for grammars, is that the representations
store \emph{positions explicitly}: in the sequence of (value, column) pairs of \cref{sec:grammar} a repeated
substring always refers to the same columns, so that both align with the same part of the vector $\xx$. This is
what the schemes covered by the impossibility results lack. We come back to this in \cref{sec:grammar}.

\subsection{Computational model}

Unless stated otherwise we use the word RAM model with $\Theta(\log N)$-bit words, as is customary in
compact data structures~\cite{navarro2016compact}. Several of the structures below rely on succinct trees and bitvectors with
$\mathrm{rank}$ and $\mathrm{select}$ support~\cite{jacobson1989space,navarro2016compact}; we recall only what we
need, where we need it.

\begin{table}[t]
\centering
\small
\caption{Summary of the four families of computation-friendly formats and the operations supported on compressed form.
``Mat-vec'' denotes products by a dense vector.
Complexities are those proved in the cited papers (see the corresponding sections for notation);
all are $\Oh(|\mathcal{C}|+n)$ with the corresponding notion of $|\mathcal{C}|$.}
\label{tab:taxonomy}
{\setlength{\tabcolsep}{4pt}%
\begin{tabular}{@{}p{1.9cm}p{2.5cm}p{1.9cm}p{3.0cm}p{3.0cm}@{}}
\toprule
Family & Redundancy exploited & Matrices & Operations & Time of mat-vec\\
\midrule
%\gonzalo{Ojo!} \claude{Un recorrido visita todos los nodos:
%$O(\nnz\log(n^2/\nnz))$, y $\Omega(\nnz)$ salvo que se agreguen bloques
%llenos; los papers no lo eval\'uan.} \gonzalo{Es verdad pero preferí
%simplificar la fórmula}
$k^2$-tree variants (\cref{sec:quadtree}) & Clustering of the 1s in the plane & Binary & Full Boolean algebra: $+,\cap,-,\oplus,\times,{}^{+}$, $\T$ & $\Oh(\nnz(A)\log n+n)$ \\[3pt]
Referencing  (\cref{sec:reference}) & Rows similar to other rows & Binary; extendable to real & Mat-vec (right; left by duality); access, etc. & $\Oh(\nnz(A')+n)$ \\[3pt] %, where $A'$ is the residual \\[3pt]
Grammars (\cref{sec:grammar}) & Repeated (value, column) subsequences & Real, few distinct values or sparse & Mat-vec (right and left) & $\Oh(|C|+|R|)$ \\[3pt]
Bicliques$\,\,\,$ (\cref{sec:biclique}) & Complete bipartite subgraphs & Binary & Mat-vec (right and left); Compositional mat-mat & $\Oh(\sum_r(|S_r|+|T_r|)+\nnz(R_A)+n)$ \\ %, where $R_A$ is the residual \\
\bottomrule
\end{tabular}}
\end{table}

\subsection{A taxonomy of computation-friendly formats}
\label{sec:taxonomy}

We classify the formats according to the kind of redundancy they extract from the matrix. 
%\Cref{fig:taxonomy} shows the resulting map, and 
\cref{tab:taxonomy} summarizes the families covered in \cref{sec:quadtree,sec:reference,sec:grammar,sec:biclique}. After presenting them, we show in \cref{sec:factorization} that the last three are instances of a single scheme.

\section{$k^2$-Trees and their Algebra}
\label{sec:quadtree}

The $k^2$-tree family exploits the \emph{spatial clustering} of the nonzeros in the
matrix by means of a recursive partitioning into quadrants. It
supports a rich Boolean algebra of matrices (sums, intersections, products, closure) whose results are again
$k^2$-trees. The work surveyed in this section was motivated by graph databases, where the algebra is used to solve
regular path queries (RPQs), but it is of independent interest as a compressed-matrix toolbox.

\subsection{The $k^2$-tree}

The $k^2$-tree of \citet{brisaboa2014compact,brisaboa2009k2} represents a Boolean matrix
$A$ of side $n$ (padded to a power of $k$) as a cardinal tree: the root stands for the whole matrix, which is divided into
$k^2$ equal quadrants, each of which is either empty (a leaf) or is recursively represented. It is a compact
representation of the MX-quadtree~\cite{finkel1974quad,samet1984quadtree}. We use $k=2$ throughout, so each node has four possible
children, in the order top-left, top-right, bottom-left, bottom-right. The tree is
stored level by level, writing for each internal node a $4$-bit \emph{signature} that marks which quadrants are nonempty
(\cref{fig:k2tree}); the resulting bitvector $L$ is the cardinal version of the LOUDS representation of
trees~\cite{jacobson1989space}, and navigation uses $\mathrm{rank}(L,i)$, the number of 1s in $L[1..i]$: the $t$-th nonempty child of the node whose
signature is the $v$-th one is the signature number $\mathrm{rank}(L,4(v-1))+t+1$. A constant-time $\mathrm{rank}$ needs $o(|L|)$ additional
bits~\cite{navarro2016compact}. If $a=\nnz(A)$, in the worst case $4a\log_4(n^2/a)(1+o(1))$ bits are used, and
the space is significantly smaller when the 1s are clustered~\cite{debernardo2023faster}: if the points can be
packed into $c$ square clusters of sides $n_1,\ldots,n_c$, holding $a_1,\ldots,a_c$ points each, then the 
$k^2$-tree uses $\Oh(c\log n + \sum_i a_i\log n_i)$ bits.
Variants of the $k^2$-tree are well suited to represent
Web graphs~\cite{brisaboa2014compact}, RDF stores~\cite{alvarez2015compressed}, raster images~\cite{debernardo2013compact}, and 
$d$-dimensional relations~\cite{arroyuelo2022optimal}, among several others.

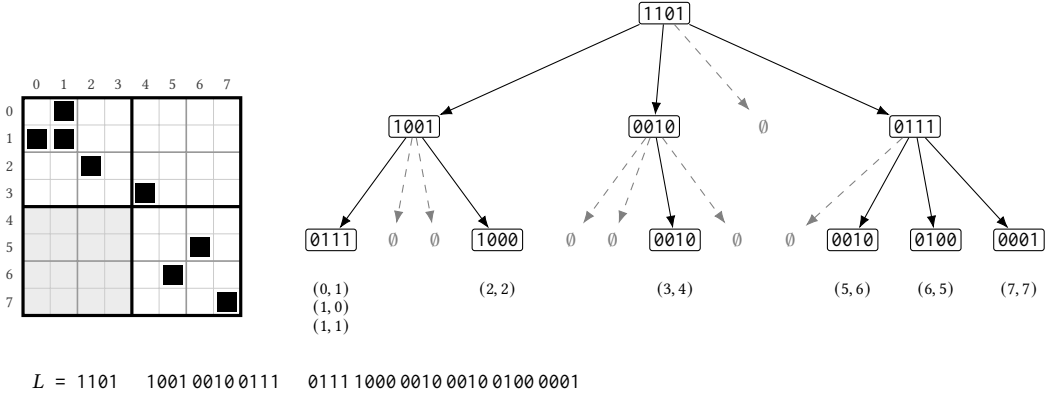
\begin{figure}[t]
\centering
\begin{tikzpicture}[x=0.5cm,y=0.5cm,>=Latex,
  sig/.style={draw,rounded corners=1pt,inner sep=1.5pt,font=\ttfamily\footnotesize,fill=white},
  em/.style={font=\scriptsize,gray}]
% ---- matrix (8x8): rows 0..7 from the top, columns 0..7
\begin{scope}[x=0.36cm,y=0.36cm]
\fill[gray!15] (0,0) rectangle (4,4); % empty bottom-left quadrant
\draw[step=1,gray!40,very thin] (0,0) grid (8,8);
\draw[step=2,gray!80,semithick] (0,0) grid (8,8);
\draw[very thick] (0,0) rectangle (8,8);
\draw[very thick] (4,0)--(4,8); \draw[very thick] (0,4)--(8,4);
\foreach \r/\c in {0/1,1/0,1/1,2/2,3/4,5/6,6/5,7/7}{\fill (\c+0.15,7-\r+0.15) rectangle (\c+0.85,7-\r+0.85);}
\foreach \i in {0,...,7}{\node[font=\tiny,gray!50!black] at (\i+0.5,8.5) {\i};\node[font=\tiny,gray!50!black] at (-0.5,7.5-\i) {\i};}
\end{scope}
% ---- tree
\node[sig] (r) at (16.97,8.0) {1101};
\node[sig] (a) at (10.35,5.0) {1001};
\node[sig] (b) at (16.70,5.0) {0010};
\node[sig] (c) at (23.60,5.0) {0111};
\node[em] (e0) at (19.60,5.0) {$\emptyset$};
\node[sig] (a0) at (8.15,2.0) {0111};
\node[em] (ea1) at (9.80,2.0) {$\emptyset$};
\node[em] (ea2) at (10.90,2.0) {$\emptyset$};
\node[sig] (a3) at (12.55,2.0) {1000};
\node[em] (eb0) at (14.50,2.0) {$\emptyset$};
\node[em] (eb1) at (15.60,2.0) {$\emptyset$};
\node[sig] (b2) at (17.25,2.0) {0010};
\node[em] (eb3) at (18.90,2.0) {$\emptyset$};
\node[em] (ec0) at (20.30,2.0) {$\emptyset$};
\node[sig] (c1) at (21.95,2.0) {0010};
\node[sig] (c2) at (24.15,2.0) {0100};
\node[sig] (c3) at (26.35,2.0) {0001};
\draw[->] (r) -- (a);\draw[->] (r) -- (b);\draw[->] (r) -- (c);\draw[->,gray,dashed] (r) -- (e0);
\draw[->] (a) -- (a0);
\draw[->,gray,dashed] (a) -- (ea1);
\draw[->,gray,dashed] (a) -- (ea2);
\draw[->] (a) -- (a3);
\draw[->,gray,dashed] (b) -- (eb0);
\draw[->,gray,dashed] (b) -- (eb1);
\draw[->] (b) -- (b2);
\draw[->,gray,dashed] (b) -- (eb3);
\draw[->,gray,dashed] (c) -- (ec0);
\draw[->] (c) -- (c1);
\draw[->] (c) -- (c2);
\draw[->] (c) -- (c3);
\node[font=\tiny,anchor=north,align=center] at (8.15,1.15) {$(0,1)$\\$(1,0)$\\$(1,1)$};
\node[font=\tiny,anchor=north] at (12.55,1.15) {$(2,2)$};
\node[font=\tiny,anchor=north] at (17.25,1.15) {$(3,4)$};
\node[font=\tiny,anchor=north] at (21.95,1.15) {$(5,6)$};
\node[font=\tiny,anchor=north] at (24.15,1.15) {$(6,5)$};
\node[font=\tiny,anchor=north] at (26.35,1.15) {$(7,7)$};
% bitvector
\node[font=\ttfamily\footnotesize,anchor=north west] at (0,-1.3) {$L$ = 1101 \; 1001\,0010\,0111 \; 0111\,1000\,0010\,0010\,0100\,0001};
\end{tikzpicture}
\caption{An $8\times8$ Boolean matrix and its $k^2$-tree for $k=2$, with three levels. The root signature $1101$ says that the top-left, top-right and bottom-right quadrants contain 1s and the bottom-left one (shaded) does not. Each nonempty quadrant has, at the next level, a signature of its own quadrants, and so on until the $2\times2$ submatrices, whose signatures mark the cells that are $1$ (listed below each of them). Empty quadrants (dashed) are not stored, nor are the cells. The bitvector $L$ is the concatenation of the signatures, level by level.}
\label{fig:k2tree}
\end{figure}

\paragraph{Supporting access.} The $k^2$-tree navigation allows computing any
$A_{ij}$ by descending along $\Oh(\log n)$ levels from the root, in $\Oh(1)$ time per level, until finding an 
empty submatrix or the explicit cell value.
At each signature, algebraic calculation with $(i,j)$ determines the child to follow, and a constant-time
rank operation computes the position of the signature of that child.

The format can be extended to associate semiring values (e.g., reals) to the nonempty cells. Those values can
be stored in an array $V[1..\nnz(A)]$ in the left-to-right order of the leaves. Once accessing a cell at some
leaf position $L[i]$, its corresponding real value is at $V[\textrm{rank}(L,i)-\textrm{rank}(L,z)]$, where $z$
is the last position in $L$ corresponding to signatures not in the last level. Several of the algorithms in this 
section can be extended to work on semiring values by using this technique, with some caveats: for example, in
the case of a ring like $(\mathbb{Z},+,\cdot)$, a zero cell can result from adding nonzero cells, which cannot 
occur, say, the Boolean semiring.

\subsection{Products by a vector}
\label{sec:k2vec}

Multiplying $A$ by a dense vector $\xx$ is simple and natural with the $k^2$-tree representation. Though it does not appear in the papers on the $k^2$-tree algebra, the implementation distributed with \citet{arroyuelo2025evaluating} (\url{https://github.com/adriangbrandon/rpq-matrix}) includes it; see \cref{alg:k2vec}.
The tree is traversed from the root. A single operation on a node returns its signature (the bitmap of its four quadrants) and the nodes of its nonempty children are obtained with a child operation (built on rank, as explained); the recursion continues only on those. The recursion keeps track of the corresponding portions of $\xx$ and $\yy$ to operate; note that $\ell$ is the side of the child quadrants. In the last level, the node is a cell with a $1$, and it adds one entry of $\xx$ to one entry of $\yy$. The matrix is padded to a power of two, but the vectors are not: the padding has no $1$s, so it is never reached. The code obtains the left product, $\yy\T A$, by computing $A\T\yy$. The transpose of $A$ is easily simulated by (virtually) exchanging the 2nd and 3rd bit of every node signature as it is traversed.

\begin{algorithm}[t]
\caption{Right product $\yy=A\xx$ on a $k^2$-tree with $k=2$, assuming $n$ has been rounded to the next power of $2$. Invoke with $\yy \gets 0$; \textsc{Mult}($\mathrm{root},n/2,1,1$); assume $\ell/2=0$ when $\ell=1$.}
\label{alg:k2vec}
\begin{algorithmic}[1]
\Procedure{Mult}{$v,\ell,j,i$}
  \If{$\ell=0$} \State $y_i\gets y_i+x_j$ \Comment{$v$ is a cell with a $1$}
  \Else
    \For{each quadrant $(a,b)\in\{0,1\}^2$ with bit $1$ in the signature of $v$}
      \State \Call{Mult}{$\mathrm{child}(v,a,b),\ \ell/2,\ j+b \ell,\ i+a \ell$}
    \EndFor
  \EndIf
\EndProcedure
%\State $\yy\gets 0$; \Call{Mult}{$\mathrm{root},L,\xx,\yy$}
\end{algorithmic}
\end{algorithm}

Every visited node across the computation has a $1$ in the bitmap of its parent, and a matrix with $a=\nnz(A)$ has $\Oh(a\log(n^2/a))$ such nodes, so the product time, including that of initializing $\yy$, is $\Oh(a\log(n^2/a)+n)$, with $\Oh(\log n)$ words of working space for the recursion. The traversal accesses the bitvector $L$ left-to-right, which is cache-friendly. Still, this time complexity is worse than the $\Oh(a+n)$ that is easily be obtained with a plain COO, CSR or CSC sparse representation; compared to those formats, the $k^2$-tree saves space, but not computation time. The formats of the next sections are able to speed this up.

\paragraph{Breaking the $\Theta(a)$ barrier (our derivation).}
A $k^2$-tree variant designed for matrices with large blocks of 1s \cite{debernardo2013compact} marks with $\Oh(1)$ bits the quadrants that are full of 1s. To multiply in time proportional to the number of nodes of that $k^2$-tree (which can be much less than $\nnz(A)$) one should do $y[i] \gets y[i] + \sum_{j=j_1}^{j_2} x[j]$ for all $i \in [i_1..i_2]$ when
arriving at a full quadrant $A[i_1..i_2,j_1..j_2]$. The sum is easily computed in constant time by precomputing 
prefix sums on $\xx$: $s[0]=0$ and $s[j]=s[j-1]+x[j]$ for all $j \ge 1$, thus $\sum_{j=j_1}^{j_2} x[j] = s[j_2]-s[j_1-1]$. Writing all the $y$ values can be done efficiently by creating a perfect binary tree on top of $\yy$,
where nodes at depth $d$ cover consecutive ranges of $n/2^d$ cells of $\yy$. The nodes have counters initialized
at zero, and a full quadrant $A[i_1..i_2,j_1..j_2]$ increments the counter of the node that covers 
$y[i_1..i_2]$ by $s[j_2]-s[j_1-1]$. At the end, we traverse that binary tree from the root to each cell of $\yy$, accumulating the counters in the path from the root and writing the sum to each cell arrived at.

\subsection{A Boolean matrix algebra on $k^2$-trees}
\label{sec:algebra}

\citet{arroyuelo2023evaluating,arroyuelo2025evaluating} design and implement an algebra of Boolean
matrices on $k^2$-trees. The operations $C = A~\textsc{op}~B$ are applied to $n\times n$ matrices, and $C$ is 
returned in the same $k^2$-tree format, so the algebra is \emph{compositional}. We call $a=\nnz(A)$, $b=\nnz(B)$, and $c=\nnz(C)$.

\paragraph{Transposition.}
Since $A\T$ is obtained by exchanging the roles of the second and third quadrants at every node, the $k^2$-tree is
not modified: a \emph{transposed} flag is associated with the matrix, toggled in $\Oh(1)$ time, and honored by every operation in any submatrix of $A$, as illustrated with left products.

\paragraph{Sum.}
If neither matrix is transposed, $A+B$ is computed by a sequential pass over both bitvectors, level by level, with a
queue of tasks of two kinds: a \emph{merge} task outputs the bitwise OR of the next two signatures and enqueues, for
each child, a merge task (both have it), or a \emph{copy} task for the matrix that has it, or nothing; a \emph{copy} task copies one signature
and enqueues the copy tasks of its children. No $\mathrm{rank}$ operations are necessary; the time is proportional to the sizes of the two trees,
$\Oh(a\log(n^2/a)+b\log(n^2/b))\subseteq\Oh((a+b)\log n) = \Oh(c\log n)$, and consecutive copy tasks are merged so that they are
executed word-wise, not bit-wise.

\paragraph{Other pointwise operations.}
Intersection ($\cap$, Boolean AND), difference ($-$) and symmetric difference ($\oplus$) lack the property that the root signature
of the result is a function of the signatures of the operands. They are then computed recursively over the quadrants,
\begin{equation*}%\label{eq:setlike}
A\circ B=\begin{pmatrix}A_0\circ B_0 & A_1\circ B_1\\ A_2\circ B_2 & A_3\circ B_3\end{pmatrix},
\end{equation*}
writing each result directly into the level-wise output buffers, and copying whole subtrees word by word when one
operand is empty. The authors prove adaptive bounds in terms of the \emph{alternation} $\delta$ between the two sets of 
positions of the 1s~\cite{barbay2008alternation}: for example, the intersection takes $\Oh(\delta\log(n^2/\delta))$ time, with $c\le\delta\le\min(a,b)$.

\paragraph{Product.}
The multiplication follows the classic divide-and-conquer recursion on quadrants,
\begin{equation*}%\label{eq:mult}
A\times B=\begin{pmatrix}A_0B_0+A_1B_2 & A_0B_1+A_1B_3\\ A_2B_0+A_3B_2 & A_2B_1+A_3B_3\end{pmatrix}.
\end{equation*}
Empty quadrants (a $0$ in the signature) prune the recursion, and the sum of a product with an empty matrix is just a reference
to the product. The cost is shown to be $\Oh\bigl(\min(a,b)\sqrt{\max(a,b)}\,\log n\bigr)$ on average, for
uniformly distributed 1s (for $a=b$, this is $\Oh(a^{3/2} \log n)$).

\paragraph{Transitive closure.}
The standard approach, $A\gets A+A\times A$ until the matrix stops changing, needs up to $\lg n$ multiplications.
The authors give an algorithm inspired by Warshall's~\cite{warshall1962theorem} and by the ZCQ representation of
Penn~\cite{penn2006efficient} that uses only $\Oh(1)$ multiplications' worth of work. Splitting $A$ into quadrants,
a first stage computes $A'$, the matrix of paths whose intermediate nodes lie in the first half of the indices: one
starts computing $A_0'=A_0^+$ recursively, and then
\begin{equation*}%\label{eq:closure1}
A'=\begin{pmatrix}A_0^{+} & A_1+A'_0A_1\\ A_2+A_2A'_0 & A_3+A_2A'_1\end{pmatrix} %,\quad A'_0=A_0^{+},
\end{equation*}
(computing $A_3'$ after the other quadrants). The second stage allows the paths to cross the second half: one
first computes $A''_3=(A'_3)^{+}$ recursively, and then
\begin{equation*}%\label{eq:closure2}
A^{+}=A''=\begin{pmatrix}A'_0+A'_1A''_2 & A'_1+A'_1A''_3\\ A'_2+A''_3A'_2 & (A'_3)^{+}\end{pmatrix} %, \quad A''_3=(A'_3)^{+},\; A''_2=A'_2+A''_3A'_2,
\end{equation*}
(computing $A_0''$ after the other quadrants).
The number of recursive calls equals that of a multiplication, which makes the average time
$\Oh((a^{+})^{3/2}\log n)$ with $a^{+} = \nnz(A^{+})$, a $\log n$ factor better than iterated squaring.

\paragraph{Parallelism.}
%\paragraph{Restrictions, parallelism, plans.}
%Because a $k^2$-tree can be traversed in a way that ignores quadrants that do not intersect a given row or column, the
%algebra supports \emph{restricted} operations $\langle r\rangle A\langle c\rangle$, which are the key to answering path queries that start
%from a constant node. 
Multi-threaded versions of the recursive algorithms give near-perfect speedup in the PRAM model, though in
practice thread-creation overhead outweighs it on very sparse matrices.
%, and a heuristic query optimizer orders.
%sums (Huffman-like) and products (greedy).

\paragraph{An uncompressed baseline.}
\label{sec:baseline-ops}
The same paper implements a baseline for sparse Boolean matrices (CSR plus CSC views) with
Schoor's algorithm~\cite{schoor1982fast} (\cref{sec:plain}) for the product, taking $\Oh((ab/n)\log n)$ average time, and a closure based on strongly connected
components~\cite{purdom1970transitive,tarjan1972depth}. This baseline is used again in \cref{sec:biclique}. \Cref{tab:alg-complexity} compares the average times of both formats.

\begin{table}[t]
\centering\small
\caption{Simplified average-case times of the Boolean algebra on $n\times n$ matrices with $s$ uniformly distributed 1s; $s^{+}$ is the size of the closure % and $p$ the number of processors
(from Table~1 of~\cite{arroyuelo2025evaluating}).}
\label{tab:alg-complexity}
\begin{tabular}{@{}lcc@{}}
%\begin{tabular}{@{}lccc@{}}
\toprule
Operation & Baseline (CSR/CSC) & $k^2$-tree \\%& $k^2$-tree, parallel\\
\midrule
Transposition & $1$ & $1$ \\%& $1$\\
Sum and relatives & $s$ & $s\log n$ \\%& $\frac1p m\log n+m+\log n\log p$\\
Multiplication & $\frac{s^2}{n}\log n$ & $s^{3/2}\log n$ \\%& $\frac1p m^{3/2}\log n+\frac{m^2}{n}+\log n\log p$\\
Transitive closure & $s^{+}\log n$ & $(s^{+})^{3/2}\log n$ \\%& $\frac1p (m^{+})^{3/2}\log n+\frac{(m^{+})^2}{n}+\log^2 n\log p$\\
\bottomrule
\end{tabular}
\end{table}

\subsection{Reported results for the algebra}
\label{sec:algebra-results}

%The algebra was designed for answering two-way regular path queries (2RPQs) on graph databases, which translate into sums, products, transposition and closures of the matrices of the edge labels (following \citet{losemann2013complexity}): on a Wikidata graph with $959$M edges and $5\,419$ predicates, the $k^2$-tree version
%takes $4.33$ bytes per triple, a quarter of the space of the CSR/CSC baseline and $14$--$22\times$ less than classical RDF engines, and answers most queries in a few seconds, on average $2.6\times$ slower than the baseline. We focus here, however, on the results of the algebra on matrices themselves,
%which are the ones that matter for linear algebra.
\citet{arroyuelo2025evaluating} show that $k^2$-trees use $4.33$ bytes per cell on very sparse matrices from Wikidata graph databases. They also test various algebra operations on $80$ synthetic $1\,000\times1\,000$ matrices with densities from $10^{-4}$ to $10^{-1}$, against the uncompressed baseline and against GNU Octave (an optimized software for sparse matrices, as a sanity check). \Cref{tab:alg-exp} summarizes the space and time ratios they report against the baseline, which in turn is only slightly slower than Octave.
\begin{table}[t]
\centering\small
\caption{Time ratios reported by \citet{arroyuelo2025evaluating} on synthetic Boolean matrices.}
% (as stated in the text of their experimental section; we did not re-derive them from the plots).}
\label{tab:alg-exp}
\begin{tabular}{@{}p{2.6cm}p{10.0cm}@{}}
\toprule
Operation & Compared to $k^2$-trees, the baseline is...\\
\midrule
Space & $4\times$ larger \\
Sum   & $3$--$5\times$ faster \\
Intersection & $4$--$9\times$ faster \\
Multiplication & $14$--$53\times$ faster (the parallel $k^2$-tree reaches  $10\times$ speedup on dense products) \\
Transitive closure & $4$--$189\times$ faster (except at density $10^{-4}$, where the $k^2$-tree beats the baseline; parallelism gives about $3\times$ speedup)\\
\bottomrule
\end{tabular}
\end{table}
The picture is that of a space--time trade-off: the $k^2$-tree is four times smaller and offers a complete compositional algebra, but a plain sparse representation with Schoor's algorithm is faster, as can be expected from the complexities of \cref{tab:alg-complexity}. This is in line with the results shown in \cref{fig:bic-mult}.

\subsection{Kache-trees: cache-friendly $k^2$-trees}
\label{sec:kache}

The level-wise layout of the $k^2$-tree makes any subtree, and in particular a subtree that must be copied to the output, scattered across levels. \citet{farina2025cache} propose
a depth-first deployment, the \emph{kache-tree}, so that every subtree is contiguous. The classical DFUDS representation of a cardinal tree~\cite{benoit2005representing,navarro2016compact}
uses $k^2+2$ bits {\em per node}, with a redundant sequence of parentheses apart from the signatures. They observe instead that the DFUDS description of a node is a function of its signature and of its depth, so the parentheses sequence
need not be stored: the signatures of the internal nodes in depth-first order, extended with one bit that marks the nodes in the last internal level,
are enough, for $k^2+1$ bits {\em per internal node}. 
%Navigation uses $\mathrm{fwd\_search}$ over an rmM-tree~\cite{navarro2014fully} built
%over blocks of signatures, and a $64$\,KB universal table that processes three $5$-bit signatures at once, so that the whole table fits in cache.
The sum becomes a single sequential pass, in which subtrees present in only one operand are copied as one contiguous chunk; the product is the same recursion. %, with partial results not equipped with rmM-trees until the end.
\begin{figure}[t]
\centering
\includegraphics[width=0.75\linewidth]{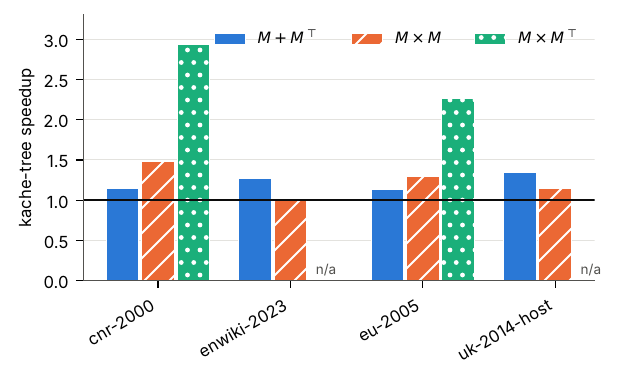}
\caption{Kache-trees against $k^2$-trees on real Web graphs, as speedups ($k^2$-tree time divided by kache-tree time, so values above $1$ mean that kache-trees are faster; data from \citet{farina2025cache}). The $k^2$-tree is smaller by a factor of $0.85$ in all cases.
The product $M\times M\T$ did not finish in two cases for lack of memory.}
\label{fig:kache}
\end{figure}
On synthetic matrices of side up to $16\,000$ and density $10^{-4}$ to $10^{-1}$ the kache-tree uses $12$--$18\%$ more space, is up to $60\%$ faster for sums (more so on sparser matrices, because more whole submatrices can be copied) and products of dense matrices; \cref{fig:kache} shows results on real graphs. The gain is attributed to cache misses: \texttt{perf} counts $31\%$--$42\%$ more misses in the $k^2$-tree for sums, and 
up to $225\%$ more for products. %The authors plan to implement the transitive closure, a fully cache-oblivious variant, and a Huffman coding of small blocks.

A second depth-first deployment \cite{carmona2025depthspire,carmona2025depth}, designed for other purposes, 
is discussed in \cref{sec:k2variants}.

\subsection{Qdags and prejoining}
\label{sec:qdags}

A different use of $k^2$-trees performs \emph{relational joins} on graph databases. \citet{arroyuelo2022optimal} 
store, for each distinct label $\ell$ in the graph, a matrix of the subgraph of the edges with label $\ell$, and
represent those matrices with $k^2$-trees. To perform the join, they introduce \emph{qdags}, which simulate the 
extension of each participating $k^2$-tree to a virtual $d$-ary relation (i.e., a $k^d$-tree), where $d$ is 
the number of attributes in the join. The join then boils down to a $k^d$-tree intersection (as seen above, where
now each node has $2^d$ quadrants) of all the involved qdags. This technique is shown to be worst-case-optimal (wco)
\cite{ngo2012worst,atserias2013size,veldhuizen2014triejoin} while using very small space (that of the
$k^2$-trees): $4.75$ bytes per edge in a Wikidata subgraph, whereas classical wco indexes use over $100$.
Their drawback is, as expected, a penalty factor $2^d$ in the intersection time.

\citet{arroyuelo2026boosting} introduce \emph{prejoining}: pick a subset of $d'<d$ attributes, project every involved matrix onto those, join the projections
(a lower-dimensional, hence much faster, join) and add the result as one more $k^{d'}$-tree in the main join. They prove that prejoining preserves worst-case optimality, and explore various prejoining strategies:
\emph{one-shot} (prejoin once, on all the joined attributes), \emph{leapfrog} (a sequence of prejoins on increasingly larger attribute sets, loosely simulating Leapfrog Triejoin \cite{veldhuizen2014triejoin}), \emph{kernel} (reduce every table by semijoin
with the one-shot prejoin) and \emph{weak kernel} (the same, attribute by attribute, using precomputed projections). Storing the two projections of each relation
as $1$-dimensional qdags increases the space by about $30\%$ (from $4.75$ to $6.17$ bytes per edge), and the speedups over the original algorithm reach three orders of magnitude on some query shapes, making them, on real Wikidata queries, faster on average than the other indexes of similar size.

The key idea that makes prejoining relevant for this survey is that a matrix multiplication, $C = A \times B$, can 
be expressed with relational joins ($\Join$) and projections ($\pi$), 
\[ C(a,c)=\pi_{a,c}\bigl(A(a,b)\Join B(b,c)\bigr).\]
In particular, a one-shot prejoin computes $I(b)=\pi_b(A)\Join\pi_b(B)$, which is
the intersection of the nonempty columns of $A$ with the nonempty rows of $B$, and then every $b\in I$ generates the Cartesian product of the column $b$ of $A$ and the row $b$ of $B$.
This is \emph{exactly Schoor's algorithm}~\cite{schoor1982fast} (\cref{sec:plain}), which had only been implemented on uncompressed
matrices~\cite{arroyuelo2025evaluating}, and was well known for its efficiency on sparse matrices; the new technique runs it directly on $k^2$-trees.

\Cref{fig:prejoin} shows the average time of the product of $2$ to $4$ matrices (P2--P4) on the Wikidata Graph Pattern
Benchmark subgraph ($81$M edges), compared with the recursive algorithm of \cref{sec:algebra}, and for a single product on the full graph ($959$M edges).
\begin{figure}[t]
\centering
\includegraphics[width=0.8\linewidth]{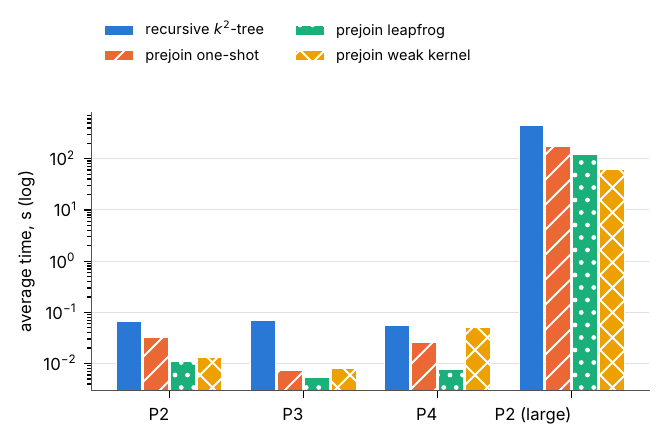}
\caption{Average time (seconds, log scale) of chains of $2$--$4$ matrix products (P2, P3, P4) using prejoin strategies on qdags, against the recursive $k^2$-tree algorithm, on the Wikidata Graph Pattern Benchmark subgraph; and a single product on the full Wikidata graph (P2 large). Data from Table~5 of \citet{arroyuelo2026boosting}.}
\label{fig:prejoin}
\end{figure}
Speedups are about $2$--$6\times$ for single products, $9$--$13\times$ for chains of three, up to $7\times$ for chains of four, and $2.5$--$7\times$ on the full graph; leapfrog is best on chains, weak kernel on the large graph.
To remain compositional, the product must also project the result on each of the two attributes. 
%The paper leaves the application to the transitive closure as future work.

\subsection{Variants that exploit repetitions}
\label{sec:k2variants}

Two other approaches keep the underlying $k^2$-tree structure but exploit repeated submatrices, thereby providing a link to the upcoming sections (most closely, to the grammars of \cref{sec:grammar}).

\paragraph{Two-dimensional block trees (2D-BTs).} \citet{brisaboa2022two} extend the block trees of one-dimensional sequences to matrices. As in a $k^2$-tree, the matrix is recursively split into $k^2$ submatrices, but a block that is equal to an earlier occurrence of the same submatrix (possibly overlapping up to four adjacent blocks, so not necessarily aligned with the tree) is not subdivided: it becomes a leaf with a pointer to its first occurrence, as in Lempel--Ziv compression \cite{ziv1977universal}. This captures repetition, whereas a plain $k^2$-tree captures only sparseness and clustering. On Web graphs, the 2D-BT uses $60\%$--$80\%$ of the space of the $k^2$-tree, with cell access between $30\%$ faster and three times slower; construction is the main drawback (on \texttt{eu-2005}, $7{,}200$ seconds against $67$ for the $k^2$-tree). The paper supports access to cells and to regions, but does not address other operations. While it seems feasible to adapt \cref{alg:k2vec} to this format, it is not clear that one can retain $\Oh(a\log n)$ time due to the possibility of traversing paths that yield no points. Recursive algorithms for matrix--matrix products and point-wise operations can also be adapted, but again with unclear upper bounds.

\paragraph{Depth-first $k^2$-trees.} \citet{carmona2025depthspire,carmona2025depth} observe that the usual level-wise layout of the $k^2$-tree makes it hard to detect and share identical subtrees, and propose instead depth-first layouts: a plain one, one enriched with skip values for large subtrees, and one based on balanced parentheses in which identical subtrees are shared using pointers. On Web graphs, removing the identical subtrees improves compression (for example, to 73\% on their graph ID-2004). Further, some depth-first variants multiply Boolean matrices faster than the standard $k^2$-tree using the recursive algorithm of \citet{arroyuelo2025evaluating}, although not on all graphs (and the variant with the best compression, the one sharing submatrices, is the slowest). Their experiments concern matrix--matrix products, not products by a vector, which are mentioned only as a motivation for the sequential scans that the depth-first layout allows.
%\claude{The claim that these are ``still slower than CSR/CSC'' is not supported by Carmona--Manzini; check the paper.}

\subsection{Discussion}

The $k^2$-tree line shows that a compressed sparse matrix can support a \emph{general} algebra without decompression, with each operation costing roughly its uncompressed counterpart times a $\log n$ factor (\cref{tab:alg-complexity}), and the
trade-off between space and speed is explicit: $4\times$ less space for several times more time, or up to $60\%$ gain from a better memory layout for $18\%$ more space. It also shows the importance of the algorithm: the move
from recursive multiplication to Schoor-like prejoining yields larger gains than most layout improvements.
While some variants manage to exploit structure beyond sparsity and clustering, they do not do so in a computation-friendly form: they are still slower than a plain sparse format like CSR/CSC.

\section{Referencing and Set-Difference Formats}
\label{sec:reference}

\subsection{Origins}

The idea of expressing a row in terms of a similar one dates back at least to \citet{bookstein1991compression}, who proposed compressing a collection of correlated bit-vectors by
representing each as a set and storing, for a ``reference'' set, only the symmetric difference. Consider the
complete graph whose nodes are the sets, with edges weighted by the size of the symmetric difference of their
endpoints. A spanning forest of rooted trees that minimizes the cardinalities of the roots (stored explicitly)
plus the weights of the edges gives the best way of expressing the collection this
way. The idea was reinvented several times for tables and graphs~\cite{buchsbaum2000engineering,adler2001towards,bharat1998connectivity}.
\citet{bjorklund2001fast} seem to have been the first to observe that the weight of such a
forest also bounds the time to multiply matrices: if the product of a row $u$ of $A$ by a column $v$ of $B$ is known,
then the product of another row $u'$ by $v$ is obtained by multiplying only the entries in the symmetric difference
of $u$ and $u'$ with the corresponding entries of $v$ and using that to adjust the result.

On Web graphs, rows with adjacent URLs tend to be similar, and Boldi and Vigna's WebGraph
framework~\cite{boldi2004webgraph,boldi2005codes} exploits this \emph{copy property}: each adjacency list is encoded
with respect to a reference among the previous $W$ lists (the \emph{window}), by marking which entries of the
reference not to copy and listing the extra entries. The result is further compressed with universal integer
codes. Limiting the candidate references to a window keeps the cost of extracting a list cache-efficient and the
construction fast, and a good vertex ordering (lexicographic order of URLs, or label propagation
orderings~\cite{boldi2011layered,boldi2009permuting}) makes the lists in the window similar. This family of ideas has
also been optimized and extended by others~\cite{grabowski2011merging,grabowski2014tight}.
\citet{francisco2022graph,francisco2018exploiting} observed that
this representation also accelerates the product by a vector, which is the
starting point of this section (an earlier similar result
\cite{nishino2014accelerating} limited the technique to detecting full
suffixes shared among rows, instead of any subset of the rows).

\paragraph{Supporting access.} While accessing an arbitrary cell $A_{ij}$ is not particularly easy or practical in
this format, \citet{gagie2026set,gagie2026compressed} show how this could be done in $\Oh(\log n)$ time with sophisticated data structures that use essentially the space to represent all the symmetric differences, or $\Oh(\log\log n)$ time if using a constant number of times that space. They also support more complex operations.

\subsection{Matrix--vector products on a referencing format}

Let $A\in\mathbb{K}^{n\times m}$ over a ring $(\mathbb{K},+,\cdot)$, with rows $\Vv_1,\dots,\Vv_n$, and let 
$\mathbf{r}=(r_1,\dots,r_n)$ be a \emph{referencing vector}:
$r_i<i$ and $\Vv_{r_i}$ is the row used to encode $\Vv_i$, with $\Vv_0=\mathbf{0}$ and $r_i=0$ when no reference is used.
Define the {\em residual} matrix $A'$ whose $i$-th row is $\Vv_i-\Vv_{r_i}$ (we can use the technique on the 
Boolean semiring by extending the symbols to $\{-1,0,+1\}$; the residual values will stay inside this
set). Note that $\nnz(A')$ is the number of entries added plus the number of entries removed when
moving from the reference to the row.

\begin{proposition}[\cite{francisco2022graph}]\label{prop:ref}
Let $\ww$ be the vector with $w_i=\Vv_{r_i}\cdot\xx$. Then $A\xx=A'\xx+\ww$.
\end{proposition}
\begin{proof}
$w_i=y_{r_i}$, where $\yy=A\xx$, and $y_i=\Vv_i\cdot\xx=(\Vv_i-\Vv_{r_i})\cdot\xx+\Vv_{r_i}\cdot\xx$.
\end{proof}
Because $r_i<i$, the entry $w_i=y_{r_i}$ is already known when $y_i$ is computed, which yields \cref{alg:ref}. The number
of operations is proportional to the number of nonzeros of $A'$, that is, to the size of the compressed
representation (plus $n$ for the additions of the references).

\begin{algorithm}[t]
\caption{Right product $\yy=A\xx$ from a referencing representation $(A',\mathbf{r})$~\cite{francisco2022graph}.}
\label{alg:ref}
\begin{algorithmic}[1]
\State $y_0\gets 0$
\For{$i=1,\dots,n$}
  \State $y_i\gets y_{r_i}+\sum_{j:\,A'_{ij}\neq 0}A'_{ij}\,x_j$
\EndFor
\State \Return $\yy$
\end{algorithmic}
\end{algorithm}

\paragraph{Left products (our derivation).}
The paper handles the left products that PageRank needs by storing the \emph{transposed} matrix with its own
reference vector (\cref{eq:pagerank} uses left products, whereas the representation is row-oriented). The
referencing structure also supports the left product $\xx\T=\yy\T A$ directly, by duality
(an instance of the general principle presented later, in \cref{sec:factorization}). Unrolling the references gives $\Vv_i=\sum_{k\in\mathrm{path}(i)}A'_{k\ast}$, where
$\mathrm{path}(i)$ contains $i,r_i,r_{r_i},\dots$ until the root. Therefore, $A'_{k\ast}$ contributes to all
the rows in $\mathrm{subtree}(k)$, the subtree of $k$ in the reference forest. This allows, if we precompute
$z_k=\sum_{i\in\mathrm{subtree}(k)}y_i$, to compute $\yy\T A=\sum_k z_k A'_{k\ast}$.
The values $z_k$ are precomputed in one pass over $k=n,n-1,\dots,1$, setting
first $z_k\gets y_k$ and then adding $z_i$ to $z_{r_i}$ when processing $i$; this is possible because $r_i<i$.
The cost is again $\Oh(\nnz(A')+n)$.

\paragraph{Vertical slicing.}
The quality of the method depends on how good the references are. For wide matrices it is hard to find a previous
row that matches in all columns, but easier to find one that matches in a narrow band of columns.
Francisco et al.\ therefore split $A$ vertically into $\tau$ submatrices $A_1,\dots,A_\tau$ of $\Theta(m/\tau)$
columns each, compress each with its own referencing vector, and compute $\yy=\sum_{t=1}^{\tau}A_t\xx_t$ where
$\xx_t$ is the corresponding slice of $\xx$. Setting $\tau=1$ gives the original method. %, and $\tau=m$ the schoolbook product.

\subsection{Differential compressed binary trees}

\citet{nelson2017queryable,nelson2019billion} (see also the dissertation~\cite{nelson2019thesis}) use the same idea of encoding each row as a difference with a previous row in a window, but encode the rows differently, with trees instead of integer codes.
Each row of the Boolean matrix is a \emph{compressed binary tree} (CBT): a binary tree of depth $\lceil\log_2 n\rceil$ over the column range, in which every node stands for a range of columns and is stored with a single bit, which is $1$ if the range contains at least a nonzero (and, at the last level, if the cell is a nonzero) and $0$ otherwise. Empty ranges are thus pruned, and the path to each nonzero descends to the last level. The binary tree is deployed in preorder, so it uses at most $2a_i(\log_2(n/a_i)+\Oh(1))$ bits, where $a_i=\nnz(A_{i\ast})$, and allows obtaining the 1s in the row with a full traversal. The dissertation also improves the encoding of branches that are full of 1s. 
%Mas bien confunde porque la cosa es con las referencias y ahi igual hay que recorrer todo: Random access is not supported. We note that this is similar to a $k^2$-tree (\cref{sec:qdags}) in one dimension and with $k=2$, which would use the same asymptotic space and would provide random access to any $A_{i\ast}$.

In the \emph{differential} variant, row $i$ is stored as the CBT of its differences with a row $r_i$ in the window (so the 1s mark the cells where the rows differ), plus one bit in each leaf that tells whether the change inserts or deletes the cell, exactly as the residual $A'$ with entries $+1$ and $-1$ of this section. A product $A\xx$ then reuses $A_{r_i}\xx$ in the form $A_i\xx=A_{r_i}\xx+(A_i-A_{r_i})\xx$, and its time is proportional to the size of the compressed differences, as in \cref{alg:ref}; the encoding of $A'$ is what changes.

\subsection{Matrix--matrix products}
\label{sec:refmatmat}

\citet{francisco2022graph} only sketch matrix--matrix products. If $A$ has rows $\Vv_i$ and reference vector $\mathbf{r}$, and $B\T$ has
rows $\ww_j$ and reference vector $\mathbf{c}$, then for $C=A\times B$,
\begin{equation}\label{eq:matmat}
C_{ij}=C_{i,c_j}+C_{r_i,j}-C_{r_i,c_j}+(\Vv_i-\Vv_{r_i})\cdot(\ww_j-\ww_{c_j}),
\end{equation}
so that, when $C$ is computed in an order in which $C_{i,c_j}$, $C_{r_i,j}$ and $C_{r_i,c_j}$ are already known,
only the dot product of two \emph{residual} vectors remains. Those vectors are sparse, and choosing to iterate
over the support of the sparser one (or intersecting the supports adaptively \cite{barbay2008alternation}) may give an additional speedup.
This was left for future work, and, to our knowledge, \cref{eq:matmat} has not been evaluated.
Note that it still computes all $n^2$ entries of $C$, so it is useful only when the output is dense; for
sparse outputs the techniques of \cref{sec:quadtree} are more natural.

\citet{alves2025accelerating} independently define the Compressed Binary Matrix (CBM) format, which is the referencing format of this section applied to the product $C=A\times B$ of a binary adjacency matrix $A$ by a dense matrix of embeddings $B$, the core operation of training and inference of GNNs. Beyond an obvious column-wise computation of $C$, the product is done in two stages: an off-the-shelf sparse matrix multiplication computes $A'B$, and then the rows of the result are updated in topological order along the reference forest of $A'$, $C_{i\ast}\gets C_{i\ast}+C_{r_i\ast}$.
We derive a more general version of this idea later, in \cref{sec:factprod}.

\citet{nelson2019thesis} considers Boolean matrix--matrix products $C=A\times B$ when both are represented with references, concretely, the rows of $A$ and the \emph{columns} of $B$ are CBTs. An entry $C_{ij}$ is the Boolean inner product of the row $i$ of $A$ and the column $j$ of $B$, computed by traversing their two trees in preorder at the same time: the branches that are $0$ in either tree are skipped in both, and the computation stops with a 1 as soon as two leaves with a $1$ coincide. Looping over the columns of $B$ and then over the rows of $A$ builds $C$ column by column, directly as CBTs, so $C$ is already in the format required by the \emph{right} operand of the next product: products can then be composed in that direction, and the output has to be transposed to be used as the left operand. % a mi no :-) This resembles the recursive multiplication of the $k^2$-trees of \cref{sec:algebra}, in one dimension. 
Nelson describes this product for the plain encoding and does not report experiments on it for static matrices.

\subsection{Experimental results}

\Cref{fig:ref} (left) summarizes the evaluation of \citet{francisco2022graph} on Web and social graphs from the LAW
collection~\cite{boldi2004webgraph,boldi2011layered,boldi2018bubing}, using ten PageRank iterations. The residual matrix $A'$ was extracted from the
WebGraph representation, recompressed as a WebGraph, and the computation was performed on it (a positive and a
negative graph for the $+1$ and $-1$ entries).

\begin{figure}[t]
\centering
\includegraphics[width=\linewidth]{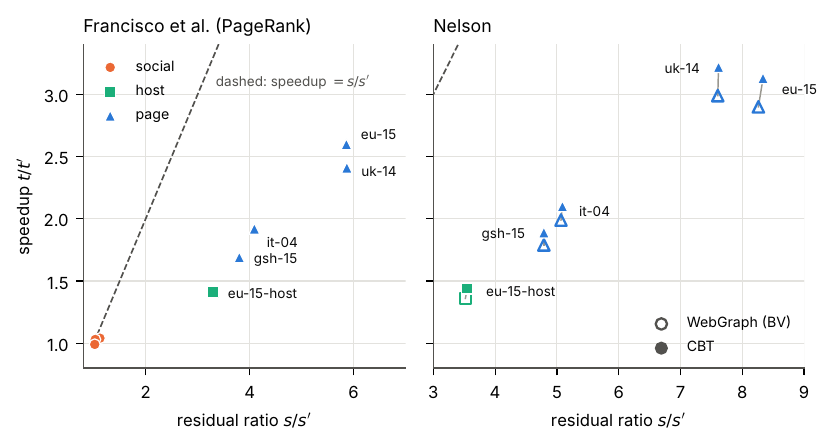}
\caption{Speedup $t/t'$ of a product by a vector with the residual matrix $A'$
over the product with $A$, as a function of the reduction in nonzeros $s/s'$
($s=\nnz(A)$, $s'=\nnz(A')$). Points on the dashed line would achieve a
speedup equal to the reduction. \emph{Left}: PageRank, window $W=64$ (except
$W=16$ for \texttt{gsh-2015}), data from \citet{francisco2022graph}.
\emph{Right}: data from \citet{nelson2019thesis}, with the WebGraph (BV) scheme (hollow marks) and compressed binary trees (filled marks) on the same graphs. The two panels come from different codes and machines, so only the comparison within a panel is meaningful, not across panels.}
\label{fig:ref}
\end{figure}

The results follow the compressibility of the graph. On page graphs, $A'$ has $17\%$--$26\%$ of the nonzeros of
$A$ and the speedups are between $1.7\times$ and $2.6\times$; on host graphs it has about $30\%$ and the
speedup is about $1.4\times$; social networks, which are known to be less compressible in this
sense~\cite{chierichetti2009compressing}, show no gain. The speedup is systematically below $s/s'$: profiling
showed that the nonzeros of $A'$ are more scattered in memory than those of $A$, so the plain product benefits from
contiguous accesses. Vertical slicing gave only a small further reduction of $\nnz$ (a few percent for
$32$ to $128$ slices), whereas the total size of the compressed submatrices grew considerably with $\tau$.
The practical bottleneck is the preprocessing: recompressing the graphs took hours for the largest instances.

The right of \cref{fig:ref} shows the results of \citet{nelson2019thesis}
on five Web graphs of up to $9.2\times10^{10}$ edges, compared with the
implementation of \citet{francisco2022graph}. The CBT versions are slightly
smaller ($2\%$--$9\%$), the number of nonzeros of the residual is slightly
smaller too, and the time with differences is between $17\%$ and $26\%$ lower;
the speedup over the product without differences (that is, from reusing
results) goes from $1.44\times$ to $3.22\times$, against $1.36\times$ to
$2.99\times$ with WebGraph. The author attributes the gain to a faster
decoding and to a better compression of the differences, which allows more
reuse. Once again, the speedup follows the reduction in nonzeros.

\subsection{Finding the best references}
\label{sec:setdiff}

The quality of a referencing format depends entirely on how the references are
chosen. Let us interpret rows $i$ as sets $S_i = \{ j, A_{ij} \neq 0 \}$, so 
the matrix $A$ is seen as a set of sets 
$\mathcal{S}=\{S_1,\dots,S_n\}$ over a universe $[1..m]$ (so
$\nnz(A)=\sum_i|S_i|$), and let $S_0=\emptyset$ be an additional empty set (whose use as a reference means that the row is stored explicitly).
Consider the complete graph on $\mathcal{S}\cup\{\emptyset\}$ where the edge between $S$ and $S'$ has weight $|S\triangle S'|$, the size of their symmetric difference.
A referencing vector \textbf{r} is exactly a spanning tree of this graph,
rooted at $\emptyset$ (each node $S_i$ points to its parent, $S_{r_i}$), the number of nonzeros $\nnz(A')$ of the residual matrix $A'$ is the weight of the tree, and the
time of the product of \cref{alg:ref} is $\Oh(\nnz(A')+n)$. (The condition $r_i<i$ used in \cref{alg:ref} is only a convenient way of stating that rows are processed in an order
in which the reference precedes the row, for instance a breadth-first order of the tree, which can always be obtained by renumbering the rows.) The best referencing vector is therefore a
\emph{minimum spanning tree} (MST) of this graph (the tree is the set of references $r_i$; we return to this view in \cref{sec:factorization}). See \cref{fig:mst} for an example.

\begin{figure}[t]
\centering
\begin{tikzpicture}[
  row/.style={circle,draw=black,fill=black!10,thick,minimum size=7.5mm,inner sep=0pt,font=\small},
  w/.style={midway,fill=white,inner sep=1.2pt,font=\scriptsize},
  mst/.style={very thick,draw=mstc},
  oth/.style={thin,dashed,draw=gray!70},
  setl/.style={font=\scriptsize,orange!80!black,align=center}]
\definecolor{mstc}{HTML}{2A78D6}
\node[row] (0) at (0.5,2.1) {$\emptyset$};
\node[row] (1) at (-4.5,0.2) {$1$};
\node[row] (2) at (-4.5,-1.5) {$2$};
\node[row] (3) at (-4.5,-3.2) {$3$};
\node[row] (4) at (-0.5,0.2) {$4$};
\node[row] (5) at (2.0,0.2) {$5$};
\node[row] (6) at (4.5,0.2) {$6$};
\node[setl,left=1pt of 1]  {$\{1..4\}$};
\node[setl,left=1pt of 2]  {$\{1..5\}$};
\node[setl,left=1pt of 3]  {$\{1..5,8\}$};
\node[setl,below=1pt of 4] {$\{3,4,7,8\}$};
\node[setl,below=1pt of 5] {$\{1,2,7,8\}$};
\node[setl,below=1pt of 6] {$\{5..8\}$};
\draw[oth] (1) to[bend left=40] node[w] {2} (3);
\draw[oth] (4) -- (5) node[w] {4};
\draw[oth] (5) -- (6) node[w] {4};
\draw[mst] (0) -- (1) node[w] {4};
\draw[mst] (0) -- (4) node[w] {4};
\draw[mst] (0) -- (5) node[w] {4};
\draw[mst] (0) -- (6) node[w] {4};
\draw[mst] (1) -- (2) node[w] {1};
\draw[mst] (2) -- (3) node[w] {1};
\end{tikzpicture}
\caption{Symmetric-difference graph of the rows of the running example,
with the extra node $\emptyset$. Edge weights are $|S\triangle S'|$; only some
edges are drawn. The solid blue edges form an MST rooted at $\emptyset$, of
weight $4+1+1+4+4+4=18=\mathrm{nnz}(A')$; rows $4$--$6$ are at distance $4$
from $\emptyset$ and from each other, so a reference among them (for example
$r_5=4$) gives the same $\mathrm{nnz}(A')$ but adds a pointer, and the edge
$1\!-\!3$ (dashed) of weight $2$ is not used because $3$ is at distance $1$
from $2$. Adding the $2$ reference pointers $r_2=1$ and $r_3=2$ gives the size $18 + 2 = 20$ of \cref{tab:example}.}
\label{fig:mst}
\end{figure}
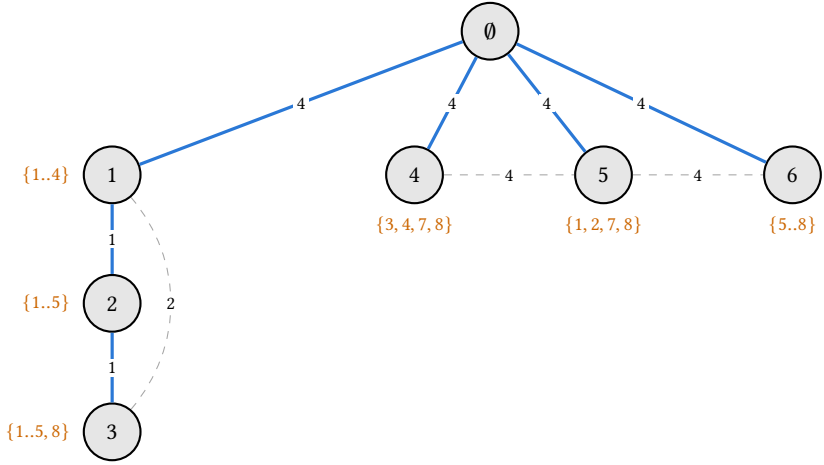

This reduction goes back to \citet{bookstein1991compression}, was used by \citet{bjorklund2001fast} to speed up Boolean matrix products, and has been rediscovered several
times~\cite{buchsbaum2000engineering,adler2001towards,bharat1998connectivity}; it is also the basis of the Compressed Binary Matrix format that \citet{alves2025accelerating} use to accelerate GNNs.
\citet{gagie2026set} refine it by also allowing a second root $[1..m]$ (the universe), so that a row can be built by deleting elements from $[1..m]$ instead of inserting elements into $\emptyset$,
which benefits dense rows; the weight of the MST of the resulting graph is called the \emph{symdiff compressibility} $\Delta(\mathcal{S})$.
They also study the simpler \emph{insertion compressibility} $I(\mathcal{S})\ge\Delta(\mathcal{S})$, where
each set is described by the elements it adds to its largest subset in the collection (the dual of the containment entropy of \citet{alanko2025compact}). 
The products then cost $\Oh(\Delta(\mathcal{S})+n)$.

\paragraph{The cost of finding it.}
The MST is expensive to obtain. Just computing the weights of the $\binom{n}{2}$ edges takes $\Oh(ns) \subseteq \Oh(n^3)$ time, where $s=\nnz(A)$ \cite{gagie2026set}, and Prim's algorithm on the complete graph
$\Oh(n^2)$ more; this is also close to the $\Oh(ns+n^2\log n)$ cost given by \citet{alves2025accelerating}. Such a cost is infeasible for Web graphs, where $n$ is in the hundreds of millions or billions. 
Practical alternatives give up optimality or exploit special structure:
\begin{itemize}
\item \emph{Windows.} The WebGraph framework~\cite{boldi2004webgraph} restricts the candidates of row $i$ to the previous $W$ rows, which reduces construction time to $\Oh(Ws)$, and picks the best one greedily (this is a heuristic: the result is a forest, not an MST, of the
  restricted graph). It relies on an ordering that puts similar rows close to each other (URL order, label propagation~\cite{boldi2011layered}).
\item \emph{Approximate MSTs.} \citet{bjorklund2001fast} already used approximations of the MST, since the exact one is too costly.
Likewise, to avoid computing all the pairwise distances one could look for
similar rows with sketching and locality-sensitive
hashing~\cite{broder1998min,harpeled2012approximate} or approximate Hamming
distances~\cite{kowaluk2026approximate}; MinHash is already the first step of
biclique extraction (\cref{sec:biclique}).
\item \emph{An exact MST without all the weights.} \citet{gagie2026set} observe that in compressible matrices the edges of the MST are light. They build a suffix tree of the
  concatenation of the sorted sets, which gives in constant time the longest common prefix of any two sets (and thus allows enumerating the elements of $S\triangle S'$ one by one), and run a variant of Prim's algorithm that
  discovers weights incrementally, from smallest to largest. The time is $\Oh(s\log m+\sum_{S,S'}\min(\ell,|S\triangle S'|))$, where $\ell$ is the largest weight in the MST. This
  improves on the quadratic bounds whenever $\ell$ is small, that is, exactly when the matrix is compressible. 
\end{itemize}
Even so, the preprocessing is the practical bottleneck: the recompression of the largest graphs by \citet{francisco2022graph} took hours (this can be more manageable on smaller graphs: \citet{alves2025accelerating} reports 16 seconds on a 32-million-edge graph).

\paragraph{Referencing for GNNs.}
The technique of \citet{alves2025accelerating} to compute $C= A \times B$ with dense $B$ differs from the plain MST scheme in that a threshold $\alpha$ discards the references that save fewer than $\alpha$ operations, since traversing the tree has a cost; the graph becomes directed (a minimum cost arborescence replaces the MST) and the root gets a larger out-degree, which exposes more parallelism, at the price of losing useful references if $\alpha$ is too large.
%Diagonal scalings such as the normalization $D^{-1/2}(A+I)D^{-1/2}$ of GCNs are handled by scaling the columns of $A'$ and fusing the row scaling into the update.
They report speedups close to $5\times$ over the sparse matrix multiplication
of Intel MKL, almost $3\times$ for the inference of a two-layer GCN on the best graphs (co-paper graphs of average degree $57$--$75$), and no gain on graphs of average degree below $6$ (Cora, PubMed).

\no{
\paragraph{Applications in bioinformatics.}
The sets of colors (the genomes that contain a $k$-mer) of colored neighboring
vertices in a de Bruijn graph are likely to be similar. \citet{almodaresi2020efficient} took advantage of this in the tool Mantis~\cite{pandey2018mantis}, by considering only edges between the color sets of neighboring vertices and encoding the color sets on the
  minimum spanning tree of the induced graph. \citet{campanelli2024where} went further by considering unitigs (paths of vertices with in- and out-degree $1$, which become single edges in the compacted graph) and compressing their color sets, and \citet{levallois2026kaminari} give an approximate solution using minimizers~\cite{schleimer2003winnowing,roberts2004reducing}
  and integer compression. \citet{alanko2025compact} proposed to choose, for each set $S$ of the collection, the smallest superset $S'$ of it in the collection as its parent in a tree (close to a deletion-only version of the setting of \cref{sec:setdiff}), to path-compress the tree so that its height is logarithmic in the size of the universe, and to encode $S$ as a bitvector that
  says which elements of $S'$ belong to $S$, which supports predecessor and successor queries on the sets in logarithmic time. On the tool Themisto~\cite{alanko2023themisto} with $16$ thousand bacterial genomes, it improved the compression of the default ($0.18$ against $0.32$ bits per element), although query times were not reported.
}

\paragraph{Beyond spanning trees.} 
The MST uses only the given rows as nodes. \citet{gagie2026compressed} explore
the possibility of adding additional graph nodes that are good references. For
example, adding an extra row $\{7,8\}$ to our running example allows rows 4,
5 and 6 to refer to it at cost $2$, thereby reducing $\Delta(\mathcal{S})$ to $14$
instead of $18$. Finding the minimum-weight MST with added rows then becomes
an {\em undirected Steiner tree} problem (the auxiliary sets are exactly the 
auxiliary nodes of the factorization of \cref{sec:factorization}).
The Steiner problem is NP-hard and APX-hard in
general graphs~\cite{chlebik2008steiner}, but  $\approx 1.4$-approximation
algorithms exist~\cite{byrka2013steiner}. Still, the MST on the original nodes 
is already a $2$-approximation in any metric, so auxiliary nodes could reduce
$\Delta(\mathcal{S})$ to a half at best.
Similar ideas could be applied for the grammar and biclique techniques of the
next sections. 

Other extensions explored \cite{gagie2026compressed} include operations like
taking the union or intersections of other rows. The optimization problem now
becomes the {\em multitarget circuit complexity}, which is hard to approximate
below any relevant ratio \cite{ilango2020nphardness}.

\subsection{Discussion}

Referencing formats are simple, effective on data with strong row-to-row similarity
(web graphs, co-authoring/co-citation graphs),
and support the right and left products with extremely simple algorithms. Their weaknesses are
(i)~reliance on a good row ordering or an expensive MST for construction, (ii)~the fact that the realized speedup is
bounded by memory locality rather than by the number of operations, and (iii)~the dependence on \emph{whole-row}
similarities, which makes them blind to repeated pairs of values in two columns that are
separated by columns whose contents vary. The grammar approach of the next section addresses the last point.

\section{Grammar-Compressed Matrices}
\label{sec:grammar}

Referencing formats exploit similarity between \emph{entire rows}. Many matrices, in particular those arising in machine learning,
instead contain \emph{partial} repetitions: the same combination of values appears in the same columns in many rows,
even if the rows differ elsewhere. This was observed in the Compressed Linear Algebra (CLA) system of
\citet{elgohary2018compressed,elgohary2019compressed}, which groups correlated columns and compresses each group with
lightweight schemes (offset lists, run-length encoding, dictionary coding) that support matrix--vector products directly.
\citet{ferragina2022improving} propose a more general mechanism: compress the matrix, once linearized,
with a \emph{grammar compressor}, and show that products can still be computed in time linear in the size of the grammar.

\subsection{The CSRV representation}

Their starting point is a variant of CSR that suits matrices with few distinct values. Let $V[1..k]$ be the table of
distinct nonzero values of $A\in\mathbb{S}^{n\times m}$ over a semiring
$(\mathbb{S},+,\cdot)$. The \emph{Compressed Sparse Row/Value} (CSRV) representation is a
sequence $S$ obtained by scanning the matrix row by row: for every nonzero
$A_{ij}$ append the pair $\langle \ell,j\rangle$ with
$V[\ell]=A_{ij}$, and at the end of each row append a separator $\$$. The length of $S$ is then $\nnz(A)+n$. The key
property is that the pair stores the \emph{column} together with the value index, so that two entries are the same
symbol only if they have the same value in the same column. The right product
$\yy=A\xx$ is done by a single scan
($y_i\mathrel{+}=V[\ell]\,x_j$ for each pair $\langle\ell,j\rangle$ in row
$i$) and the left product $\xx\T = \yy\T A$ with another
($x_j\mathrel{+}=y_i\,V[\ell]$). 
For example, rows $(1.2,\,3.4,\,5.6,\,0,\,2.3)$ and $(1.2,\,3.4,\,2.3,\,4.5,\,0)$ share
the prefix $\langle 1.2,1\rangle\langle 3.4,2\rangle$ (with the value indices of $V$), which a grammar can factor out (nonterminal $N_2$ in \cref{fig:grammar}, see next).
%Columns can be reordered so as to produce longer shared segments (\cref{sec:reorder}). 

\subsection{Grammar compression}

A grammar-compressed representation of a string $S$ is a context-free grammar that generates only $S$. For simplicity,
we assume a \emph{straight-line program} (SLP)~\cite{lohrey2012algorithmics}: a set of rules
$N_i\to XY$, where $X$ and $Y$ are terminals or earlier nonterminals, and the expansion of the last rule is $S$. The \emph{size} of the grammar is the
number of rules $q=|R|$ (each with two symbols on the right-hand side). Finding the smallest grammar is NP-hard~\cite{charikar2005smallest,storer1982data}, but
practical compressors such as RePair~\cite{larsson2000offline} obtain good grammars heuristically. Grammars can
compress very repetitive inputs, within a logarithmic factor of the best Lempel--Ziv-based
compressors~\cite{rytter2003application,jez2015approximation,navarro2021indexing}.
RePair repeatedly finds the most frequent pair of adjacent symbols $XY$, replaces all
its occurrences by a new nonterminal $N\to XY$, and stops when no pair appears twice. Note that its output is not exactly an SLP: it is a set of
rules $R$ plus a final sequence $C$ of terminals and nonterminals, whose expansion is $S$ (one can regard $C$ as a last rule of arbitrary length). The size of the representation is
then $2|R|+|C|$.

To compress CSRV, the sequence $S$ is fed to the grammar compressor, which is forbidden from using the separator
$\$$ in any rule. The result is a triple $(C,R,V)$ in which the expansion of each nonterminal is a sequence of
$\langle\ell,j\rangle$ pairs, and $C = C_1 \$ C_2 \$ \cdots \$ C_n$ is the
final sequence, still containing the row separators; note $|C| \ge n$.
See \Cref{fig:grammar} again for an example.

\begin{figure}[t]
\centering\small
\begin{tabular}{@{}l@{\hspace{1.2em}}l@{}}
$A=\begin{pmatrix}
1.2&3.4&5.6&0&2.3\\
2.3&0&2.3&4.5&1.7\\
1.2&3.4&2.3&4.5&0\\
3.4&0&5.6&0&2.3
\end{pmatrix}$
&
$\begin{array}{@{}l@{}}
V=[1.2,\,1.7,\,2.3,\,3.4,\,4.5,\,5.6]\\[2pt]
S=\underbrace{\langle1,1\rangle\langle4,2\rangle}_{N_2}\,\underbrace{\langle6,3\rangle\langle3,5\rangle}_{N_4}\,\$\;
\langle3,1\rangle\langle3,3\rangle\langle5,4\rangle\langle2,5\rangle\,\$\cdots\\[2pt]
R=\{N_1\to\langle3,3\rangle\langle5,4\rangle,\;N_2\to\langle1,1\rangle\langle4,2\rangle,\;N_4\to\langle6,3\rangle\langle3,5\rangle,\;\dots\}
\end{array}$
\end{tabular}
%\claude{This tabular overflows the text width by 12pt (overfull box); e.g.\ shrink the column gap or use \texttt{\textbackslash small}. I rewrote the taxonomy table of \cref{sec:taxonomy} (orange) because it overflowed by 17pt and cut ``predecessor in''.}
\caption{A small matrix with its CSRV sequence $S$ (first rows) and some of the rules RePair would create (adapted from~\cite{ferragina2022improving}).
The pair $\langle3,1\rangle$ stands for the value $V[3]=2.3$ in column $1$; the same value in column $3$ is the different symbol $\langle3,3\rangle$.}
\label{fig:grammar}
\end{figure}

\paragraph{Supporting access.} Because the symbols of any $C_i$ partition the nonzero columns into contiguous ranges, we
can binary search $C_i$ to find its only symbol that could contain the cell at column $j$, and then navigate towards
the right terminal $\langle \ell,j\rangle$ by the rules, to determine that
$A_{ij}=V[\ell]$ or zero. If the grammar is balanced (i.e., of height
$\Oh(\log n)$, as in many grammar constructions), this enables accessing $A$
in $\Oh(\log n)$ time. Note that this does not work if we reorder columns to
improve compression (\cref{sec:reorder}).
See also \citet{tabei2016scalable}, which predates \citet{ferragina2022improving} in the use of grammar-compressed
matrices but still decompress whole rows or columns in order to operate the matrices.
%\claude{Check that Tabei et al.\ really ``do not focus on matrix multiplications'': they do use the grammar for products in PLS regression.}
%\gonzalo{checked. they have operations like that but decompress to operate}

\subsection{Products on the grammar}

Let $\mathrm{eval}_{\xx}(\langle\ell,j\rangle)=V[\ell]\,x_j$ and, for a nonterminal $N$ with expansion
$\langle\ell_1,j_1\rangle\cdots\langle\ell_h,j_h\rangle$,
$\mathrm{eval}_{\xx}(N)=\sum_{k=1}^{h}V[\ell_k]\,x_{j_k}$.
Then \cite[Lemmas~3.2--3.3]{ferragina2022improving} $\mathrm{eval}_{\xx}(N)=\mathrm{eval}_{\xx}(X)+\mathrm{eval}_{\xx}(Y)$ for a
rule $N\to XY$, and $y_r$ equals the sum of $\mathrm{eval}_{\xx}$ over $C_r$.
Hence the values of all the rules are obtained in one pass, in creation order (each nonterminal depends only on earlier ones),
and then each row sums the values of its symbols:

\begin{theorem}[Right product~\cite{ferragina2022improving}]\label{thm:right}
Given $(C,R,V)$, the product $\yy=A\xx$ can be computed in $\Oh(|C|+|R|)$ time using $\Oh(|R|)$ words of auxiliary space.
\end{theorem}

The left product $\xx\T=\yy\T A$ is the transposed computation. Define $\mathrm{rows}(N)$ as the rows whose compressed
representation uses $N$, and $\mathrm{sum}_{\yy}(N)=\sum_{\ell\in\mathrm{rows}(N)}y_\ell$. Initializing an array $W$ so
that $W[i]$ holds the sum of $y_\ell$ over the rows $\ell$ such that $C_\ell$ contains $N_i$, and scanning the rules
\emph{backwards} from $N_q$ down to $N_1$, a rule $N_j\to XY$ transfers $W[j]$ to each of its symbols: to $W[i]$ if the
symbol is a nonterminal $N_i$ (with $i<j$), or to $x_k\mathrel{+}=V[h]\,W[j]$ if it is a terminal $\langle h,k\rangle$.
When a rule is reached, all the rules that mention it have already been processed (they have larger indices), so
$W[j]=\mathrm{sum}_{\yy}(N_j)$.

\begin{theorem}[Left product~\cite{ferragina2022improving}]\label{thm:left}
Given $(C,R,V)$, the product $\xx\T=\yy\T A$ can be computed in $\Oh(|C|+|R|)$ time using $\Oh(|R|)$ words of auxiliary space.
\end{theorem}

\Cref{alg:grammar} gives both procedures, which are an instance of the forward and backward substitution of the general scheme of
\cref{sec:factorization}.

\begin{algorithm}[t]
\caption{Right and left products on a grammar-compressed matrix $(C,R,V)$~\cite{ferragina2022improving}.
$\mathrm{val}(\alpha)$ is $V[\ell]\,x_j$ for a terminal $\langle\ell,j\rangle$ and $W[i]$ for a nonterminal $N_i$; $q$ is the number of rules $|R|$.}
\label{alg:grammar}
\begin{algorithmic}[1]
\Procedure{RightProduct}{$C,R,V,\xx$}
  \For{each nonterminal $N_k \to XY$, $k=1,\dots,q$} 
     \State $W[k]\gets \mathrm{val}(X)+\mathrm{val}(Y)$
  \EndFor
  \For{each row $r$}
     \State $y_r\gets\sum_{\alpha\in C_r}\mathrm{val}(\alpha)$
  \EndFor
\EndProcedure
\Procedure{LeftProduct}{$C,R,V,\yy$}
  \State $W[1..q]\gets 0$;\; $\xx\gets 0$
  \For{each row $r$ and each symbol $\alpha\in C_r$}
     \State add $W[i] \mathrel{+}= y_r$ if $\alpha=N_i$, or $x_j\mathrel{+}=V[\ell]\,y_r$ if $\alpha=\langle\ell,j\rangle$
  \EndFor
  \For{each nonterminal $N_k \to XY$, $k=q,\dots,1$} 
     \For{$\alpha\in\{X,Y\}$}
        \State add $W[i] \mathrel{+}= W[k]$ if $\alpha=N_i$, or $x_j\mathrel{+}=V[\ell]\,W[k]$ if $\alpha=\langle\ell,j\rangle$
     \EndFor
  \EndFor
\EndProcedure
\end{algorithmic}
\end{algorithm}

\no{ redundante:
\paragraph{Why this does not contradict the impossibility results.}
Recall from \cref{sec:limits} that grammar-compressed \emph{vectors} do not admit inner products in time linear
in the grammar size under standard assumptions~\cite{abboud2020impossibility}. The CSRV encoding avoids the obstacle: every symbol carries its own
column, so a nonterminal's value $\mathrm{eval}_{\xx}(N)$ depends only on the grammar and on $\xx$, and is
\emph{independent of where in the row the nonterminal occurs}. Without explicit positions, the same substring
occurring at different offsets would be applied to different slices of $\xx$, and no single value could be shared.
The price is that two rows share a nonterminal only if they agree on \emph{both} the value and the column.
}

\subsection{Implementation and variants}

The prototype~\cite{ferragina2022improving} builds the CSRV sequence (as $32$-bit integers: $0$ for $\$$,
$1+\ell m+j$ for $\langle\ell,j\rangle$ with $0\le j<m$) and runs a modified
RePair that does not pair separators. Three encodings of the
grammar trade time for space: \texttt{re\_32} (plain $32$-bit integers, the fastest), \texttt{re\_iv} (packed arrays with
$1+\lceil\lg \max\rceil$ bits per entry, $\max$ being the largest integer to encode) and \texttt{re\_ans} (packed rules, and $C$ compressed with
the asymmetric numeral systems entropy coder of \citet{moffat2020large}). The more compact
encodings reduce peak memory but not the number of operations, so they are slower. For multithreading the matrix is split
into $b$ blocks of rows that are compressed and multiplied independently; the left products are summed at the end, and the value table $V$ is
shared.

\subsection{Column reordering}
\label{sec:reorder}

RePair factors out pairs of \emph{adjacent} symbols that repeat; adjacency is crucial for it to run in linear
time (we will see in \cref{sec:biclique} the complications that arise when allowing non-adjacent symbols).
Therefore, placing correlated columns next to each other helps. Since the order in which the pairs are written
inside each row segment $C_r$ is free, any permutation of the columns can be applied without altering the 
multiplication algorithms and without storing the permutation, because each pair carries its original column
information. Finding the best order is NP-hard for table compression in general~\cite{buchsbaum2003improving}, and heuristics are standard in
that area~\cite{johnson2004compressing,vo2007compressing,chakrabarti2004fully}. Ferragina et al.\ define a
column-column similarity $\mathrm{CSM}_{ij}=R^{\mathrm{PNZ}}_{ij}/n$, where $R^{\mathrm{PNZ}}_{ij}$ is the number of
repeated pairs of nonzeros in columns $i$ and $j$ (this models how well RePair would compress the two columns if
placed together), computed by hashing or sorting in $\Oh(m^2n)$ or $\Oh(m^2n\log n)$ time, and optionally pruned to the
top $k$ values per column. They then use it in four reordering algorithms:
\begin{itemize}
\item \textbf{LKH}: a travelling-salesman formulation solved with the Lin--Kernighan--Helsgaun heuristic~\cite{helsgaun2000effective};
\item \textbf{PathCover}: a Kruskal-like algorithm that picks edges by decreasing weight if they extend disjoint paths;
\item \textbf{PathCover+}: a variant that recomputes the similarity of the nodes to a path after each extension (never the best in
  their experiments);
\item \textbf{MWM}: a maximum-weight matching on a bipartite graph that induces a predecessor/successor relation.
\end{itemize}

\subsection{Experimental results}

The experiments use eight matrices from machine-learning repositories, stored as $8$-byte doubles, and compare with
\texttt{gzip}, \texttt{xz} and CLA; the benchmark alternates right and left products
($\yy_i=A\xx_i$, $\xx_{i+1}\T\propto\yy_i\T A$), 500 times, which mimics conjugate-gradient iterations.
%The implementation (\texttt{mm-repair}) and the matrices are publicly available (\cref{tab:code}); 
The numbers 
below are taken from Tables~1 and~5 of \citet{ferragina2022improving}.
\cref{fig:gram-both} shows space and time for matrix--vector products on some
matrices.
\Cref{fig:gram-scatter}
relates, for each matrix, how much \texttt{re\_iv} compresses it to how much faster or slower the product is than on the dense matrix.

\begin{figure}[t]
\centering
\includegraphics[width=\linewidth]{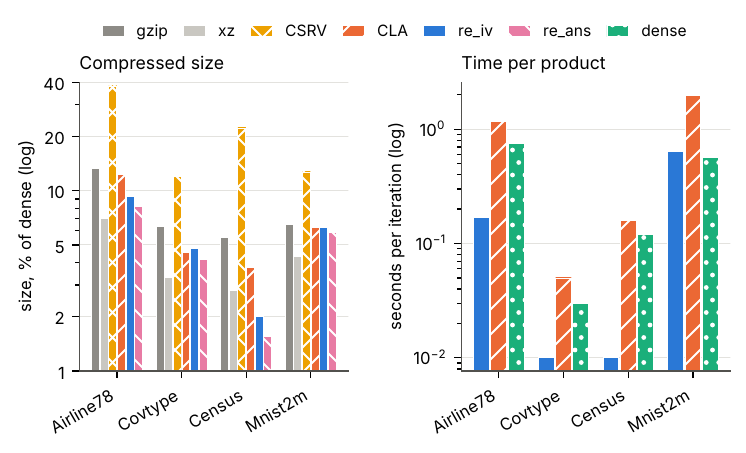}
\caption{Left: compressed size (percentage of the dense matrix, log scale) of
\texttt{gzip}, \texttt{xz}, CSRV and the grammar-based formats \texttt{re\_iv}
and \texttt{re\_ans}, with no column reordering (Table~1 of
\citet{ferragina2022improving}). Right: seconds per iteration (log scale) of
the matrix--vector benchmark for \texttt{re\_iv} (16 threads, with column
reordering), CLA (80 threads, time including compression) and the dense matrix
(Table~5 of the same paper).}
\label{fig:gram-both}
\end{figure}

\begin{figure}[t]
\centering
\includegraphics[width=0.62\linewidth]{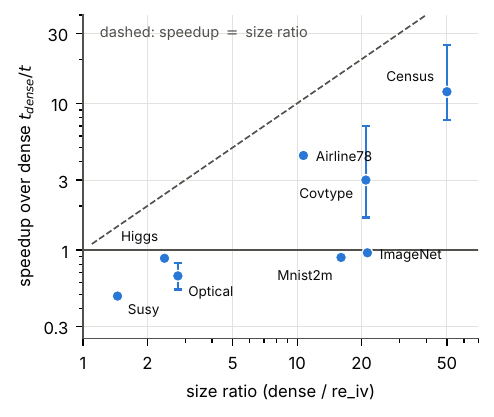}
\caption{Size ratio (dense size over \texttt{re\_iv} size) against speedup of the matrix--vector product over the uncompressed matrix, for the eight matrices of \citet{ferragina2022improving} (\texttt{re\_iv} with column reordering and 16 threads; Table~5 of that paper). The dashed line is speedup $=$ size ratio; the bars show the uncertainty caused by rounding the times to $0.01$\,s. Points below the horizontal line at $1$ are slower than the dense product.}
\label{fig:gram-scatter}
\end{figure}

The main findings are: (i)~the size is usually smaller than that of \texttt{gzip} and often not far from that of \texttt{xz}, which cannot be multiplied without decompression; (ii)~with 16 threads the
peak memory is for most matrices between $6\%$ and $50\%$ of the dense size, and for
five of the eight inputs below $20\%$; (iii)~compared with CLA, \texttt{re\_ans} is smaller on all the matrices on which CLA could run, except \texttt{Higgs}
(\texttt{re\_iv} is marginally larger than CLA on \texttt{Covtype} and \texttt{Mnist2m}), the peak memory of the grammar-based
representations is between $3.1\times$ and $19.1\times$ smaller, and they are at least twice as fast (CLA's running
time includes its compression, and CLA runs inside a Java system, which makes the comparison of absolute times delicate; CLA failed on \texttt{Susy}); (iv)~as \cref{fig:gram-scatter} shows, on the most compressible matrices (\texttt{Airline78},
\texttt{Covtype}, \texttt{Census}, with size ratios from $10$ to $50$) the compressed multiplication is even faster than the one on the dense matrix, by factors between $3$ and $12$; on
the least compressible ones (\texttt{Susy}, \texttt{Higgs}, \texttt{Optical}, with ratios below $3$) the product is slower than the dense one, and for \texttt{Mnist2m} and \texttt{ImageNet} the two are about equal.
Compression itself is slower than CLA on some matrices; yet it is a one-time cost.

The compression performance is mainly driven by the time to run RePair, which is linear in $\nnz(A)$.
For example, building the variant \texttt{re\_iv} on matrices with 285--400 million edges with 16 threads
takes 1.0 to 2.5 minutes. While there is no available data for direct comparisons, this seems to be considerably
faster than the compression times reported in \cref{sec:reference}.

With respect to column reordering, LKH is orders of magnitude slower than the others and is not competitive. The best of PathCover and MWM, applied per block of rows,
gave size reductions with respect to no reordering of up to $16\%$ (on \texttt{Covtype}) and $10\%$
(on \texttt{Airline78}), and reductions of the peak memory during the product between $5\%$ and $15\%$ on half of the matrices.
Reordering sometimes hurts: on \texttt{Mnist2m} it left memory unchanged and increased time by $5\%$.
For \texttt{Airline78} the memory reduction translates into a $25\%$ reduction of the average running time; on \texttt{Susy} reordering slightly increases the peak memory, with no significant change in time.

\subsection{Discussion and related work}

Grammar compression has also been used for graphs~\cite{claude2010fast,maneth2018grammar}
and for bipartite sparse binary matrices~\cite{yang2016bipartite}. The distinctive features of~\citet{ferragina2022improving}
are the explicit handling of the \emph{values} (so that real matrices are covered), products in time linear in the size of the grammar, and the column reordering step. The approach is
a natural candidate for becoming a new compressor within a CLA-like framework. The paper lists, as
open directions, row permutations and co-clustering, semiring-annotated data (joins computed by replacing $+$ by OR and
$\times$ by AND), columnar databases, and graphs; see also~\cite{tosoni2025toward} for lossless compressed formats in the context of greener matrix operations.
Limitations are the space for compression (RePair needs $\Oh(|S|)$ words of memory), the dependence of the left product on a backward scan of the rules, which
restricts the choice of encoders for $R$ (few compressors support right-to-left access), and the disappointing single-threaded
running time for large matrices, so blocks of rows and parallelism are needed.

\section{Biclique Decompositions}
\label{sec:biclique}

\subsection{Idea and representation}
\label{sec:biclique-idea}

A \emph{biclique} of a directed graph is a pair of node sets $(S,T)$ such that every node of $S$ has an edge to every node of $T$;
in the adjacency matrix it is an all-ones submatrix, up to permutation of rows and columns. A biclique with $|S|=s$ and $|T|=t$
describes $st$ edges with $s+t$ node identifiers, so large bicliques give large savings. \citet{feder1995clique} studied this kind of \emph{clique partition} as a way to compress graphs and, in the same
paper, to speed up graph algorithms by running them on the compressed graph. \citet{buehrer2008scalable} mined such dense patterns in Web graphs with pattern-mining techniques and replaced each
one by a \emph{virtual node} $w$ with edges $S\to w\to T$, which is the form later used by \citet{karande2009speeding}
to speed up algorithms on compressed Web graphs (although, as noted in~\cite{francisco2022graph}, that approach
requires correcting the results of the computation). \citet{hernandez2014compressed} designed a
scalable extraction method for Web and social graphs, which is the one used by the works surveyed here.

\begin{definition}[Biclique-based representation]\label{def:biclique}
Let $G=(V,E)$ be a directed graph, with adjacency matrix $A$. A biclique-based representation is a pair $(\mathcal{D},R)$, where $\mathcal{D}=\{(S_r,T_r)\}_r$ is a set of bicliques of $G$ and $R=E\setminus\bigcup_rS_r\times T_r$ is the \emph{residual} graph, the edges that
are in no biclique. We write $D_A = \sum_r \mathbf{1}_{S_r}\mathbf{1}_{T_r}\T$ and $R_A$ for the matrices of the edges covered by $\mathcal{D}$ and of $R$, so that $A=D_A+R_A$ (Boolean sum).
\end{definition}
Edges may belong to several bicliques (duplicates are allowed, because they
can lead to a better representation). The space is $\sum_r(|S_r|+|T_r|)+|R|$
words, with each set stored contiguously (the residual graph sorted by rows) and a bitmap to mark where each set starts; in practice
the identifiers are further compressed with differential encoding and universal codes~\cite{boldi2005codes}. 

Finding a set of bicliques that minimize space is NP-hard, hence extraction is heuristic.
The extraction algorithm of \citet{hernandez2014compressed} works in four steps: (1)~it computes a small MinHash signature~\cite{broder1998min} of each adjacency list, to estimate similarity cheaply; (2)~it clusters lists with similar signatures; (3)~inside each cluster, it sorts the elements of each list by their frequency in the cluster, so that frequent elements come first; and
(4)~it builds a prefix tree of the sorted lists, from which the bicliques are
extracted. This is repeated for several iterations with parameters such as the minimum degree of the nodes that are processed (\texttt{minAdjNodes}) and the minimum number of edges of the
extracted bicliques. The format is competitive in space
for both Web graphs and social graphs, which is its main advantage over referencing formats (\cref{sec:reference}), where
social networks are hard to compress.
The extraction heuristics make capturing non-contiguous sets of columns (and rows) feasible, even if still slower
than the linear-time RePair grammar construction of \cref{sec:grammar}. For example, \citet{araya2026compressed}
report 15 minutes to extract the bicliques of a graph with 640 million edges. We return to this point later.

\paragraph{Theoretical background.}
The size of a biclique decomposition is the total number of vertices of the bicliques. \citet{krapivin2026optimal} show that every $n$-vertex graph has a biclique partition (every edge in exactly one biclique) of total size about $\tfrac12 n^2/\lg n$, that this is tight, and that for \chg{a fixed} edge density $\gamma$ the optimal size is about $\tfrac{h_2(\gamma)}{2}\,n^2/\lg n$, where $h_2$ is the binary entropy; they give deterministic constructions in $\Oh(n^2)$ time, and $\Oh(m)$ \no{in sparse graphs with $m$ edges}\chg{from the adjacency lists of a graph with $m$ edges}. \chg{For $n\times n$ Boolean matrices (directed graphs with possible loops, or bipartite graphs with $n$ vertices on each side) the constants are $1$ and $h_2(\gamma)$, respectively, by the results of Lupanov and Nechiporuk that they restate.}
They note that this yields $\Oh(n^2/\log n)$-time algorithms for independent set
and cut queries. Since the product by a vector costs time proportional to that
size (\cref{sec:biclique-vec}), this implies that any Boolean matrix can also
be multiplied by a vector in $\Oh(n^2/\log n)$ operations after the decomposition (a consequence we draw).
These are worst-case bounds; the algorithms below are heuristics that exploit the structure of real graphs, and they also allow a residual set of edges not covered by any biclique.

\paragraph{Supporting access.} While \citet{hernandez2014compressed} showed how to extract all the nonzero cells
of a matrix row or column in optimal time, they cannot efficiently compute an
individual cell 
$A_{ij}$ (without extracting the whole row $i$ or the whole column $j$). This is because the edge $(i,j)$ can
appear in any biclique (or in the residual matrix). While they use a format
that allows jumping to the next biclique $(S,T)$ where $i \in S$ or $j \in T$,
the search takes $\Oh(|\mathcal{D}|)$ time in the worst case.

\paragraph{A space comparison across families.}

\citet{araya2026compressed} provide a (rough) comparison of compressed space across
the families we have studied, for binary matrices. \cref{fig:bic-space} shows the results in
graphic form, for seven graphs of $2.3$M to $640$M edges (from LAW, SNAP and
the Network Repository).
The comparison includes standard $k^2$-trees (\cref{sec:quadtree}), WebGraph as a representative of reference compression
(\cref{sec:reference}), and RePair32, an implementation of the compression of
\citet{ferragina2022improving} for binary matrices, as a representative of grammar compression
(\cref{sec:grammar}). It also shows the CSR/CSC baseline of
\citet{arroyuelo2025evaluating}.

\begin{figure}[t]
\centering
\includegraphics[width=\linewidth]{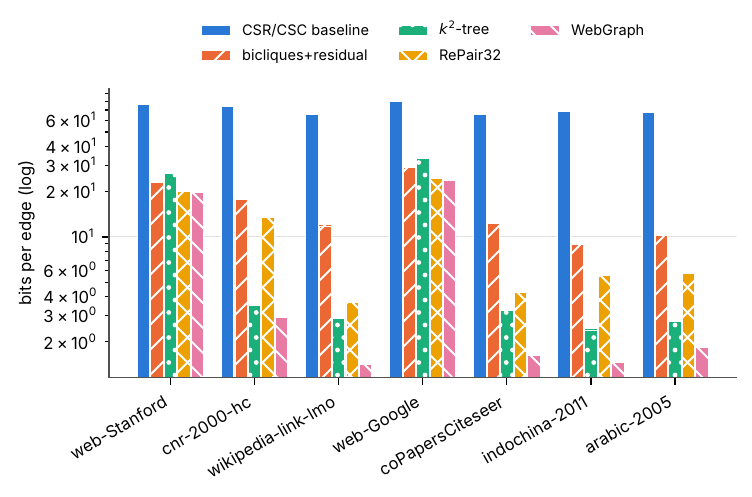}
\caption{Bits per edge of the CSR/CSC baseline, the biclique-based representation (adaptive extraction), the $k^2$-tree, RePair32 and WebGraph, on the input matrices (log scale; data from \citet{araya2026compressed}).}
\label{fig:bic-space}
\end{figure}

It can be seen that, in most cases, WebGraph uses much less space than the
others, followed by $k^2$-trees. Those use in turn much less space than RePair32, and
bicliques use the most space. Still, bicliques use only $13\%$--$36\%$ of the 
space of the baseline. 
% ($26\%$--$72\%$ if the baseline store only its CSR part, an improvement implemented in that paper \cite{araya2026compressed}).
On the other hand, only bicliques and $k^2$-trees have implemented
matrix--matrix
multiplications and, as shown in \cref{sec:bicliqueexp}, bicliques are much faster.

\subsection{Products by a vector}
\label{sec:biclique-vec}

\begin{algorithm}[t]
\caption{Right product $\yy=A\xx$ from a biclique decomposition
$\{(S_r,T_r)\}$ plus residual $R_A$~\cite{francisco2022graph}.}
\label{alg:biclique}
\begin{algorithmic}[1]
\State $\yy\gets 0$
\For{each biclique $(S_r,T_r)$}
  \State $c_r\gets\sum_{j\in T_r}x_j$
  \For{each $i\in S_r$}
     \State $y_i\gets y_i+c_r$
  \EndFor
\EndFor
\For{each residual edge $(R_A)_{ij}=1$}
  \State $y_i\gets y_i+x_j$
\EndFor
\State \Return $\yy$
\end{algorithmic}
\end{algorithm}

The biclique structure
enables computing the product as in \cref{alg:biclique} \cite{francisco2022graph}.
Its time is $\Oh\bigl(\sum_r(|S_r|+|T_r|)+\nnz(R_A)+n\bigr)$, that is,
proportional to the size of the compressed graph plus the output.
The same procedure with the roles of $S_r$ and $T_r$ exchanged ($c_r=\sum_{i\in S_r}y_i$, then add $c_r$ to $x_j$
for every $j\in T_r$) computes left products. A biclique is thus a shared
sub-computation that is evaluated once, in a way that resembles a nonterminal of a grammar, with the difference that
it shares a set of rows \emph{and} the same set of columns; a grammar rule on a CSRV sequence $\langle\ell,j\rangle\cdots$
factors out a set of columns that are \emph{consecutive in a row} (possibly after reordering) among the rows that contain
that subsequence, and can be nested arbitrarily, whereas a biclique has depth one.
In this sense, bicliques are the simplest instance of the general scheme of
\cref{sec:factorization}.

\begin{figure}[t]
\centering
\includegraphics[width=0.62\linewidth]{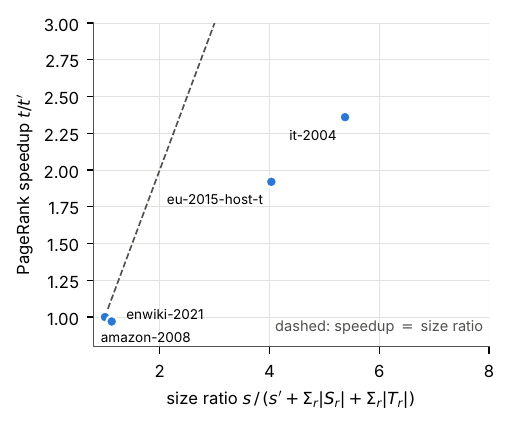}
\caption{PageRank on biclique decompositions, data from
\citet{francisco2022graph}: speedup $t/t'$ as a function of the size ratio
$s/(s'+\sum_r|S_r|+\sum_r|T_r|)$, with $s=\nnz(A)$, $s'=\nnz(R_A)$ residual edges and the sums counting the vertices of the bicliques. The dashed line is speedup $=$ size ratio. \texttt{eu-2015-host-t} is the transposed matrix.}
\label{fig:bic-pr}
\end{figure}

\paragraph{Experimental results.}
\cref{fig:bic-pr} shows the results on PageRank, obtained by
\citet{francisco2022graph} on top of the extraction software of
\citet{hernandez2014compressed}.
As with referencing, the speedup grows with the reduction ratio: with a ratio near $1$ the (small) cost of processing the bicliques
dominates and the method is slightly slower than the plain product, whereas on \texttt{it-2004} the speedup is $2.36\times$ (with a size ratio of $5.38$) and
on the transposed \texttt{eu-2015-host} it is $1.92\times$ (ratio $4.03$). Remarkably, the biclique sizes
(the totals of $|S_r|$ and $|T_r|$) have a super-linear effect on $s'$: a relatively small description removes many edges; on \texttt{it-2004} the bicliques add about $60$ million vertices to remove about $1\,000$ million edges.

%The evaluation covers few graphs because of limitations of the extraction software, which assumes that the number of edges fits in $32$ bits and that every node has a self-loop; the authors corrected the latter assumption by post-processing the result of the extraction, adding or removing the self-loops that it had added or lost.

\subsection{Boolean matrix--matrix products}
\label{sec:bicliquemult}

\citet{araya2026compressed} extend the biclique representation to support products of two Boolean matrices $C = A\times B$, with $A\in\B^{n\times m}$ and $B\in\B^{m\times p}$, and make it \emph{compositional}: the output is
also a biclique-based representation, so products can be chained. The idea builds on Schoor's algorithm~\cite{schoor1982fast}, which is also the basis of the prejoining technique of \cref{sec:qdags}.

\paragraph{Schoor's algorithm revisited.}
We recalled in \cref{sec:plain} the algorithm for $C=A\times B$, with $A$ stored by columns (CSC) and $B$ by rows (CSR): each $t$ in the intersection of the nonempty columns of $A$ and the nonempty rows of $B$ contributes a \emph{biclique} to $C$, formed by the rows $i$ of column $t$ of $A$ and the columns $j$ of row $t$ of $B$. The compact implementation of \citet{arroyuelo2025evaluating} takes $\Oh(\nnz(A)\nnz(B)\log m/m)$ time and $2(\nnz(C)+n+p)$ words of space, and is the baseline in the experiments below.

\paragraph{Products on bicliques.}
Write $A=D_A+R_A$ and $B=D_B+R_B$ (\cref{def:biclique}). Then
\begin{equation}\label{eq:bicmat}
A\times B \;=\; R_A\times R_B \;+\; R_A\times D_B \;+\; D_A\times R_B \;+\; D_A\times D_B ,
\end{equation}
and each of the four terms is computed with an adaptation of Schoor's algorithm, always producing bicliques rather than individual cells, except for the first one:
\begin{itemize}
\item \emph{Residual by residual} ($R_A\times R_B$) is the original algorithm of \citet{arroyuelo2025evaluating}, and produces the residual of the output in CSR format.
\item \emph{Matrix by bicliques} ($R_A\times D_B$). For each biclique $(S,T)$
of $\mathcal{D}_B$, intersect the nonempty columns of $R_A$ with $S$. Every common element
$t$ gives the rows of column $t$ of $R_A$. The result is a single biclique whose row set is the concatenation (after removing duplicates) of the rows obtained
  from $R_A$ for all common $t$, and whose column set is $T$.
\item \emph{Bicliques by matrix} ($D_A\times R_B$) is symmetric: for each biclique $(S,T)$ of $A$, intersect $T$ with the nonempty rows of $R_B$; the output biclique has rows $S$ and, as columns, the union of the columns of the rows of $R_B$ found.
\item \emph{Bicliques by bicliques} ($D_A\times D_B$). A table $\mathit{marks}(D_B)$ lists, for each vertex $v$ appearing in some row set $S$ of a biclique of $\mathcal{D}_B$, the column sets $T$ of all the bicliques of $\mathcal{D}_B$ whose row set contains $v$. For each biclique $(S,T)$ of $A$, every element of $T$ is looked up in $\mathit{marks}$,
  and the output biclique has rows $S$ and, as columns, the union of the column sets found.
\end{itemize}
The output is then $D_C=R_A\times D_B+D_A\times R_B+D_A\times D_B$ (Boolean
sum) and $R_C=R_A\times R_B$, which is again a biclique-based representation (some edge of $R_C$ may also be in $D_C$, which does not affect correctness). This composition is the simplest case of a general product of compressed matrices, which we present in \cref{sec:factprod}.
Compositionality lets matrices stay in this format \emph{throughout} a computation (and makes it easy to add transposition and sums to the algebra). Its drawback is that the output bicliques may be worse than those that extraction would find on $C$ (for example, bicliques with $|S|=1$ or $|T|=1$ are better in the residual),
so the output can be \emph{recompressed} at any time, by converting it to CSR
(using $\mathit{marks}$ to concatenate and deduplicate the column sets of each
row) and running the extraction again. Only the CSR of each matrix is stored;
the CSC needed by Schoor's algorithm is built on the fly, since the cost of
doing so is very small compared to that of the multiplication (the baseline of
\cref{sec:baseline-ops}, instead, stored both CSR and CSC).

\paragraph{Usage scenarios.}
The same algorithms support three scenarios: (i)~bicliques used only as an
\emph{accelerator}, with matrices in CSR on input and output (extraction
and conversion back to CSR is part of the multiplication process);
(ii)~bicliques as the permanent representation of inputs and outputs, with recompression of the output; and (iii)~the same without recompression.

\paragraph{Boolean sums (our derivation).} 
It is easy to compute $C=A+B$ as $D_C = D_A+D_B$ and $R_C=R_A+R_B$ (Boolean sums) in time proportional to the 
compressed spaces. That is, we simply take the union of the bicliques and of the residual edges, exploiting the
fact that we allow repeated edges. While the result is correct, it possibly uses more space than the result of 
recompressing $C$. An intermediate alternative is to improve the representation, for example eliminating duplicates
in $R_A+R_B$, eliminating bicliques contained in others, detecting bicliques that can be extended, and so on.
For example, contained bicliques can be detected by using Bloom filters and checking the cases that pass the 
filters. For extending cliques we can compute signatures for each $S_k$ in $\mathcal{D}_A$
and $S_l$ of $\mathcal{D}_B$; if the signatures coincide we can check that $S_k=S_l$ and, if so, merge bicliques 
$(S_k,T_k)$ and $(S_l,T_l)$ into one, $(S_k,T_k\cup T_l)$. We can analogously merge $S_k$ and $S_l$ when 
$T_k=T_l$. 

\subsection{Experimental results for matrix--matrix products}

\paragraph{Adaptive biclique extraction.}
\label{sec:adaptive}

\begin{figure}[t]
\centering
\includegraphics[width=\linewidth]{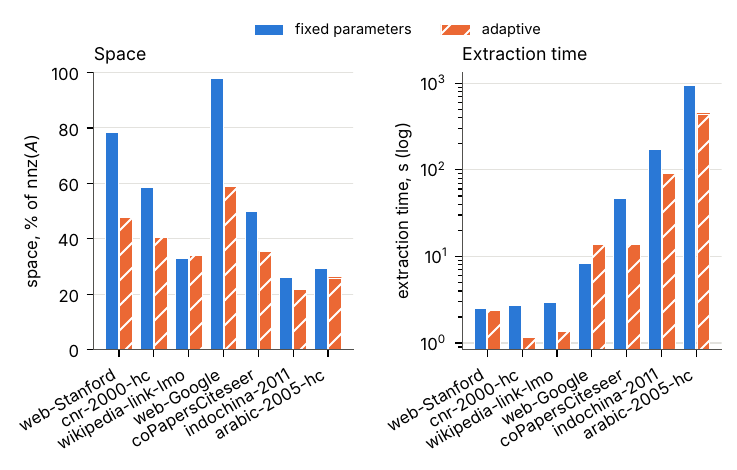}
\caption{Biclique extraction with fixed and adaptive parameters, ten
iterations (data from \citet{araya2026compressed}). Left: space of the
representation, $100\,(\nnz(R_A)+\sum_i(|S_i|+|T_i|))/\nnz(A)$, as a percentage of the number of edges. Right: extraction time in seconds (log scale).}
\label{fig:bic-extract}
\end{figure}

%The price to capture repetitions across non-contiguous matrix rows and columns is the high cost of the {\em biclique
%extraction} process. To give an idea, \citet{araya2026compressed} report, with sequential code, $442$ seconds ($7.4$ minutes) for a single adaptive extraction on \texttt{arabic-2005} ($22.7$ million nodes and $640$ million edges), against $952$ seconds for the fixed-parameter one, and $90$ seconds on \texttt{indochina-2011} ($194$ million edges). When computing $A^2$, which requires extracting bicliques from two matrices, the extraction is $884$ of the $1\,246$ seconds in total ($71\%$) on \texttt{arabic-2005} and $179$ of $262$ seconds ($69\%$) on \texttt{indochina-2011}; if the output is also recompressed, it is $87\%$ of the total on \texttt{arabic-2005}. Whether this pays off depends on the size of the product: on \texttt{indochina-2011} the baseline multiplication takes $8\,983$ seconds, so the extraction is $2\%$ of it, but on \texttt{web-Google} the adaptive extraction ($14$ seconds) takes $5.6$ times longer than the baseline multiplication ($2.4$ seconds).

The extraction algorithm of \citet{hernandez2014compressed} sets its parameters arbitrarily at the start and does not use information about the graph, so tuning it is cumbersome, and the extraction is the dominant cost of scenario (i). \citet{araya2026compressed} make it adaptive. First, to target large bicliques and reduce the work per iteration, the first iterations process only the nodes of
high degree: \texttt{minAdjNodes} is initialized to the $90$th percentile of the degree distribution (that is, only the $10\%$ of nodes with highest degree are processed). Second, after each iteration, if the number of bicliques extracted is below $0.8$ times the historical average per iteration, \texttt{minAdjNodes} is reduced by $20\%$, which lets later iterations reach nodes of lower degree.
\Cref{fig:bic-extract} compares the adaptive extraction with the fixed-parameter one (ten iterations in both): the adaptive version is better in compression and in time in almost all cases, usually taking $45\%$ of the time (as little as $29\%$) and, measured in bits per edge, $61\%$--$88\%$ of the space of the non-adaptive result. The exceptions are \texttt{web-Google}, where the adaptive version
takes $162\%$ of the time but the non-adaptive one practically fails to compress (it keeps $98\%$ of the space of the plain edge list, against $59\%$ for the adaptive one), and \texttt{wikipedia-link-lmo}, where the adaptive result is slightly larger ($104\%$ of the space).

\paragraph{Multiplications.}
\label{sec:bicliqueexp}
\citet{araya2026compressed} extend the C++ implementation of
\cite{arroyuelo2025evaluating} and measure matrix--matrix multiplication times on the same graphs of
\cref{sec:biclique-idea} on a $256$\,GB AMD EPYC server, with sequential code. To measure products of the form $A\times B$ they compute $A^2$ and charge the extraction time of the two matrices.
We summarize the main findings; \cref{fig:bic-mult} gives the data. 

\begin{figure}[t]
\centering
\includegraphics[width=\linewidth]{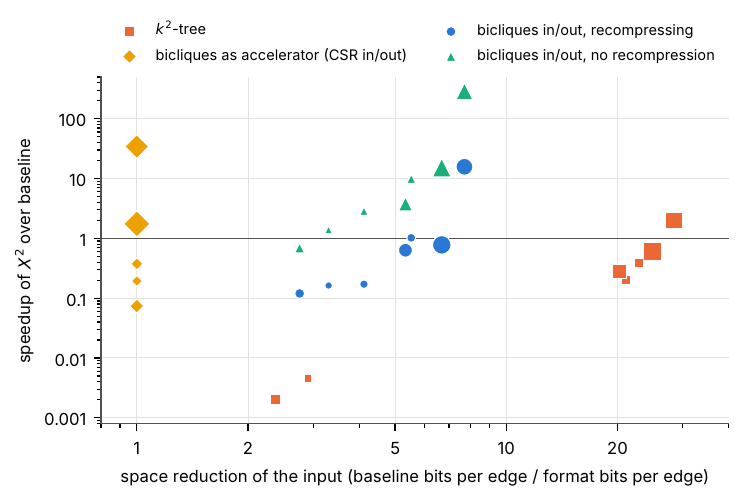}
\caption{Space reduction of the input (bits per edge of the CSR baseline over bits per edge of the format) against speedup in the time to compute $A^2$ with respect to the CSR implementation of Schoor's algorithm, for the seven graphs of \citet{araya2026compressed} (log--log scale; marker size grows with the number of edges; values below $1$ in the vertical axis are slowdowns). In the accelerator scenario the matrices stay in CSR, so the space reduction is $1$ and the extraction time is included.}
\label{fig:bic-mult}
\end{figure}

\begin{itemize}
\item \emph{Bicliques as accelerator} (scenario i; the diamonds of \cref{fig:bic-mult}). The matrices stay in CSR and the bicliques are used only to speed up the product. Even counting extraction, the method is up to $34.3\times$ faster than the baseline on \texttt{indochina-2011} and $1.75\times$ on the largest graph (\texttt{arabic-2005}), whereas on small matrices, which the baseline multiplies in a few seconds, it is slower.
  Extraction is the dominant cost, which is why the improvements of \cref{sec:adaptive} are key.
\item \emph{Bicliques as representation, with recompression} (scenario ii; the
circles). The output uses $18\%$--$38\%$ of the space of the baseline (which
stores only CSR); on large matrices it is competitive with the baseline in time ($30\%$--$60\%$ slower, and $16\times$ faster in one case), and it is slower by up to $8\times$ on small ones.
  It is almost always faster than the $k^2$-tree, by up to $59\times$ (the $k^2$-tree wins only once, by $18\%$), although the output of the $k^2$-tree is generally smaller by a factor of $2.3$--$2.6$.
\item \emph{Without recompression} (scenario iii; the triangles). The
multiplication is up to $20\times$ faster than in scenario ii, the compressed
space degrades by $7\%$--$40\%$ in most cases, the output uses $20\%$--$52\%$
of the space of the baseline (storing only CSR), and it is almost always faster than the baseline, by up to $289\times$ (and slower in one case, by $44\%$).
\item \emph{Powers of a matrix.} Matrices become much denser when multiplied: the number of edges grows by about an order of magnitude with each squaring, and matrices that are products are much more compressible with bicliques (which is not surprising in light of Schoor's algorithm, which generates them as bicliques). When computing $A^4$ with and without recompression of $A^2$, both are faster than the baseline (up to $13.6\times$) and the $k^2$-tree
  (up to $586\times$), and recompressing $A^2$ makes the whole process almost twice as fast as not doing so, so intermediate recompression pays off. The two largest matrices could not be raised to the fourth power for lack of memory.
\end{itemize}

Overall, \cref{fig:bic-mult} shows a trade-off between space and time: the $k^2$-tree compresses the input the most ($20$--$28$ times on five of the graphs) but is slower than the baseline in six of the seven cases; bicliques compress less ($3$--$8$ times) and, without recompression, give the largest speedups. The small graphs, with the smallest markers, are always below the large ones, because the fixed cost of the extraction is only paid back when the product is expensive.

\subsection{Discussion}

Bicliques are the most \emph{structured} family: they have the clearest semantics (they capture communities of nodes that point to the same targets) and they work well on graphs that are not handled by referencing. 
Exploiting the fact that Schoor's algorithm produces bicliques by construction, they feature rather natural 
compositional matrix--matrix multiplication algorithms, the only one implemented and evaluated in the family
of the factorization type (\cref{sec:factorization}, presented next). Their
weaknesses are relatively large size compared to other compression methods, and
difficulty of extraction: finding large bicliques is a heuristic search.

As a
byproduct, a biclique decomposition is a form of \emph{lossless graph summarization}; \citet{yang2016bipartite} propose
bipartite grammar-based representations of sparse binary matrices which, in the language of \cref{sec:factorization}, sit
between this family and the previous one. We are not aware of an experimental study of bicliques for \emph{real}-valued matrices, although
\cref{alg:biclique} extends immediately to a decomposition into blocks in which all the nonzeros have the same value $\alpha$ (the block contributes $\alpha\,c_r$ instead of $c_r$).
The matrix--matrix algorithms are limited to Boolean matrices; the authors leave transitive closure on bicliques, parallel and GPU versions, and weighted and dynamic graphs, for future work.

% Section moved here from the former Sections 3.1 and 3.2 of the taxonomy.
\section{A Common View: Sparse Straight-Line Factorizations}
\label{sec:factorization}

The three families of \cref{sec:reference,sec:grammar,sec:biclique} look different, but they share a simple algebraic structure, which makes it easy to see why the same kind of
guarantee holds for all of them, and which we find the most useful way to
think about them. In this section we abstract that structure, show how each
family instantiates it, and derive from it a way of multiplying two compressed
matrices. The framing is ours; it is not stated in the original papers, though it is abstracted from
their algorithms. \chg{The section also contains new results of ours: the relation to linear circuits (\cref{sec:circuits}), the cost of products, sums and Kronecker products (\cref{sec:factprod}), simple recompression steps (\cref{sec:recomp}), and the size of the smallest factorization, with its lower bounds, separations and parallelism (\cref{sec:size}). Their proofs are in \cref{app:proofs}, except for the short ones of the sizes of the three instances (Lemmas~\ref{lem:ref}, \ref{lem:gram} and \ref{lem:bic}).} The $k^2$-tree family of \cref{sec:quadtree}, which as we saw follows a different principle, is not covered by it.

\subsection{Definition and products by a vector}
\label{sec:factor-vector}

\begin{definition}[Sparse straight-line factorization]
A \emph{sparse straight-line factorization} of $A\in\mathbb{K}^{n\times m}$, over a ring $(\mathbb{K},+,\cdot)$ \chg{(not necessarily commutative)} consists of sparse matrices
$P\in\mathbb{K}^{n\times q}$, $Q\in\mathbb{K}^{q\times q}$ (strictly lower triangular, hence nilpotent), $T\in\mathbb{K}^{q\times m}$ and
$T'\in\mathbb{K}^{n\times m}$ such that
\begin{equation}\label{eq:factorization}
A \;=\; P\,(I-Q)^{-1}\,T \;+\; T' .
\end{equation}
Its size is $\nnz(P)+\nnz(Q)+\nnz(T)+\nnz(T')$ (an identity matrix, when used for $P$, is implicit and not counted).
\chg{We call $q$ the number of \emph{nodes}; it is not part of the size. The \emph{depth} of the factorization is the number of nodes of the longest chain $k_1<\dots<k_d$ with $Q_{k_{t+1}k_t}\neq0$ for all $t$ (and $0$ if $q=0$).}
\end{definition}

To interpret this formula, recall that $C=A\times B$ can be understood as $C_{i\ast} = \sum_j A_{ij}B_{j*}$, that is, the entries of $A_{i\ast}$ give the coefficients of the linear combination of the rows of $B$ that form the row $C_{i\ast}$. If $A$ is sparse, we can also see $A_{i\ast}$ as telling which rows of $B$ is $C_{i\ast}$ made of.
Now consider vectors (or ``fragments'') $F_{k\ast}\in\mathbb{K}^{m}$ defined by
$F_{k\ast} = T_{k\ast}+\sum_{i<k}Q_{ki}\,F_{i\ast}$, that is, a sparse set of column positions (the nonzeros of $T_{k\ast}$) plus a sparse combination of earlier fragments (earlier because $Q$ is lower-triangular). Letting $F$ be the matrix with the fragments $F_{k\ast}$ as rows, we have $F = T + QF$, or $F = (I-Q)^{-1}T$.\footnote{Since $I-Q$ is lower-triangular with 1s in the diagonal, it has an inverse; we use this inverse for notational convenience only, we never compute it.}
The matrix $P$ then says
which fragments each row of $A$ is made of, and $T'$ holds the entries of $A$ that are not covered by any shared node (the \emph{residual}). Equation~\eqref{eq:factorization} is then $A = PF + T'$, that is,
$A_{i\ast}=\sum_kP_{ik}F_{k\ast}+T'_{i\ast}$.

The intuition is that of a \emph{straight-line program} (a sequence of assignments, each using only earlier values) that builds each row of $A$, seen as a sparse vector indexed by the column positions $1,\dots,m$ (see \cref{fig:factorization}; when computing $\yy=A\xx$, the node $x_j$ holds the entry of $\xx$ at column position $j$).
The $q$ \emph{auxiliary nodes} (in the middle in the figure) stand for shared sub-computations, representing the fragments $F_{k\ast}$. Finally, $P$ sends those values from $F$ to the output vector $\yy$. Matrix $T'$ sends values directly from $\xx$ to $\yy$.

Note that the entries of the four matrices belong to a ring in general;
Boolean matrices are regarded as $0/1$ matrices over the integers, and the
factors may then contain entries $-1$, for instance in the residual of a
referencing format.\footnote{Over a \chg{(not necessarily commutative)} semiring $(\mathbb{S},+,\cdot)$, the same
factorization works by writing $\sum_{i<q}Q^i$ instead of
$(I-Q)^{-1}$, and \cref{thm:fact} holds with semiring operations.
Referencing formats that \emph{delete} elements need a ring, however. 
On the other hand, bicliques may overlap only if the sum
is idempotent (as for $\vee$ or $\min$).\label{fn:semiring}}

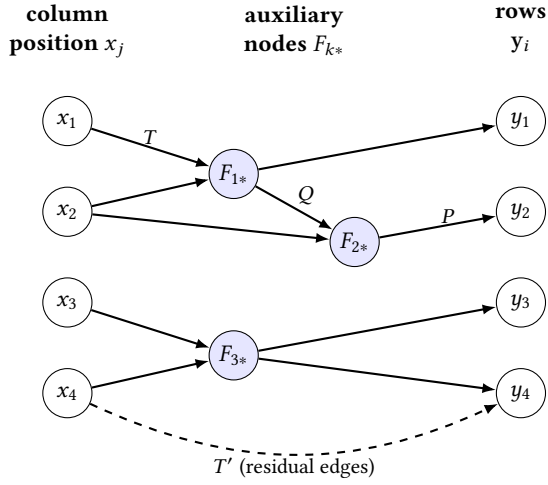
\begin{figure}[t]
\centering
\begin{tikzpicture}[
  nd/.style={draw,circle,minimum size=6.5mm,inner sep=0pt,font=\small},
  sh/.style={nd,fill=blue!10},
  ar/.style={-{Latex[length=1.8mm]},thick},
  lb/.style={font=\footnotesize,inner sep=1pt}]
\node[font=\small\bfseries,align=center] at (0,3.0) {column\\position $x_j$};
\node[font=\small\bfseries,align=center] at (3.0,3.0) {auxiliary\\nodes $F_{k\ast}$};
\node[font=\small\bfseries,align=center] at (6.0,3.0) {rows\\$\yy_i$};
\foreach \i/\y in {1/1.8,2/0.6,3/-0.6,4/-1.8}{\node[nd] (x\i) at (0,\y) {$x_{\i}$};}
\node[sh] (u1) at (2.2,1.1) {$F_{1\ast}$};
\node[sh] (u2) at (3.8,0.2) {$F_{2\ast}$};
\node[sh] (u3) at (2.2,-1.3) {$F_{3\ast}$};
\foreach \i/\y in {1/1.8,2/0.6,3/-0.6,4/-1.8}{\node[nd] (y\i) at (6.0,\y) {$y_{\i}$};}
\draw[ar] (x1) -- (u1) node[lb,midway,above] {$T$};
\draw[ar] (x2) -- (u1);
\draw[ar] (x3) -- (u3);
\draw[ar] (x4) -- (u3);
\draw[ar] (x2) -- (u2);
\draw[ar] (u1) -- (u2) node[lb,midway,above right] {$Q$};
\draw[ar] (u2) -- (y2) node[lb,midway,above right] {$P$};
\draw[ar] (u1) -- (y1);
\draw[ar] (u3) -- (y3);
\draw[ar] (u3) -- (y4);
\draw[ar,dashed] (x4) to[out=-25,in=-155] node[lb,midway,below] {$T'$ (residual edges)} (y4);
\end{tikzpicture}
\caption{The evaluation graph of a sparse straight-line factorization (schematic): the inputs $x_j$ correspond to the column positions of $A$ and the outputs $y_i$ to its rows. Edges $T$ go from column positions to nodes, $Q$ between nodes (from earlier to later), $P$ from nodes to rows, and $T'$ directly from column positions to rows.
The right product $\yy=A\xx$ propagates values from left to right; the left product propagates them from right to left over the same edges. The cost of either is the number of edges.}
\label{fig:factorization}
\end{figure}

\paragraph{Right and left products.}
The right product $A\xx$ is evaluated by \emph{forward substitution}: compute $\ww=T\xx$,
then $\mathbf{u}_k=\ww_k+\sum_{i<k}Q_{ki}\mathbf{u}_i$ for $k=1,\dots,q$ (so $\mathbf{u}=(I-Q)^{-1}T\xx$), and finally $\yy=P\mathbf{u}+T'\xx$.
\no{The left product is the transpose of this computation: $\mathbf{g}=P\T\yy$, then backward substitution on $\mathbf{z}$
with $\mathbf{z}=(I-Q\T)^{-1}\mathbf{g}$ (that is, $\mathbf{z}_k=\mathbf{g}_k+\sum_{i>k}Q_{ik}\mathbf{z}_i$ for $k=q,\dots,1$), and finally $\xx=T\T\mathbf{z}+T'{}\T\yy$.}\chg{The left product is the same computation traversed backwards: $\mathbf{g}\T=\yy\T P$, then backward substitution $\mathbf{z}_k=\mathbf{g}_k+\sum_{i>k}\mathbf{z}_iQ_{ik}$ for $k=q,\dots,1$ (so $\mathbf{z}\T=\mathbf{g}\T(I-Q)^{-1}$), and finally $\yy\T A=\mathbf{z}\T T+\yy\T T'$. Here $\yy\T$ is just the row vector with the entries of $\yy$: we never transpose products, since $\mathbb{K}$ may not be commutative.}
\Cref{fig:factorization} shows the idea again: the products now traverse the same graph in opposite directions.

\begin{theorem}\label{thm:fact}
Given a sparse straight-line factorization of size $s$, both $A\xx$ and $\yy\T A$ can be computed with $\Oh(s)$ arithmetic operations and $\Oh(q+n+m)$ words of working space.
\end{theorem}
\begin{proof}
In each of the substitutions every nonzero of $P$, $Q$, $T$ and $T'$ is used exactly once, in one multiplication and one addition, and the vectors $\ww,\mathbf{u}$ (or $\mathbf{g},\mathbf{z}$) have $q$ entries. The triangular structure of $Q$ guarantees that when node $k$ is processed all the nodes it depends on have been processed (in the order $1,\dots,q$ for the right product, and in the reverse order for the left product).
\chg{For the left product, note that $\yy\T A=(\yy\T P)(I-Q)^{-1}T+\yy\T T'=\mathbf{z}\T T+\yy\T T'$, where $\mathbf{z}\T=\mathbf{g}\T(I-Q)^{-1}$, that is, $\mathbf{z}\T(I-Q)=\mathbf{g}\T$; this is the system that the backward substitution solves, as $Q$ is strictly lower triangular.}
\end{proof}

\paragraph{Why it is useful.}
\Cref{thm:fact} is the only algorithmic fact behind the results of \cref{sec:reference,sec:grammar,sec:biclique}, and it
has four consequences that apply to all the formats at once. 
\begin{enumerate}[label=(\roman*)]
\item The \emph{compression ratio relative to the plain sparse format}, $\nnz(A)/s$, is the reduction in arithmetic operations, and hence an upper bound on the speedup one can expect from arithmetic alone. 
\item Left products cost exactly the same as right products, so there is no need to store \chg{a second compressed copy of the matrix to multiply from the left (for example, a representation of $A\T$, in the commutative case, to compute $\yy\T A$ as $(A\T\yy)\T$)}\no{the transpose}, which halves the space for algorithms such as PageRank or conjugate gradient that need both. 
\item The \emph{depth} of the nodes in $Q$
(the longest chain of dependencies) is the minimum number of sequential steps and thus it governs parallelism.
For example, referencing trees may be deep, bicliques have depth one, and
grammars can be made of logarithmic depth by design \chg{(see \cref{sec:size})}.
\item Different families are different ways of \emph{finding} a small factorization, so that comparing them amounts to comparing the quality and the cost of the corresponding search problems, not the multiplication algorithms.
\end{enumerate}

\chgbegin
\paragraph{Transposition.}
The size of a factorization does not change if we transpose the matrix. Over a commutative ring, transposing \cref{eq:factorization} gives $A\T=T\T(I-Q\T)^{-1}P\T+T'{}\T$, where $I-Q\T$ is upper triangular. If $J$ is the permutation matrix that reverses the order of the nodes, then $JQ\T J$ is strictly lower triangular, and $A\T=(T\T J)\,(I-JQ\T J)^{-1}\,(JP\T)+T'{}\T$ is a factorization of $A\T$ with the same size and depth, where the roles of $P$ and $T$ are exchanged. (In a non-commutative ring the same holds with $A\T$ taken over the opposite ring.) Thus $s^\ast(A)=s^\ast(A\T)$, where $s^\ast$ is the minimum size (\cref{sec:size}). This is not the case for referencing: \citet{gagie2026set} show matrices $A$ whose symmetric-difference measure satisfies $\Delta(A)=o(\Delta(A\T))$ (and vice versa).
\chgend

\subsection{The three families as instances}
We show how the families of \cref{sec:reference,sec:grammar,sec:biclique} are particular cases of the general principle, with the notation of those sections.

\begin{itemize}
\item \emph{Referencing} (\cref{sec:reference}): It encodes each row $A_{i\ast}$ as a referenced row $A_{r_i\ast}$ with $r_i<i$ (if any) plus a (hopefully sparse) {\em residual} row $A'_{i\ast}$. This fits in the general scheme with 
$q=n$, $P=I$, $Q_{i,r_i}=1$ (if it exists) and $0$ elsewhere, $T=A'$ and $T'=0$, that is, $A=(I-Q)^{-1}A'$:
row $i$ is the residual row $A'_{i\ast}$ plus the (already computed) row of its reference, and $Q$ is the forest of references\chg{. The size is $\nnz(A')$ plus the number of references, and the depth is that of the forest (Lemma~\ref{lem:ref})}.
\item \emph{Grammars} (\cref{sec:grammar}): The $q$ nodes are the nonterminals; $Q_{ki}=1$ if the right-hand side of $N_k$ mentions $N_i$ (and $0$ if not), and
  $T_{kj}=V[\ell]$ if it mentions the terminal $\langle \ell,j\rangle$; $P$ records which nonterminals appear in the part of the final sequence $C$ that encodes each row, and $T'$ does the same for the terminals that remain in $C$\chg{. The size is at most $2|R|+|C|-n$, because every symbol of a rule or of $C$ that is not a separator contributes at most one nonzero (Lemma~\ref{lem:gram})}.
\item \emph{Bicliques} (\cref{sec:biclique}): the $q$ nodes are the bicliques $(S_r,T_r)$; $P$ and $T$ are the indicator matrices
  of the sets $S_r$ and $T_r$, $Q=0$, and $T'=R_A$ is the residual matrix (the cells not in any biclique): $A=R_A+\sum_r \mathbf{1}_{S_r}\mathbf{1}_{T_r}\T$\chg{. The size is $\sum_r(|S_r|+|T_r|)+|R|$ and the depth is one (Lemma~\ref{lem:bic})}.
\end{itemize}
\chgbegin
The sizes stated above follow from the next three lemmas. Here a \emph{fragment} is a row of $F=(I-Q)^{-1}T$.

\begin{lemma}[Referencing]\label{lem:ref}
Let $(A',\mathbf{r})$ be a referencing representation of $A$. Then $q=n$, $P=I$, $Q_{i,r_i}=1$ if $r_i\ne0$ and $0$ elsewhere, $T=A'$ and $T'=0$ is a sparse straight-line factorization of $A$. Its size is $\nnz(A')$ plus the number of references $r_i>0$, its depth is that of the reference forest, and $Q$ is the (forest of) references.
\end{lemma}
\begin{proof}
The fragments satisfy $F_{i\ast}=A'_{i\ast}+F_{r_i\ast}$ (with $F_{0\ast}=\mathbf{0}$), which by induction on $i$ equals $\Vv_i$, since $r_i<i$. Hence $F=A$ and $A=(I-Q)^{-1}A'$ with $P=I$.
\end{proof}

\begin{lemma}[Grammars]\label{lem:gram}
Let $(C,R,V)$ be a grammar-compressed representation of $A$. Take as nodes the $q=|R|$ rules and start with all matrices in zeros. Then, for each rule $N_k\to XY$ and each child $Z\in\{X,Y\}$, increment $Q_{ki}$ if $Z=N_i$ or increase $T_{kj}$ by $V[\ell]$ if $Z=\langle\ell,j\rangle$. For each row $i$ and each symbol of $C_i$, increment $P_{ik}$ if the symbol is $N_k$, and increase $T'_{ij}$ by $V[\ell]$ if it is $\langle\ell,j\rangle$. This is a sparse straight-line factorization of $A$ of size at most $2|R|+|C|-n$.
\end{lemma}
\begin{proof}
By induction on the rules, $F_{k\ast}$ is the vector whose entry $j$ is the sum of $V[\ell]$ over the terminals $\langle\ell,j\rangle$ in the expansion of $N_k$, since it is the sum of the contributions of its two children. Then $\sum_kP_{ik}F_{k\ast}+T'_{i\ast}$ adds the contributions of all the symbols of $C_i$, which is the expansion of row $i$, that is, $A_{i\ast}$ (every pair $\langle\ell,j\rangle$ in a row stands for $A_{ij}=V[\ell]$, and a column appears once per row). Each symbol of a rule or of $C$ that is not a separator contributes at most one nonzero, which gives the size.
\end{proof}

\begin{lemma}[Bicliques]\label{lem:bic}
Let $(\mathcal{D},R)$ be a biclique decomposition of $A$. Take as nodes the bicliques, let $P$ and $T$ be the indicator matrices of the sets $S_r$ and $T_r$ (i.e., $P_{ir}=1$ if $i\in S_r$ and $0$ otherwise, and $T_{rj}=1$ if $j\in T_r$ and $0$ otherwise), $Q=0$, and $T'=R_A$. This is a sparse straight-line factorization of $A$ of size $\sum_r(|S_r|+|T_r|)+|R|$ and depth $1$.
\end{lemma}
\begin{proof}
$PT=\sum_r\mathbf{1}_{S_r}\mathbf{1}_{T_r}\T=D_A$ and $T'=R_A$, so $A=D_A+R_A$. The numbers of nonzeros of $P$ and $T$ are $\sum_r|S_r|$ and $\sum_r|T_r|$.
\end{proof}
\chgend

%\chg{The differential compressed binary trees of \cref{sec:refmatmat} are also a referencing format, in which the residual $T=A'$ is stored as trees instead of integers.}

\paragraph{A worked example.}
Consider the Boolean matrix
\[
A=\begin{pmatrix}
1&1&1&1&0&0&0&0\\
1&1&1&1&1&0&0&0\\
1&1&1&1&1&0&0&1\\
0&0&1&1&0&0&1&1\\
1&1&0&0&0&0&1&1\\
0&0&0&0&1&1&1&1
\end{pmatrix},\qquad \nnz(A)=27 .
\]

\begin{table}[t]
\centering\small
\caption{Three sparse straight-line factorizations of the matrix $A$ of the running example. Sizes count the nonzeros of $P$, $Q$, $T$ and $T'$ (with $P=I$ implicit for referencing). The three differ in shape: referencing has a chain of depth~$3$ ($1\to2\to3$), the grammar a nested rule ($N_5\to N_1N_2$), and the bicliques share column sets among non-adjacent rows and with non-consecutive columns ($\{3,6\}\times\{5,8\}$).}
\label{tab:example}
\begin{tabular}{@{}p{1.6cm}p{4.2cm}p{6.0cm}r@{}}
\toprule
Family & Nodes & Residual $T'$ and connections & Size\\
\midrule
Referencing & $q=6$ rows. $r_2=1$, $r_3=2$, $r_1=r_4=r_5=r_6=0$ & $A'$: row 1 $\{1..4\}$; row 2 $\{+5\}$; row 3 $\{+8\}$; row 4 $\{3,4,7,8\}$; row 5 $\{1,2,7,8\}$; row 6 $\{5..8\}$ ($\nnz(A')=18$); $Q$ has $2$ entries & $20$\\
Grammar & $N_1\to 1\,2$; $N_2\to 3\,4$; $N_3\to 5\,6$; $N_4\to 7\,8$; $N_5\to N_1N_2$ & rows: $N_5$, $N_5\,5$, $N_5\,5\,8$, $N_2N_4$, $N_1N_4$, $N_3N_4$ ($P$: 9 entries; $T'$: 3; $T$: 8; $Q$: 2) & $22$\\
Bicliques & $(\{1,2,3\},\{1,2,3,4\})$, $(\{4,5\},\{7,8\})$ and
$(\{3,6\},\{5,8\})$ & residual edges $(2,5),(4,3),(4,4),(5,1),(5,2),$ $(6,6),(6,7)$ ($T'$: 7; $P$: 7; $T$: 8) & $22$\\
\bottomrule
\end{tabular}
\end{table}

Rows $1$--$3$ share the columns $1$--$4$, rows $4$--$5$ share the columns $7$--$8$, and the columns $5$ and $8$ are shared by rows $3$ and $6$, which are not adjacent; in row $6$ the columns $6$ and $7$ lie between them. \Cref{tab:example} shows the three factorizations.
The corresponding \cref{eq:factorization} for {\em Referencing} is:

{\setlength{\arraycolsep}{2.5pt}\small
\[
A \;=\;
\overbrace{
\left( I \,-\,
\smash[b]{\underbrace{\begin{pmatrix}
0&0&0&0&0&0\\
1&0&0&0&0&0\\
0&1&0&0&0&0\\
0&0&0&0&0&0\\
0&0&0&0&0&0\\
0&0&0&0&0&0
\end{pmatrix}}_{Q}}
\right)^{-1}
\smash[b]{\underbrace{\begin{pmatrix}
\hphantom{-}1&\hphantom{-}1&\hphantom{-}1&\hphantom{-}1&\hphantom{-}0&\hphantom{-}0&\hphantom{-}0&\hphantom{-}0\\
\hphantom{-}0&\hphantom{-}0&\hphantom{-}0&\hphantom{-}0&\hphantom{-}1&\hphantom{-}0&\hphantom{-}0&\hphantom{-}0\\
\hphantom{-}0&\hphantom{-}0&\hphantom{-}0&\hphantom{-}0&\hphantom{-}0&\hphantom{-}0&\hphantom{-}0&\hphantom{-}1\\
\hphantom{-}0&\hphantom{-}0&\hphantom{-}1&\hphantom{-}1&\hphantom{-}0&\hphantom{-}0&\hphantom{-}1&\hphantom{-}1\\
\hphantom{-}1&\hphantom{-}1&\hphantom{-}0&\hphantom{-}0&\hphantom{-}0&\hphantom{-}0&\hphantom{-}1&\hphantom{-}1\\
\hphantom{-}0&\hphantom{-}0&\hphantom{-}0&\hphantom{-}0&\hphantom{-}1&\hphantom{-}1&\hphantom{-}1&\hphantom{-}1
\end{pmatrix}}_{T}}
}^{F}
\]
}

\vspace{10mm}

Here the references form a chain, $1\to2\to3$, so the row $3$ is obtained by adding $8$ to the row $2$, which in turn is the row $1$ plus $5$; the other rows are stored explicitly (their reference is $\emptyset$). The rows $4$--$6$ could also reference one another at the same weight, but they would pay for a pointer.

The decomposition for {\em Grammar} is as follows. The rules $N_1,\dots,N_4$ are the runs of columns that are repeated in several rows, and $N_5\to N_1N_2$ is a rule that uses other rules, so $Q\ne0$: the fragment of $N_5$ is $F_5=F_1+F_2=(1,1,1,1,0,0,0,0)$, the columns $1$--$4$. Each row of $A$ is then a short sequence of nonterminals and terminals given by $P$ and $T'$; row $3$ needs the terminals $5$ and $8$, which are not consecutive in row $6$, so no rule can share them between the two rows. Rules can be nested, but they can share only runs of consecutive columns inside a row.

{\setlength{\arraycolsep}{2pt}\small
\[
A \;=\;
\smash[b]{\underbrace{\begin{pmatrix}
0&0&0&0&1\\
0&0&0&0&1\\
0&0&0&0&1\\
0&1&0&1&0\\
1&0&0&1&0\\
0&0&1&1&0
\end{pmatrix}}_{P}}
\overbrace{
\left( I \,-\,
\smash[b]{\underbrace{\begin{pmatrix}
0&0&0&0&0\\
0&0&0&0&0\\
0&0&0&0&0\\
0&0&0&0&0\\
1&1&0&0&0
\end{pmatrix}}_{Q}}
\right)^{-1}
\smash[b]{\underbrace{\begin{pmatrix}
1&1&0&0&0&0&0&0\\
0&0&1&1&0&0&0&0\\
0&0&0&0&1&1&0&0\\
0&0&0&0&0&0&1&1\\
0&0&0&0&0&0&0&0
\end{pmatrix}}_{T}}
}^{F}
\;+\;
\smash[b]{\underbrace{\begin{pmatrix}
0&0&0&0&0&0&0&0\\
0&0&0&0&1&0&0&0\\
0&0&0&0&1&0&0&1\\
0&0&0&0&0&0&0&0\\
0&0&0&0&0&0&0&0\\
0&0&0&0&0&0&0&0
\end{pmatrix}}_{T'}}
\]
}

\vspace{12mm}

Finally, the factorization by {\em Bicliques} has $Q=0$ (depth one), so the
factor $(I-Q)^{-1}$ is not needed. Its
nodes are complete bipartite subgraphs: the first one connects the rows
$1$--$3$ with the columns $1$--$4$, the second the rows $4$--$5$ with the
columns $7$--$8$ and the third the non-adjacent rows $3$ and $6$ with the
non-consecutive columns $5$ and $8$: in row $6$ the columns $6$ and $7$ lie between them, so a grammar cannot capture this sharing with a rule, whereas a biclique can. A row may belong to several bicliques (row $3$ belongs to two), and the column set of a biclique need not be consecutive.

{\setlength{\arraycolsep}{2.5pt}\small
\[
A \;=\;
\smash[b]{\underbrace{\begin{pmatrix}
1&0&0\\
1&0&0\\
1&0&1\\
0&1&0\\
0&1&0\\
0&0&1
\end{pmatrix}}_{P}}
\overbrace{
%\left( I \,-\,
%\smash[b]{\underbrace{\begin{pmatrix}
%0&0&0\\
%0&0&0\\
%0&0&0
%\end{pmatrix}}_{Q}}
%\right)^{-1}
\smash[b]{\underbrace{\begin{pmatrix}
1&1&1&1&0&0&0&0\\
0&0&0&0&0&0&1&1\\
0&0&0&0&1&0&0&1
\end{pmatrix}}_{T}}
}^{F}
\;+\;
\smash[b]{\underbrace{\begin{pmatrix}
0&0&0&0&0&0&0&0\\
0&0&0&0&1&0&0&0\\
0&0&0&0&0&0&0&0\\
0&0&1&1&0&0&0&0\\
1&1&0&0&0&0&0&0\\
0&0&0&0&0&1&1&0
\end{pmatrix}}_{T'}}
\]
}

\vspace{15mm}

The three factorizations have sizes $20$, $22$ and $22$, all smaller than $\nnz(A)=27$ (they have that many nonzeros in their matrices), and evaluating $A\xx$ with any of them takes that many multiplications/additions instead of $27$. For the biclique factorization, for example, this means $c_1=x_1+x_2+x_3+x_4$, $c_2=x_7+x_8$ and $c_3=x_5+x_8$ are computed once and
added to $y_1,y_2,y_3$, to $y_4,y_5$, and to $y_3,y_6$, respectively.

\paragraph{How the families differ.}
Seen this way, the three families differ in \emph{how the factorization is found}, and the compression ratio they
obtain is the sparsity ratio $\nnz(A)/s$. Referencing uses a forest on the rows and is the most constrained
(each node is a row); grammars discover shared \emph{consecutive subsequences} and allow arbitrarily deep sharing;
bicliques share both a set of rows and a set of columns simultaneously, but have depth one.
Finding the smallest grammar is NP-hard~\cite{charikar2005smallest,storer1982data}, so practical compressors are
heuristics (\cref{sec:grammar}), and biclique extraction is also heuristic~\cite{feder1995clique,hernandez2014compressed};
for referencing, the optimal choice is a minimum spanning tree problem~\cite{bookstein1991compression,gagie2026set}
when the reference may be any row, and a more constrained problem when it must lie in a window~\cite{boldi2004webgraph} (\cref{sec:setdiff}).

\paragraph{What the view does not cover.}
The view concerns products by a vector and, as seen in \cref{sec:factprod}, sums and products of factorized matrices. The latter produce factorizations whose size we cannot bound, so they must be recompressed, and other operations whose output is again a matrix, such as the transitive closure, are not covered. The view
also does not cover the $k^2$-tree family, which follows a different principle: instead of factoring the matrix, it \emph{partitions}
the matrix recursively into quadrants and discards the empty ones; those formats
support a much wider set of algebraic operations (\cref{sec:quadtree}). The
distinction is useful: factorization-based formats work better for specific
operations like products, and partition-based ones support a more general algebra.

\chgbegin
\subsection{Relation to linear circuits}
\label{sec:circuits}

A \emph{linear circuit} over $\mathbb{K}$ with inputs $x_1,\dots,x_m$ and outputs $y_1,\dots,y_n$ is a directed acyclic graph whose sources are the inputs, in which every edge has a coefficient in $\mathbb{K}\setminus\{0\}$ and every non-source node (gate) computes the linear combination of its in-neighbors with the coefficients of the edges (multiplying on the left); each output is a designated node (a gate or an input). Its \emph{size} is the number of edges. The circuit \emph{computes} $A$ if $y_i=\sum_jA_{ij}x_j$ for all $i$. In the Boolean case with $\vee$ gates these are the \emph{rectifier networks}~\cite{jukna2013complexity}.
\Cref{fig:factorization} is exactly such a circuit: its sources are the column positions $x_j$, its gates are the nodes $F_{k\ast}$ and the rows $y_i$, and its edges are the nonzeros of $T$, $Q$, $P$ and $T'$.

\begin{lemma}\label{prop:circuits}
If $A$ is computed by a linear circuit of size $s$, then it has a sparse straight-line factorization of size at most $s+n$. If $A$ has a sparse straight-line factorization of size $s$, then it is computed by a linear circuit of size at most $s+n_P$, where $n_P\le n$ is the number of rows $i$ in which $P_{i\ast}\neq I_{i\ast}$.
\end{lemma}
\begin{proof}
See \cref{app:circuits}.
\end{proof}

Thus the model is equivalent to that of linear circuits of unbounded fan-in with size measured in wires, and the three families are particular shapes of such circuits: referencing is a circuit where each output gate has an edge from a single other output gate; bicliques are circuits of depth two (a node is an OR of its columns and each row is an OR of nodes); and grammars are circuits where all gates have fan-in two, plus the output gates. This view connects the compression of matrices with the lower and upper bounds of circuit complexity, which we use in \cref{sec:size}.
\chgend

\subsection{Products of factorized matrices}
\label{sec:factprod}

The view of \cref{sec:factor-vector} concerns products by a vector, but it also gives matrix--matrix products whose output is again a factorization of the same form. As far as we know the surveyed papers do not state this in general: the compositional biclique algorithm of \cref{sec:bicliquemult} is its depth-one case. The more general derivation that follows is ours.

Let $A=P_AF_A+T'_A\in\mathbb{K}^{n\times m}$ and $B=P_BF_B+T'_B\in\mathbb{K}^{m\times p}$ be sparse straight-line factorizations with $q_A$ and $q_B$ nodes, with fragments $F_A=(I-Q_A)^{-1}T_A$ and $F_B=(I-Q_B)^{-1}T_B$. The fragments of $A$ are vectors over the $m$ column positions of $A$, which are the row positions of $B$, and $P_B$ says which fragments of $B$ make up each of those rows. The matrix $T_AP_B\in\mathbb{K}^{q_A\times q_B}$ therefore relates the two structures through the shared dimension $m$.

\begin{theorem}\label{prop:product}
Let $(X\,;Y)$ denotes the horizontal concatenation of the matrices $X$ and $Y$, and
\[
Q_C=\begin{pmatrix}Q_B&0\\ T_A P_B&Q_A\end{pmatrix},\quad
P_C=\bigl(\,T'_AP_B\;;\;P_A\,\bigr),\quad
T_C=\begin{pmatrix}T_B\\ T_AT'_B\end{pmatrix},\quad
T'_C=T'_AT'_B .
\]
Then $AB=P_C\,(I-Q_C)^{-1}\,T_C+T'_C$ is a sparse straight-line factorization with $q_A+q_B$ nodes (those of $B$ first; $Q_C$ is strictly lower triangular), whose fragments are those of $B$ followed by $F_AB$.
\end{theorem}
\begin{proof}
Let $F_C=(I-Q_C)^{-1}T_C$ and split it as $\binom{F'}{F''}$ according to the
blocks. Since $F_C=Q_CF_C+T_C$, the first block gives $F'=Q_BF'+T_B$, so $F'=F_B$. The second gives
$F''=T_A P_B F'+Q_AF''+T_AT'_B$, so $F''=(I-Q_A)^{-1}T_A(P_BF_B+T'_B)=F_AB$. Then
$P_CF_C+T'_C=T'_AP_BF_B+P_AF_AB+T'_AT'_B=T'_AB+P_AF_AB=AB$.
\end{proof}

The nodes of $A$ now hold the products $F_{k\ast}B$ of its shared fragments by
$B$, each computed once, and every row of $AB$ is a combination of them and of
the fragments of $B$, plus the residual $T'_AB$. The only new matrices are the
four sparse products $T_AP_B$, $T'_AP_B$, $T_AT'_B$ and $T'_AT'_B$ between
\emph{pieces} of the factorizations, which can be computed with any sparse
product algorithm (for instance Schoor's algorithm \cite{schoor1982fast}, whose cost is dominated by the number of elementary products $\sum_{t=1}^m c_tr_t$, where $c_t$ and $r_t$ are the numbers of nonzeros in position $t$ of the left and right pieces). The time thus depends on the compressed pieces, and not on the nonzeros of $A$ and $B$, and the size of the result is
$\nnz(P_A)+\nnz(T'_AP_B)+\nnz(Q_A)+\nnz(Q_B)+\nnz(T_AP_B)+\nnz(T_B)+\nnz(T_AT'_B)+\nnz(T'_AT'_B)$.

\chgbegin
\begin{lemma}\label{prop:productcost}
For $X\in\mathbb{K}^{a\times b}$ and $Y\in\mathbb{K}^{b\times c}$ let $\pi(X,Y)=\sum_{t=1}^b\alpha_t\beta_t$, the number of elementary products needed to multiply them, where $\alpha_t$ and $\beta_t$ are the numbers of nonzeros in column $t$ of $X$ and in row $t$ of $Y$. The factorization of \cref{prop:product} is computed in
\[ \Oh\bigl(s_A+s_B+n+p+q_A+q_B+\pi(T_A,P_B)+\pi(T'_A,P_B)+\pi(T_A,T'_B)+\pi(T'_A,T'_B)\bigr) \]
arithmetic operations. Its size is
\begin{multline*} s_C=\nnz(P_A)+\nnz(Q_A)+\nnz(Q_B)+\nnz(T_B)\\ +\nnz(T'_AP_B)+\nnz(T_AP_B)+\nnz(T_AT'_B)+\nnz(T'_AT'_B), \end{multline*}
which is at most $s_A+s_B+\pi(T_A,P_B)+\pi(T'_A,P_B)+\pi(T_A,T'_B)+\pi(T'_A,T'_B)$, and its depth is at most $d_A+d_B$.
\end{lemma}
\begin{proof}
See \cref{app:products}.
\end{proof}

\chgend

\paragraph{Matrix sums.} Sums can also be supported: for $C=A+B$ (same shape) take 
\[ Q_C=\begin{pmatrix}Q_A&0\\0&Q_B\end{pmatrix}, \quad P_C=(P_A\;;\;P_B), \quad
T_C=\binom{T_A}{T_B}, \quad T_C'=T'_A+T'_B,\]
with $q_A+q_B$ nodes\chg{, size at most $s_A+s_B$ and depth $\max(d_A,d_B)$}.
The same identities hold over any \chg{(not necessarily commutative)} semiring, with the caveats of
footnote \ref{fn:semiring}. While this is computed in time proportional to the compressed sizes, the
representation obtained may not be optimal and may benefit from recompression, just like the proposed
summation algorithm for bicliques (\cref{sec:bicliquemult}).

\paragraph{Pointwise and Kronecker products.}
The Hadamard (element-wise) product $A\circ B$ is, in the Boolean case, the intersection, which $k^2$-trees already support (\cref{sec:algebra}). In the factorization view it is less direct than the sum: row $i$ of $A\circ B$ is $\sum_{k,l}P^A_{ik}P^B_{il}\,(F^A_k\circ F^B_l)$ plus residual terms, so a fragment is needed for every pair of nodes used together by some row, and we do not see a bound on the size of the result better than that of the explicit matrix\chg{. This operation, and closures such as the transitive closure, remain open for the model}.
The Kronecker product is, in contrast, simple over a \chg{commutative} ring (our own observation). Since $A\otimes B=(P_A\otimes P_B)(F_A\otimes F_B)$ and $F_A\otimes F_B=\big((I-Q_A)\otimes(I-Q_B)\big)^{-1}(T_A\otimes T_B)$ \chg{(both by the mixed-product property $(X\otimes Z)(Y\otimes W)=(XY)\otimes(ZW)$, which requires commutativity)}, a factorization of $A\otimes B$ has one node per pair $(k,l)$, with $P=P_A\otimes P_B$, $T=T_A\otimes T_B$ and $Q=Q_A\otimes I+I\otimes Q_B-Q_A\otimes Q_B$ (strictly lower triangular in lexicographic order of the pairs); to simplify we took $T'_A=T'_B=0$, which can always be assumed by regarding each residual row as one more node \chg{(at a cost of at most $n_A$ and $n_B$ additional entries; below, $s_A$ and $s_B$ are the sizes after this step)}. \no{Its size is $\Oh(s_As_B)$ for factorizations of sizes $s_A$ and $s_B$, so, as $\nnz(A\otimes B)=\nnz(A)\nnz(B)$, the compression ratio of the Kronecker product is the product of the ratios of its factors.}
\chgbegin
\begin{lemma}\label{prop:kron}
If $A$ and $B$ have no dead nodes (Lemma~\ref{lem:dead}), the factorization of $A\otimes B$ above has size at most $4s_As_B$, so, if $\mathbb{K}$ has no zero divisors, the compression ratio $\nnz(A\otimes B)/s$ is at least a quarter of the product of the ratios of its factors.
\end{lemma}
\begin{proof}
See \cref{app:products}.
\end{proof}

\chgend

\paragraph{Special cases.}
\begin{itemize}
\item \emph{Only the left factor is compressed} ($q_B=0$, so $B$ is explicit): $AB=P_A(I-Q_A)^{-1}(T_AB)+T'_AB$. The structure $(P_A,Q_A)$ is untouched and only $T_A$ and $T'_A$ are multiplied by $B$. For referencing ($P_A=I$, $T'_A=0$, $T_A=A'$) this is $AB=(I-Q_A)^{-1}(A'B)$: the \emph{same reference forest} with the residual $A'B$. When $B$ is a dense matrix, as in the GNNs mentioned in \cref{sec:reference}, these are the two stages of \citet{alves2025accelerating} (a sparse--dense product with $A'$, followed by the substitution along the forest); when $B$ is sparse, the substitution can be skipped and the output left in compressed form.
\item \emph{Only the right factor is compressed} ($q_A=0$): $AB=(AP_B)F_B+AT'_B$, which keeps $(Q_B,T_B)$ and multiplies $A$ by $P_B$ and $T'_B$.
\item \emph{Bicliques} ($Q_A=Q_B=0$, with sets $(S^A_k,T^A_k)$ and $(S^B_l,T^B_l)$): $N_{kl}=|T^A_k\cap S^B_l|$ counts how many positions connect the biclique $k$ of $A$ with the biclique $l$ of $B$, and in the Boolean case $N_{kl}>0$ is all that matters. The four terms of \cref{eq:bicmat} are the four terms of the product, and the unions of column sets that the algorithm of \cref{sec:bicliquemult} materializes for every biclique of $A$ are the fragments $F_AB$; the table $\mathit{marks}$ of that algorithm is a row-wise representation of $P_B$, and looking up in it the elements of each $T^A_k$ computes $T_AP_B$. Materializing the fragments gives again a biclique representation ($Q_C=0$), as in that paper \cite{araya2026compressed}; keeping the nonzeros of $T_AP_B$ as edges inside $Q_C$, instead, is what our generic algorithm does and saves space, but the result is not anymore a biclique.
\end{itemize}

A dual combination for referencing, in which $A$ is used through its row structure and $B$ through the structure of its columns, underlies \cref{eq:matmat}: if $W'$ is the residual matrix of $B\T$ (the matrix whose rows are the differences with its references), then $AB=(I-Q_A)^{-1}(A'\,{W'}\T)(I-Q_B)^{-\top}$. Its output is a two-sided factorization, which does not have the form of \cref{eq:factorization}. The Boolean product on compressed binary trees of \cref{sec:refmatmat} also uses the rows of $A$ and the columns of $B$, and produces its output by columns, so that it can be the right operand of the next product.

\paragraph{What it does not give.}
We do not have a bound on the size of the result. The matrix $T_AP_B$ can have up to $q_Aq_B$ nonzeros, the other three products can have fill-in, and nothing prevents the output from being larger than a factorization of $AB$ found from scratch (or even than $\nnz(AB)$). It is therefore a valid but not necessarily good factorization, which can be improved by \emph{recompression}: discarding the nodes that are no longer reached by any row, merging nodes with identical fragments, or running again the search for a good factorization on $AB$, as in the biclique algorithms. The experiments of \citet{araya2026compressed} show both sides: skipping the recompression makes the product up to $20\times$ faster, but the space of the output degrades by $7\%$--$40\%$ in most cases.
In iterated products the effect compounds: squaring doubles the number of nodes ($q_C=2q_A$ for $C=A^2$) \chg{and the depth at most doubles, so $A^{2^t}$ has $2^tq_A$ nodes; recompression (\cref{sec:recomp}) is thus needed} to stay compressed. Finally, when only products by vectors are needed one should usually avoid forming $AB$ at all, since $(AB)\xx=A(B\xx)$ costs the sizes of the two factorizations, which is typically less than that of $AB$; forming the product pays off when it is reused (powers, closures) or needed in explicit form.

\chgbegin
\subsection{Recompression}
\label{sec:recomp}

Sums and products give valid but not necessarily small factorizations, and recompression is the problem of finding a smaller factorization of the same matrix. The strongest form is to run again the compressor of the family on the explicit matrix, but that costs the full compression time. We show instead a set of simple, linear-time steps that work directly on the factorization, possibly reducing and never increasing its size and depth. We say that a node is \emph{dead} if no row uses it (it has no path of nonzeros of $Q$ to a nonzero column of $P$) or if its fragment is zero. The \emph{in-degree} $d_{in}(k)$ of node $k$ is the number of nonzeros in row $k$ of $(Q\;;\;T)$ and its \emph{out-degree} $d_{out}(k)$ is the number of nonzeros in column $k$ of $\binom{Q}{P}$.

\begin{lemma}[Dead nodes]\label{lem:dead}
Deleting the nodes that no row uses, and the nodes with in-degree zero (together with their outgoing edges), yields a factorization of the same matrix in time $\Oh(s)$, of size and depth at most the original.
\end{lemma}
\begin{proof}
See \cref{app:recomp}.
\end{proof}

\begin{lemma}[Nodes of degree one]\label{lem:degree}
A node with out-degree or in-degree one can be eliminated without changing the matrix nor increasing the size or the depth. A factorization without dead nodes can be reduced to one in which every node has in-degree and out-degree at least two in $\Oh(s)$ expected time, provided that no coefficients cancel when parallel edges are added up; this is always the case, for example, over the Boolean semiring. In general, the time is $\Oh(s(1+c))$, where $c\le s$ is the number of cancellations.
\end{lemma}
\begin{proof}
See \cref{app:recomp}.
\end{proof}

Finally, identical fragments can be merged with a probabilistic technique that requires $\mathbb{K}$ to be a field.

\begin{lemma}[Identical fragments]\label{lem:merge}
Let $\mathbb{K}=\mathbb{F}$ be a field, let $\mathcal{S}\subseteq\mathbb{F}$ with $|\mathcal{S}|=\sigma$, and let $\mathbf{r}\in\mathcal{S}^m$ be uniformly random. Computing $\mathbf{h}=F\mathbf{r}$ by forward substitution takes $\Oh(s)$ operations, and merging the nodes with the same value of $\mathbf{h}$ (each into the node of lowest index) yields a factorization of size at most $s$ that is correct with probability at least $1-q^2/(2\sigma)$, and that can be verified in $\Oh(s)$ time with an independent random vector.
\end{lemma}
\begin{proof}
See \cref{app:recomp}.
\end{proof}

Over $\mathbb{Z}$, as for Boolean matrices, the same procedure can be applied after reducing the coefficients modulo a prime $p\ge\sigma$ larger than the magnitude of any difference of coefficients of fragments, so that distinct fragments remain distinct. We see no efficient way to verify deterministically the eliminations of Lemma~\ref{lem:merge}. A factorization on which none of the three steps applies is not minimal in general; characterizing how far such factorizations can be from the minimum one would be a first step towards a theory of recompression.

\subsection{The size of the smallest factorization}
\label{sec:size}

Let $s^\ast(A)$ be the minimum size of a sparse straight-line factorization of $A$. By \cref{thm:fact}, this is the number of operations of the product by a vector, so it measures the compressibility of the matrix for that purpose. Trivially $s^\ast(A)\le\nnz(A)$ (take $T'=A$).

\begin{lemma}\label{prop:sqrt}
Every sparse straight-line factorization of an $n\times m$ matrix $A$ has size $s\ge\max(r,c)\ge\sqrt{\nnz(A)}$, where $r$ and $c$ are the numbers of nonzero rows and nonzero columns of $A$, and this is tight.
\end{lemma}
\begin{proof}
See \cref{app:size}.
\end{proof}

\paragraph{Worst-case bounds.}
By Lemma~\ref{prop:circuits}, upper bounds for rectifier networks and linear circuits carry over; in particular, the biclique partitions of \citet{krapivin2026optimal} give $s^\ast(A)\le(1+o(1))\,n^2/\lg n$ for every $n\times n$ Boolean matrix (over the Boolean semiring, or over any ring $\mathbb{K}$, regarding the matrix as having entries $0$ and $1$ of $\mathbb{K}$; this works because the partition has no overlaps). Here the matrix is regarded as a directed graph with possible loops, or as a bipartite graph with $n$ vertices on each side, for which the optimal size is $(1+o(1))\,n^2/\lg n$ (Lupanov); the $\frac12 n^2/\lg n$ of undirected graphs holds for symmetric matrices. A counting argument shows that this order is optimal for the model.

\begin{lemma}\label{prop:lower}
For every $n\ge2$ there is an $n\times n$ Boolean matrix $A$ such that every sparse straight-line factorization of $A$ with nonzero coefficients in $\{-1,1\}$ has size at least $(n^2-2\lg n-1)/(4\lg n+3)=(1+o(1))\frac14 n^2/\lg n$.
\end{lemma}
\begin{proof}
See \cref{app:size}.
\end{proof}

Thus $s^\ast(A)=\Theta(n^2/\log n)$ for the hardest matrices, which is the size reachable by biclique partitions, and no cleverer sharing structure can improve on it in the worst case if the coefficients are restricted to $\{-1,1\}$, or to a set whose size is polynomial in $n$.
The count generalizes to a fixed density $\gamma\in(0,1)$: for $m=\nnz(A)=\gamma n^2$, using $\lg\binom{n^2}{m}\ge h_2(\gamma)n^2-\lg(n^2+1)$, the same argument shows that some matrices need size $(1-o(1))\frac{h_2(\gamma)}{4}\,n^2/\lg n$, within a factor $4$ of the $(1+o(1))\,h_2(\gamma)\,n^2/\lg n$ that biclique partitions obtain for that density (Nechiporuk; \citealp[Thm.~1.6]{krapivin2026optimal}).
For sparser matrices, $\binom{n^2}{m}\ge(n^2/m)^m$ and $s^\ast(A)\le m$ give that some matrices need size at least $(m\lg(n^2/m)-2\lg n-1)/(4\lg n+3)$; if $m=\Theta(n^{1+\epsilon})$ with $0\le\epsilon<1$, this is $(1-\epsilon-o(1))m/4$. In the worst case, then, polynomially sparse matrices cannot be compressed by more than a constant factor, in contrast with the $\Theta(\log n)$ factor available for dense matrices: these techniques work by exploiting particular properties of the matrices.

\paragraph{The model is more general than the families.}
The measure $s^\ast(A)$ can be asymptotically better than the measures of each family: $\Delta(\mathcal{S})$ for referencing~\cite{gagie2026set}, the size of the smallest grammar of the CSRV sequence for grammars~\cite{charikar2005smallest}, and the number of residual edges for bicliques.

The first separation uses the following matrix: the columns are partitioned into $k$ groups of $k$ columns, and there is one row for each pair of groups $a<b$, with ones in the columns of both groups.

\begin{lemma}\label{prop:pairs}
Referencing (even with deletions) and grammars can be worse than the smallest factorization by a factor $\Theta(\nnz(A)^{1/3})$, and so can be referencing relatively to grammars.
\end{lemma}
\begin{proof}
See \cref{app:size}.
\end{proof}

\begin{lemma}\label{prop:gap}
Let $L^n$ be the $n\times n$ Boolean lower-triangular matrix of ones. Then $s^\ast(L^n)<2n$, but every decomposition of $L^n$ into bicliques and residual edges, with overlaps allowed and sums in the Boolean semiring, has size $\Omega(n\log n)$.
\end{lemma}
\begin{proof}
See \cref{app:size}.
\end{proof}

Thus each of the three families can be polynomially or logarithmically worse than the smallest factorization: referencing and grammars by a factor $\Theta(N^{1/3})$ on the matrices of Lemma~\ref{prop:pairs}, and bicliques by a logarithmic factor on $L^n$, where depth is needed. 

\paragraph{Computing $s^\ast(A)$.}
The decision problem ``is $s^\ast(A)\le s$?'' is in NP, but we do not know if it is NP-hard, or how well it can be approximated. Related problems are hard: finding the shortest linear program over $\mathbb{F}_2$ (the smallest circuit with fan-in two) is NP-hard~\cite{boyar2013logic}, and so are several versions of biclique cover and partition; but our measure allows unbounded fan-in, counts wires and includes the residual, so a reduction is not immediate.

\paragraph{Parallelism.}
So far we have minimized the size, which matches the work of a product by a vector; if we consider parallelism, the depth $d$ is another important parameter. Let the level of a node be the length of the longest chain that ends in it. Evaluating the nodes level by level, the work is still $\Oh(s)$ and, as the sum at each node can be computed with a reduction tree, the parallel time with $p$ processors is $\Oh(s/p+d\log s)$. Biclique factorizations have $d=1$, and so parallel time $\Oh(s/p+\log s)$. A grammar with $|R|$ rules can be \emph{balanced} in $\Oh(|R|)$ time to retain $\Oh(|R|)$ size and have depth $\Oh(\log s)$~\cite{ganardi2021balancing}, which enables parallel products in time $\Oh(s/p+\log s)$. A principled solution for referencing is a bounded-diameter MST, but the problem is NP-hard for any diameter $\ge4$~\cite{garey1979computers} and hard to approximate~\cite{barilan2001generalized}; it is possible, however, to obtain logarithmic depth with an MST weight $s'=\Oh(s\log n)$~\cite{marathe1998bicriteria}, giving parallel time $\Oh(s'/p+\log^2 s)$.
For general factorizations the problem is to obtain an equivalent circuit of small depth. Over commutative rings the circuit computes a polynomial of degree one in the input, so the VSBR theorem~\cite{vsbr1983} yields a circuit of depth $\Oh(\log s)$, although of impractical size $s''=\Oh(s^3)$, and a parallel time of $\Oh(s''/p+\log^2 s)$. Such a result is not possible in general for non-commutative rings or Boolean circuits~\cite{ladner1975circuit,kosaraju1990parallel,nisan1991lower}.
\chgend

\section{Comparison and Guidelines}
\label{sec:comparison}

\Cref{tab:compare} compares the families along the dimensions that matter when choosing a format. The general
tradeoff is that, the more compression, the less computation performance. Still, compressed computation is 
viable overall: some formats achieve space and time reduction at the same time.

\begin{table}[t]
\centering\small
\setlength{\tabcolsep}{3.5pt}
\caption{Qualitative comparison of the families of computation-friendly formats. \checkmark: supported and evaluated; (\checkmark): doable and promising, but not evaluated; ---: not studied. 
Compression and speed are coarse ratings relative to the plain (CSR or dense) representation (more $+$ is better, $-$ is worse); speed refers to the product that each family supports best (by a vector for referencing and grammars, matrix--matrix for bicliques and $k^2$-trees), and both depend strongly on the matrix.}
\label{tab:compare}
\begin{tabular}{@{}p{2.1cm}ccccccrr@{}}
\toprule
Family & \makecell[c]{Right\\product} & \makecell[c]{Left\\product} & \makecell[c]{Mat\\$\times$mat} & \makecell[c]{Sum \&\\ more} & \makecell[c]{Real\\values} & \makecell[c]{Random\\access} & Compression & $~~~$Speed\\
\midrule
Referencing & \checkmark & (\checkmark) & (\checkmark) & --- & (\checkmark) & (\checkmark) & ++++ & + \\% Window heuristic (BV) or MST \\%& Needs good row order; social graphs\\
% PageRank hasta 2.6 × 2.6× más rápido (eu-2015, uk-2014, it-2004), pero ≈ 1 × 1× en grafos sociales (twitter-2010, enwiki)
Grammar & \checkmark & \checkmark & --- & --- & \checkmark & (\checkmark) & ++ & ++ \\%& RePair; column reordering \\%& Dense / unstructured data; compression time\\
% Hasta 12 × 12× más rápido que la matriz densa en las muy comprimibles (Census, Airline78, Covtype) y más lento que ella en las poco comprimibles (Susy, Higgs, Optical)
Biclique & \checkmark & (\checkmark) & \checkmark\ & (\checkmark) & --- & --- & + & +++ \\% Heuristic extraction (adaptive version faster) \\%& Extraction cost; research prototypes\\
% Solo ≤ 2.4 × ≤2.4× en el producto por vector, pero en matriz–matriz hasta 34 × 34× (como acelerador) y 289 × 289× (sin recompresión) sobre el baseline
$k^2$-tree & (\checkmark) & (\checkmark) & \checkmark & \checkmark & (\checkmark) & \checkmark & +++ & - \\%& Fast ($\approx1\,\mu$s/edge) \\%& $\log v$ overhead; Boolean only\\
% Baseline 14 14– 53 × 53× más rápido en multiplicación
%Qdag + prejoin & --- & --- & \checkmark & (\checkmark) & (\checkmark) & \checkmark & Fast; $+30\%$ space for projections \\%& Closure not yet done\\
\bottomrule
\end{tabular}
\end{table}

\subsection{Cross-cutting observations}

\paragraph{The gain follows the compression ratio.}
In all the experiments on products by vector, the speedup is a fraction of the reduction in the number of symbols processed: the compression ratio of
\cref{fig:ref,fig:gram-scatter,fig:bic-pr} upper-bounds the speedup, because the compressed traversal has worse locality (the nonzeros of $A'$ are scattered in the referencing format, nonterminals and bicliques require random writes to $\yy$) or involves
extra work (decoding, pointer chasing, rank operations). The cache-friendly deployment of kache-trees~\cite{farina2025cache} makes this explicit: with the same asymptotic work, fewer cache misses give up to $60\%$
speedups. A compressed format is therefore most useful for compressed computation if its structure is also hardware-friendly; 
cache-obliviousness and vectorization are important topics for future work.
On the other hand, matrix--matrix products may benefit superlinearly from convenient compression formats: the biclique representation
(\cref{sec:biclique}) obtains speedups of $34\times$--$289\times$, which exceed the space reduction ($3\times$--$8\times$) by a large margin.

\paragraph{The matrix decides the family.}
Web graphs ordered by URL are ideal for referencing; social networks are not, but bicliques and $k^2$-trees perform reasonably
on them; machine-learning matrices with few distinct values per column (categorical features, one-hot encodings, quantized values) are the
natural domain of grammars and of CLA; and clustered Boolean relations (knowledge graphs, RDF) suit the $k^2$-tree. In matrices with little
structure (like \texttt{Susy}, where more than $98\%$ of the entries are nonzero and there are over $20$ million distinct values) none of the formats
helps, and the dense or the plain sparse representation is the right choice.

\paragraph{Preprocessing is the price.}
All the formats are static, and the cost of building them ranges from $1.1\,\mu$s per edge for $k^2$-trees~\cite{arroyuelo2025evaluating} and a few minutes for the largest ML matrices with RePair, to hours for the recompression of the
largest Web graphs for referencing~\cite{francisco2022graph} and to heavy MST computations for the best references~\cite{gagie2026set}. Since the intended use is to multiply the same matrix many times (PageRank iterations,
iterative solvers), the amortization is generally acceptable, but it means that the formats are not meant for matrices that are operated once or change often.

\paragraph{Compositionality and recompression.}
Because the product of two factorizations is again one (\cref{sec:factprod}), the families can be chained, but the output size is not bounded and grows with every product. Skipping recompression made the biclique product up to $20\times$ faster at a cost of $7\%$--$40\%$ in space~\cite{araya2026compressed}; whether to recompress is a trade-off between the cost of the next operation and that of the extraction.

\paragraph{Working space.}
The auxiliary space is small but not null: grammar-based products need $\Oh(|R|)$ words, referencing needs the $n$ partial results, and algebraic operations produce outputs that may be much larger than the inputs (the transitive closure being the extreme).
The \emph{peak memory} of the grammar-based product is a small multiple of the compressed size (Table~5 of \citet{ferragina2022improving}), and the benefits in time are largest for matrices that compress well (\cref{fig:gram-both,fig:gram-scatter}).
%; measurements on managed runtimes (such as CLA on Java) are, in addition, harder to compare with those of native code.

\paragraph{Parallelism.}
Row-block partitioning is the natural way to parallelize the product for all the families (each block is compressed and multiplied independently), at the price of losing
inter-block redundancy: the grammar results showed that sometimes the split even \emph{improves} compression because blocks have little in common. The recursive structure of $k^2$-trees gives
parallelism in the algebra, but thread-creation overhead cancels the gain on very sparse inputs.

\subsection{Which format when?}

As a practical summary of the evidence in the surveyed papers:
\begin{itemize}
\item \emph{Large Web or host graph, many products with the same matrix} (PageRank and similar power-method computations):
  reference-based compression via WebGraph with the residual trick of \cref{alg:ref}; expect up to $2$--$3\times$ speedup on page graphs; consider bicliques when the graph is less
  copy-friendly (the biclique format reached $2.4\times$ on \texttt{it-2004}).
\item \emph{Real matrix with few distinct values, large and tall} (ML design matrices): grammar compression with column reordering, row blocks and 16 or more threads; expect compression within about 50\% of \texttt{xz} and sometimes smaller, smaller peak memory than CLA, and speed comparable to or better than dense when the compression ratio is high.
\item \emph{Boolean matrices that must be composed} (graph databases): $k^2$-trees (smallest), their kache or prejoin variants (faster) or a CSR/CSC baseline (about $4\times$ the space of the $k^2$-tree and several times faster; space drops to $2\times$ if storing just CSR \cite{araya2026compressed}). For products of large Web graphs, bicliques with Schoor's algorithm (up to $34\times$--$289\times$ faster than the baseline depending on how the bicliques are used).
\item \emph{Moderately compressible matrices, tight memory} : combine a computation-friendly format with a lightweight entropy coder, accepting a decoding cost (\texttt{re\_iv}, \texttt{re\_ans}).
\end{itemize}

In the terms of \cref{sec:factorization}, the choice depends on which matrices have small factorizations of each shape, and on how much it costs to find them.

\Cref{tab:code} lists the implementations of the surveyed works that we found.

\begin{table}[t]
\centering\small
\caption{Public implementations of the works discussed.} % ``n/g'': the paper does not give a URL.}
\label{tab:code}
\begin{tabular}{@{}p{2.6cm}p{2.8cm}p{7.3cm}@{}}
\toprule
Work & Family & Code\\
\midrule
%\citet{gagie2026set} & Set difference & n/g\\
\citet{boldi2004webgraph} & Referencing (build) & \url{https://github.com/vigna/webgraph/} \\%\url{https://law.di.unimi.it}\\
\citet{alves2025accelerating} & Referencing (CBM) & \url{https://github.com/cbm4scale}\\
\citet{ferragina2022improving} & Grammars & \url{https://gitlab.com/manzai/mm-repair}\\% (matrices on Kaggle)\\
\citet{francisco2022graph} & Bicliques & \url{https://github.com/koeppl/matrixbiclique}\\
\citet{araya2026compressed} & Bicliques (mat-mat) & \url{https://github.com/nicolasaraya/BicliqueMultiplication}\\
\citet{arroyuelo2025evaluating} & $k^2$-tree algebra & \url{https://github.com/adriangbrandon/rpq-matrix}\\
%\citet{farina2025cache} & Kache-trees & n/g\\
%\citet{arroyuelo2026boosting} & Qdags, prejoin & n/g\\
\citet{tosoni2025toward} & Various & \url{https://gitlab.com/ftosoni/green-lossless-spmv}\\
\bottomrule
\end{tabular}
\end{table}

\no{
\subsection{Through the lens of the factorization}

\Cref{sec:factorization} suggests a second way of comparing the first three families, by the \emph{shape} of the factorization $A=P(I-Q)^{-1}T+T'$ that each of them seeks. Referencing has $q=n$ nodes, one per row, and a tree-shaped $Q$ (potentially deep, so a long chain of dependencies; windows or limits on the depth keep it manageable);
grammars have as many nodes as rules, a $Q$ that is a DAG of arbitrary depth, and share only consecutive runs of (value, column) pairs inside rows; bicliques have $Q=0$ (depth one) and share a full set of columns among a set of rows. The cost of products is the same expression for all,
so that the choice among them depends only on \emph{which matrices have small factorizations of each shape} and on \emph{how much it costs to find them}: row-similar matrices (Web graphs under a good ordering) for referencing; matrices with repeated partial rows, such as those with few distinct values per column, for grammars; and graphs with dense communities for bicliques.
A single measure that bounds the three (the minimum size of a sparse straight-line factorization) is a natural but open notion (\cref{sec:open}). The product of two factorizations is again one (\cref{sec:factprod}), which includes the biclique algorithms of \cref{sec:bicliquemult} as the depth-one case. Two things escape the lens: operations with matrix output other than sums and products (closures), and the $k^2$-tree family, which is based on partitioning and for which the whole algebra is available.
}

\subsection{Relation to theory}

Every method in \cref{sec:reference,sec:grammar,sec:biclique,sec:quadtree} has time $\Theta(\text{size of the representation})$ to compute $A\xx$: the product reads each symbol once. Therefore the possible
speedup is exactly the sparsity-relative compression ratio, and the interesting questions are about \emph{which matrices admit small representations}. The formal guarantees available so far connect size to
combinatorial measures of the input: the symmetric-difference compressibility $\Delta(\mathcal{S})$ for referencing~\cite{gagie2026set}, the size of the smallest grammar that generates the CSRV sequence for grammars~\cite{ferragina2022improving,charikar2005smallest} (which practical compressors approximate only heuristically), the smallest biclique partition \cite{krapivin2026optimal} (which again, implementations approach heuristically with small biclique covers) and, for $k^2$-trees, bounds in terms of $\nnz(A)$ and of clustering~\cite{debernardo2023faster}. There is no common measure; a unification of these notions into a notion of ``computation-friendly compressibility'' is
an open challenge.

\section{Open Problems and Research Directions}
\label{sec:open}

The surveyed papers conclude with a number of future-work items, and the comparison of \cref{sec:comparison} reveals further gaps. We group them into five themes: supported algebra, compression methods, theoretical results, applications, and computation models.

\subsection{Boolean and other algebras}

The Boolean algebra of \cref{sec:quadtree} is compositional but compresses only by spatial clustering; it would be interesting to extend those formats so as to capture row similarity or repeated patterns while preserving functionality. Two-dimensional block trees and depth-first $k^2$-trees (\cref{sec:k2variants}) are first steps in this direction, as they share repeated submatrices inside a $k^2$-tree. However, for the former it is not clear that the algorithms of \cref{sec:quadtree} keep their time bounds, and for the latter the variant that shares submatrices, which compresses best, was the slowest in the products evaluated by \citet{carmona2025depth}.

From the other formats, which capture repetitions, bicliques are the family that most clearly supports Boolean matrix--matrix 
products and sums (\cref{sec:bicliquemult}), and it is compositional. This technique is generalized in \cref{prop:product},
which holds on rings and may yield new ways to address this operation. A problem, in both cases, is how to maintain the quality of 
the compression along the composition of operations; we further discuss this issue in \cref{sec:open-construction}.
\cref{eq:matmat} gives another initial idea to multiply with references; its evaluation and analysis 
is open. Other pointwise matrix operations, like Hadamard products, may not be so easy to perform on these formats.

As for transitive closures, the algorithm of \cref{sec:algebra} converts any solution for matrix products into one
for transitive closure, for recursive decompositions. The improved multiplication results on kache-trees have 
already been seen to carry over transitive closures.\footnote{Personal communication with the authors.} This
operation is still missing on all the other formats.
Other operations that are not addressed in the surveyed papers include the sparse--sparse product with a real semiring such as $(\mathbb{R},\min,+)$ for shortest paths, and matrix powers $A^k$ applied to vectors (combining
repeated matrix--vector products with iterated squaring, for example).
Over a commutative semiring, the products of \cref{thm:fact} can be evaluated with semiring operations alone (writing $\sum_{i<q}Q^i$ for $(I-Q)^{-1}$), so grammar-based formats extend to any semiring in which the values to be multiplied are the semiring elements of $T$. Referencing formats are more delicate: when a row deletes elements of its reference they need additive inverses (a ring), and over semirings without them, such as $(\mathbb{R},\min,+)$ for shortest paths or the Boolean semiring, only references that insert elements (insertion compressibility $I(\mathcal{S})$) are available. Bicliques, in turn, may overlap in an idempotent semiring (e.g., Boolean) but must partition the edges over $\R$. Which compression ratios these restricted forms achieve on matrices of interest, and whether closures over such semirings (all-pairs shortest paths, reachability) can be accelerated, are open questions.

\subsection{Compression and recompression}
\label{sec:open-construction}

Computing good references, good grammars and good bicliques is expensive. For referencing, the incremental MST of \citet{gagie2026set} is a step towards practicality, and can be complemented with adding Steiner nodes and with sketching and approximate nearest-neighbor techniques to find references (\cref{sec:setdiff}). For grammars, row permutations and co-clustering are
mentioned as future work in~\cite{ferragina2022improving}, together with different compressors per block; for bicliques, better extraction techniques may yield improved compression. Streaming or distributed
construction, and construction with explicit space/time trade-offs (for example, spending more preprocessing to get a higher compression ratio) are largely unexplored. 

\paragraph{Recompression along the algebra.}
Sums and products of factorized matrices (\cref{sec:factprod}) give valid but not necessarily small factorizations, and the problem compounds in chains of products, powers, and closures. A general problem is then to keep the quality of the compression along the operations of the algebra: either with recompression procedures cheap enough to be applied after each operation, which trade the cost of the next operation against that of the recompression (as the experiments with bicliques show~\cite{araya2026compressed}), or with formats that stay compact under the operations without extra work, as the $k^2$-tree does for the Boolean algebra (\cref{sec:algebra}). \chg{\Cref{sec:recomp} gives linear-time steps that never increase the size, but whether recompression procedures can provably yield factorizations of size related to the minimum one is open. The same holds for closures (the transitive closure over the Boolean semiring, all-pairs shortest paths over $(\R,\min,+)$) and for Hadamard products, which the model does not directly support.} 

\paragraph{Dynamic matrices.}
All formats here are static. The Boolean-algebra paper identifies dynamic $k^2$-trees~\cite{arroyuelo2019faster} as a way to support updates; for referencing and grammar formats, updates could break the references or the rules, and
incremental maintenance is not known to us. Applications such as evolving graphs (where PageRank is recomputed periodically) would benefit. For time-evolving graphs, the differential compressed binary trees of \cref{sec:refmatmat} have been extended to support the insertion and deletion of arcs and of whole time frames directly on the compressed structure, with algorithms for earliest-arrival paths and for transitive closure over it~\cite{nelson2019thesis}; whether the speedups for products by a vector and by a matrix survive such updates is, as far as we know, unexplored.

\subsection{Theory}

The most important opening of the view of \cref{sec:factorization} is theoretical. The three factorization-based families fit a single scheme, whose size is a measure of the matrix, a kind of ``straight-line complexity'' of the product by a vector, analogous to the size of the smallest grammar of a string. What is the minimum size of such a factorization of a given matrix, how does it extend 
the measures of each family ($\Delta(\mathcal{S})$, the size of the smallest grammar of the CSRV sequence, number of residual edges), and is it approximable? \chg{We show in \cref{sec:size} that it is at least $\sqrt{\nnz(A)}$, that it is $\Theta(n^2/\log n)$ for the hardest $n\times n$ Boolean matrices (\cref{sec:size}), and that it can be polynomially smaller than the measures of referencing and grammars, and logarithmically smaller than that of bicliques; the complexity of computing it, its approximability and the minimum depth for a given size remain open.} Which classes of compressed representations admit products in time linear in their size, and which do not, in the spirit of~\citet{abboud2020impossibility}? 
%The explicit-position trick of CSRV suggests a characterization, but we are not aware of one. \gonzalo{lo dudo}

In these terms, recompression has a simple abstract form: given a factorization $A = P(I-Q)^{-1}T+T'$ of a matrix, find a smaller one for the same matrix. \no{Two steps are cheap and exact in all the families (\cref{sec:factprod}): discarding the nodes that no row reaches, and merging nodes with identical fragments.}\chg{Three steps are cheap in all the families (\cref{sec:recomp}): discarding the nodes that no row reaches, eliminating the nodes of in- or out-degree one, and merging nodes with identical fragments (the last one, probabilistically).} A third step is the expensive heuristic: running again the search for a good factorization (a reference search, a grammar compressor, or a biclique extractor) on the result or on the explicit matrix. Characterizing the factorizations that cannot be reduced by the first two steps, and how far they are from the minimum one, would be a first step towards a theory of recompression, and would connect the previous question with the construction problem above.

\subsection{Applications}

The use of these techniques in applications would not only benefit them, but also bring new challenges to the area.

\paragraph{Interaction with lossy compression: ML and LLM weights.}
Lossless grammar compression gave good results on classical machine-learning matrices; the open question is its effect on modern, large and dense
weight matrices of neural networks, which are mostly not sparse and have few exact repetitions, except after quantization, where the value domain is small and the CSRV model fits well.
Combining lossy quantization with lossless computation-friendly formats~\cite{ferragina2022improving,han2016deep} is a natural but unexplored experiment at scale. In GNNs, the compression of the adjacency matrix with the
set-difference formats has been used to accelerate the training~\cite{alves2025accelerating}.

\paragraph{Semirings and database workloads.}
\citet{ferragina2022improving}
mention that grammar compression could handle semiring-annotated data, simulating joins on logical matrices ($+$ replaced by OR and $\times$ by AND), and columnar databases. The $k^2$-trees and prejoining~\cite{arroyuelo2022optimal,arroyuelo2025evaluating,arroyuelo2026boosting} tie matrix
multiplication with worst-case-optimal joins and regular path queries on graph databases, and could be a vehicle to bring compressed matrices to the optimizer of a relational or graph database engine.

\paragraph{Integration in libraries.}
The implementations are research prototypes. Making them available as formats in sparse-algebra frameworks (GraphBLAS-like interfaces, tensor compilers~\cite{kepner2016mathematical,chou2018format,donenfeld2022unified}) would make them more widely available for
practitioners. Since the surveyed papers use different datasets, hardware and baselines, a common benchmark suite would also make the families comparable.

\subsection{Computation models}

While reducing space gives opportunities to execute on higher and faster levels
of the memory hierarchy, poor locality of reference may ruin those advantages.
Kache-trees (\cref{sec:kache}) showed that layout matters as much as asymptotics; a fully cache-oblivious variant, in which the number of I/Os matches the best known for sparse matrices~\cite{dusefante2018cache,gleinig2022io}, is announced as future work~\cite{farina2025cache}. The
irregular memory access of compressed formats is a challenge for GPUs, where
the uncompressed formats (ELL, HYB) were designed around
regularity~\cite{bell2009implementing}; GPU-oriented decode-first formats are reviewed in \cref{sec:decompress}. Interesting results in this respect
are reported by \citet{tosoni2026right}, who port the right product on 
grammar-compressed matrices to GPUs: by layering the grammar (inserting 
pass-through nodes so that each level reads only the previous one) the device 
footprint is $4$--$8\times$ smaller than a CSR baseline on genotype matrices, 
with times within a small factor of cuSPARSE and lower energy. 
Distributed-memory versions would require
partitioning compressed structures and communicating partial results efficiently.

\section{Conclusions}
\label{sec:conclusions}

We have surveyed computation-friendly formats for sparse matrices: representations that are smaller than the plain sparse format and allow matrix--vector products, and in some cases matrix--matrix products, sums or closures, to be computed in time
proportional to their size. The four families we identified exploit different kinds of redundancy (spatial clustering, similar rows, repeated subsequences, dense bipartite blocks).

The main conceptual message of the survey is that three of these families, referencing, grammars and bicliques, are \emph{one} idea in three appearances: each finds a sparse straight-line factorization $A=P(I-Q)^{-1}T+T'$ of the matrix,
and the same forward substitution evaluates the right product $A\xx$ and its transposed backward substitution evaluates the left product $\yy\T A$, in time proportional to the number of nonzeros of the factors (\cref{sec:factorization}). Everything that
distinguishes the families, apart from the matrices they suit, is how the factorization is found: a minimum spanning tree over symmetric differences, a grammar compressor, or an extractor of dense bipartite subgraphs.
This view turns the question of which format to use into the question of which matrices admit small factorizations of which shape, and at what cost. The product of two factorizations is again a factorization (\cref{sec:factprod}), although of unbounded size, so a compositional algebra with recompression is the natural next step. This view also suggests that a theory of the minimum size of such factorizations is the missing foundation of the area\chg{; its first elements (the equivalence with linear circuits, worst-case bounds, and separations from each family) are summarized in \cref{sec:factorization}}. The $k^2$-tree family stands apart, built on recursive partitioning rather than on factorization, and it
offers a complete algebra of Boolean matrices, but obtains the least speedups over a plain sparse representation

The empirical results show that the approach is practical. Speedups of $2$--$3\times$ on Web graphs with referencing, of up to $2.4\times$ for products by a vector with bicliques, up to $34\times$ (up to $289\times$ when the format is kept) for Boolean matrix--matrix products with bicliques, running times below those of the dense product on the most compressible
machine-learning matrices with grammars, and speedups of $2$--$13\times$ for chains of Boolean matrix products over qdags, come with space reductions that range, depending on the family and on the data, from modest to an order of magnitude. At the same time, the speedups remain below the compression ratio, preprocessing is costly and the formats are static, and the families have been evaluated
in isolation. The most valuable next practice-oriented step for the field may be a common benchmark and a common framework, in which the formats can be compared, combined and extended to new operations and hardware.

\bibliographystyle{ACM-Reference-Format}
\bibliography{refs}

\appendix
\chgbegin
\section{Proofs}
\label{app:proofs}

This appendix contains the proofs of the results of \cref{sec:factorization} that are not given there, except the short ones of Lemmas~\ref{lem:ref}, \ref{lem:gram} and \ref{lem:bic}.

\subsection{Linear circuits}
\label{app:circuits}

\begin{proof}[Proof of Lemma~\ref{prop:circuits}]
Given the circuit, number its gates in topological order. Let $T$ hold the coefficients of the edges from inputs to gates, $Q$ those between gates (strictly lower triangular by the order), and, for each output $i$ that is a gate $k$, let $P_{ik}=1$; if the output is an input $x_j$, let $T'_{ij}=1$. By induction on the order, the fragment $F_{k\ast}=((I-Q)^{-1}T)_{k\ast}$ is the linear form computed by gate $k$, so $PF+T'$ is $A$. The number of nonzeros is at most the number of edges plus one per output.
In the other direction, take as gates the $q$ nodes plus one gate $y_i$ per row; edges $T$, $Q$ and $T'$ are edges of the circuit, and each nonzero of $P$ is an edge from a node to a row gate. If $P$ is the identity (as in referencing), the node $i$ is identified with the gate $y_i$ and no extra edges are needed. The coefficients are those of the matrices.
\end{proof}

\subsection{Products, sums and Kronecker products}
\label{app:products}

\begin{proof}[Proof of Lemma~\ref{prop:productcost}]
The cost and the size follow from the definition of $P_C,Q_C,T_C,T'_C$ and from the fact that the product $XY$ is computed with $\pi(X,Y)$ elementary operations \cite{gustavson1978two,schoor1982fast} and has at most that many nonzeros. For the depth, in $Q_C$ the nodes of $B$ depend only on nodes of $B$, and a node of $A$ depends on nodes of $A$ (through $Q_A$) and on nodes of $B$ (through $T_AP_B$), but never the other way, so a chain of dependencies is a chain of $A$-nodes followed by a chain of $B$-nodes.
\end{proof}

\begin{proof}[Proof of Lemma~\ref{prop:kron}]
Let $F_A=(I-Q_A)^{-1}T_A$ and $F_B=(I-Q_B)^{-1}T_B$ be the fragments of $A$ and $B$.
By the mixed-product property, $A\otimes B=(P_AF_A)\otimes(P_BF_B)=(P_A\otimes P_B)(F_A\otimes F_B)$, and, using it again,
\begin{align*} F_A\otimes F_B&=\big((I-Q_A)^{-1}T_A\big)\otimes\big((I-Q_B)^{-1}T_B\big)\\&=\big((I-Q_A)^{-1}\otimes(I-Q_B)^{-1}\big)(T_A\otimes T_B). \end{align*}
Since $(I-Q_A)^{-1}\otimes(I-Q_B)^{-1}=\big((I-Q_A)\otimes(I-Q_B)\big)^{-1}$ and, by bilinearity,
\[ (I-Q_A)\otimes(I-Q_B)=I\otimes I-Q_A\otimes I-I\otimes Q_B+Q_A\otimes Q_B=I-Q, \]
we obtain $F_A\otimes F_B=(I-Q)^{-1}(T_A\otimes T_B)=(I-Q)^{-1}T$ and therefore $A\otimes B=P(I-Q)^{-1}T$.

\smallskip\noindent\emph{Triangularity.} Index the nodes by the pairs $(k,l)$, with $1\le k\le q_A$ and $1\le l\le q_B$, in lexicographic order. The entry of $Q$ at row $(k,l)$ and column $(i,j)$ is $(Q_A)_{ki}[j=l]+[i=k](Q_B)_{lj}-(Q_A)_{ki}(Q_B)_{lj}$, where $[\cdot]$ is $1$ if the condition holds and $0$ otherwise. The first term is nonzero only if $i<k$ and $j=l$, the second only if $i=k$ and $j<l$, and the third only if $i<k$ and $j<l$. In the three cases $(i,j)$ precedes $(k,l)$ lexicographically, so $Q$ is strictly lower triangular.

\smallskip\noindent\emph{Size.} The matrices $P$ and $T$ have $\nnz(P_A)\nnz(P_B)$ and $\nnz(T_A)\nnz(T_B)$ nonzeros, and the sum of these is at most $(\nnz(P_A)+\nnz(T_A))(\nnz(P_B)+\nnz(T_B))\le s_As_B$ (the matrices $T'$ are zero). $Q$ has at most $\nnz(Q_A)q_B+q_A\nnz(Q_B)+\nnz(Q_A)\nnz(Q_B)$ nonzeros, one term for each part of its definition. If there are no dead nodes then every node has a nonzero in its row of $(Q\;;\;T)$, so $q_A\le s_A$ and similarly $q_B\le s_B$, and each of the three terms is at most $s_As_B$. In total, $s\le4s_As_B$.

\smallskip\noindent\emph{Ratio.} Each pair of nonzeros $A_{ij}\ne0$ and $B_{kl}\ne0$ gives the entry $A_{ij}B_{kl}$ of $A\otimes B$ at a different position, which is nonzero if $\mathbb{K}$ has no zero divisors; thus $\nnz(A\otimes B)=\nnz(A)\nnz(B)$ and the ratio is at least $\nnz(A)\nnz(B)/(4s_As_B)$.
\end{proof}

\subsection{Recompression}
\label{app:recomp}

\begin{proof}[Proof of Lemma~\ref{lem:dead}]
A node that no row uses contributes nothing to $PF$; the nodes with in-degree zero have zero fragments, so their outgoing edges contribute zero to the fragments that use them. The set of nodes with a path to $P$ is computed by a backward traversal of $Q$ in $\Oh(s)$ time, and the nodes with in-degree zero by a pass over the edges (repeating the removal in order for the nodes whose only inputs were removed, which takes $\Oh(s)$ in total with a queue). The removal does not change the remaining fragments.
\end{proof}

\begin{proof}[Proof of Lemma~\ref{lem:degree}]
Suppose the only outgoing edge of $k$ goes to $t$ (a node or a row, $t>k$) with coefficient $w$. Replace each incoming edge $u\to k$ with coefficient $a$ by an edge $u\to t$ with coefficient $wa$, adding it to the edge $u\to t$ if it exists (an input $x_j$ becomes an entry of $T$ if $t$ is a node, and of $T'$ if $t$ is a row; a node $u$ becomes an entry of $Q$ or of $P$), and delete $k$. Since $u<k<t$ the new edges keep $Q$ strictly lower triangular, the fragments of all the remaining nodes are unchanged because $F_{t\ast}$ includes $wF_{k\ast}=\sum_u wa\,F_{u\ast}$, and the size does not increase: $d_{in}(k)+1$ edges are replaced by at most $d_{in}(k)$. Suppose now that the only incoming edge of $k$ comes from $u$ (a node or an input) with coefficient $a$. Replace each outgoing edge $k\to t$ with coefficient $w$ by $u\to t$ with coefficient $wa$, and delete $k$; the size decreases by at least one. A path through a new edge corresponds to a path of the original with one more node, so the depth does not increase. 

Each elimination decreases the size by at least one, so there are at most $s$ eliminations. 
To bound the time, we eliminate the nodes in an order that makes each elimination cost only the degree of the node. Parallel edges created by the redirection are detected with hashing and added up. We assume no dead nodes, so both in- and out-degrees are at least one.

\smallskip\noindent\emph{First pass.} We process the nodes in decreasing order of index, eliminating each node that has out-degree one at the moment it is reached. Eliminating $k$, with target $t>k$, may change only the out-degree of the nodes $u<k$ that had an edge to $k$, which are processed later ($d_{out}(u)$ decreases by one if $u$ already had an edge to $t$, as in the triangle $u\to k\to t$, $u\to t$), and the in-degree of $t$, which was processed already. So when $k$ is reached its out-degree is final for this pass. Moreover, the eliminations only add edges that go to nodes with an index larger than $k$, so the in-edges of $k$ are still the original ones and the elimination costs $\Oh(d_{in}(k))$. This handles whole maximal chains of nodes of out-degree one at once, as each node of the chain sends its edges directly to the end of the chain, and cascades like the one of the triangle, since $u$ is processed after $k$. At the end no node has out-degree one, and the pass costs $\Oh(\sum_kd_{in}(k))=\Oh(s)$.

\smallskip\noindent\emph{Second pass.} We process the nodes in increasing order of index, eliminating each node $k$ that has in-degree one at the moment it is reached, with its only incoming edge from $u<k$. Now $k$ has out-degree $d_{out}(k)\ge2$ and it costs $\Oh(d_{out}(k))$, since the previous eliminations of this pass %moved edges to sources $u'<k'<k$ and 
did not touch the out-edges of $k$. The elimination changes the out-degree of $u$, which becomes $|(N(u)\setminus\{k\})\cup N(k)|\ge d_{out}(k)\ge2$, where $N$ denotes the set of out-neighbors (so no node of out-degree one is created), and changes the in-degrees of the targets $t>k$ of $k$, which are processed later ($d_{in}(t)$ decreases by one if $u$ already had an edge to $t$). At the end, no node has in-degree or out-degree one, and the pass costs $\Oh(s)$ time, expected due to the hashing.

\smallskip\noindent\emph{Cancellations.} If, when parallel edges are added up, a coefficient becomes zero, the edge disappears and the degrees may drop more than in the analysis above. For example, in the second pass the out-degree of $u$ may drop to one. Then we repeat the two passes. Each cancellation removes at least one edge, so the number of repetitions is at most one plus the number of cancellations.
\end{proof}

\begin{proof}[Proof of Lemma~\ref{lem:merge}]
The vector $\mathbf{h}$ is the vector $\mathbf{u}$ of the forward substitution of \cref{thm:fact} applied to $\xx=\mathbf{r}$ (without the last step). If two fragments $f\ne g$ are different, $(f-g)\cdot\mathbf{r}=0$ with probability at most $1/\sigma$, by the Schwartz--Zippel lemma \cite{schwartz1980fast,zippel1979probabilistic} applied to a nonzero linear polynomial; a union bound over the $\binom{q}{2}$ pairs bounds the probability that two different fragments collide. If $F_{k\ast}=F_{l\ast}$ with $k<l$, merging means redirecting every edge leaving $l$ to leave $k$ instead (adding coefficients) and deleting $l$ and its incoming edges. The targets of these edges are greater than $l>k$, so $Q$ remains strictly lower triangular, the fragments of all the other nodes are unchanged because $F_{l\ast}=F_{k\ast}$, and the size does not increase. Finally, the result can be verified, with failure probability $1/\sigma$, by comparing its product by an independent random vector with the one of the original factorization (Freivalds' technique), in $\Oh(s)$ time.
\end{proof}

\subsection{The size of the smallest factorization}
\label{app:size}

\begin{proof}[Proof of Lemma~\ref{prop:sqrt}]
A nonzero column $j$ of $A$ has a nonzero in some fragment (so column $j$ of $T$ is nonzero) or in $T'$, which gives $s \ge c$ nonzeros. A nonzero row $i$ needs a nonzero in row $i$ of $P$ or of $T'$, or, if $P$ is the implicit identity, a nonzero in row $i$ of $T$, $Q$ or $T'$; so $s\ge r$ too. Finally $\nnz(A)\le rc\le\max(r,c)^2$.
For tightness, note that the $n\times n$ all-ones matrix has $s^\ast=2n=2\sqrt{\nnz(A)}$ (one node, with $n$ entries in $T$ and $n$ in $P$).
\end{proof}

\begin{proof}[Proof of Lemma~\ref{prop:lower}]
Let $s$ be the maximum over all $n\times n$ Boolean matrices of the minimum size of such a factorization, and suppose $s\le n^2$ (otherwise the claim is trivial). A factorization of size at most $s$ has at most $s$ nodes (every node has at least one outgoing edge, or else it is dead and can be removed) and is specified by (i) its number of nonzeros (one of $s+1$ values) and (ii) a list of at most $s$ entries, each one given by the matrix it belongs to (4 options), its row and column (each at most $\max(n,q)\le\max(n,s)\le n^2$, since the matrices have $n$ or $q\le s$ rows and columns, and $s\le n^2$) and its sign (2 options), that is, by at most $8(n^2)^2=8n^4$ options. Hence there are at most $(s+1)(8n^4)^s$ different factorizations, and since each of the $2^{n^2}$ Boolean matrices needs one, $n^2\le\lg(s+1)+s(4\lg n+3)\le2\lg n+1+s(4\lg n+3)$ (using $s+1\le2n^2$).
\end{proof}

\begin{proof}[Proof of Lemma~\ref{prop:pairs}]
Consider the following matrix. Let $k\ge3$, $m=k^2$, $n=\binom{k}{2}$ and $N=\nnz(A)=2kn$. Let $G_1,\dots,G_k$ be a partition of the columns $\{1,\dots,m\}$ into groups of $k$ columns, and let $A$ have one row for each pair $a<b$, with ones in the columns of $G_a\cup G_b$. Then $s^\ast(A)\le 2k^2=\Oh(N^{2/3})$, because the biclique decomposition that has, for each $a$, the biclique formed by the $k-1$ rows $\{a,\cdot\}$ and the columns $G_a$ has $\sum_a((k-1)+k)=k(2k-1)\le2k^2$ nonzeros in $P$ and $T$, and $T'=0$. On the other hand:
\begin{enumerate}[label=(\alph*)]
\item every referencing of $A$ has size at least $N$, for any arrangement of the rows and any group partition. Two different rows share at most one group, so they differ in at least $2k$ columns. A row referenced to another row has therefore at least $2k$ nonzeros in $A'$ (plus the reference), and a row with no reference has $2k$ nonzeros: a reference costs as much as an explicit representation, and the total is at least $2kn=N$;
\item if $G_a=\{a,a+k,a+2k,\dots\}$ (interleaved groups), the size of every grammar for $A$ is at least $N$ (in the sense of Lemma~\ref{lem:gram}). Here the row $\{a,b\}$ is the sequence of columns $a,b,a+k,b+k,\dots$ (we write the terminal $\langle1,j\rangle$ as $j$). Two consecutive columns $c,c'$ of a row determine the pair $\{c\bmod k,c'\bmod k\}=\{a,b\}$, so no two rows have a common pair of consecutive terminals, and inside a row all such pairs are different. Thus every rule of the grammar, whose expansion has at least two terminals, occurs only once in the sequence, so it is used exactly once in the derivations (a rule used twice would have its expansion twice). The derivation forest of the sequence has then $L=N$ leaves, $\rho$ roots (the symbols of $C$ that are not $\$$) and $r$ internal nodes, each with two children, all different; counting nodes, $r=L-\rho$, and the size of the grammar is $2r+\rho=2N-\rho\ge N$ since $\rho\le L=N$ (unused rules can only increase it);
\item if $G_a=\{(a-1)k+1,\dots,ak\}$ (contiguous groups), there is a grammar of size $3k^2=\Oh(N^{2/3})$: a chain of $k-1$ rules builds each block $G_a$ (of size $2(k-1)$ each), and the row $\{a,b\}$ is the two symbols $N_aN_b$, for a total of $2k(k-1)+2n=3k(k-1)\le3k^2$.
\end{enumerate}
\end{proof}

\begin{proof}[Proof of Lemma~\ref{prop:gap}]
Referencing each row to the previous one gives $A'=I$ and $n-1$ references, so $s\le2n-1$ (the factorization needs no deletions, so it is valid over the Boolean semiring too). A biclique decomposition is a Boolean circuit of depth two with OR gates, of size equal to its size as a decomposition (each residual edge is a direct wire from an input to an output); and it is known that depth-two OR circuits for $L^n$ need $\Omega(n\log n)$ wires \cite{jukna2013complexity}.
\end{proof}

\chgend

\section{Formats that Require Decompression: A Brief Tour}
\label{sec:decompress}

For completeness, we briefly review ``decode-first'' matrix compression formats, in which the matrix must be decoded, wholly or in
blocks, before (or while) it is operated upon, so that the cost of the operation does not shrink with the size of the
compressed representation, that is, operation time is always $\Omega(\nnz(A))$. This is not the focus of the survey and the review is deliberately short, but it is useful in
order to place the computation-friendly formats in context. The border is not sharp: the formats range from those that
decompress everything into memory, to those that decode small pieces on the fly inside the multiplication loop, to
those that operate directly on compressed pieces (the latter being the subject of the main body of this survey).
\Cref{tab:decode} summarizes the classes.

\begin{table}[t]
\centering\small
\caption{Decode-first classes of formats. ``Lossy'' means that the matrix itself is modified.}
\label{tab:decode}
\begin{tabular}{@{}p{3.0cm}p{4.3cm}p{2.5cm}p{2.3cm}@{}}
\toprule
Class & Examples & What must be decoded & Gain\\
\midrule
General-purpose & \texttt{gzip}, \texttt{bzip2}, \texttt{xz}, \texttt{zstd} & Whole block / & Disk, network\\
 &  & matrix & \\
Floating-point & fpzip, ZFP, Gorilla & Blocks of values & Disk, memory\\
Index/value coding & CSR-DU, CSX, DCSR; & Streams while & Memory \\
 & Elias--Fano; PFOR & multiplying & bandwidth \\
Blocked / structured & BCSR, VBR, CSB, DCSC, & Nothing (extra & Fewer indices; \\
 & ELL, DIA, HYB & zeros stored) & SIMD/GPU\\
Lightweight columnar & RLE, dictionary (CLA, AWARE) & Mostly operate on & Space; time\\
 & & compressed data & \\
Lossy (ML) & Pruning, quantization, low rank & Weights, at & Space; time\\
 & & inference & \\
\bottomrule
\end{tabular}
\end{table}

\subsection{General-purpose and floating-point compressors}

General-purpose compressors based on Lempel--Ziv~\cite{ziv1977universal} (\texttt{gzip}/DEFLATE~\cite{deutsch1996deflate}, LZMA as in \texttt{xz}),
on the Burrows--Wheeler transform~\cite{burrows1994block} (\texttt{bzip2}), or on a mixture of LZ and entropy coding
(\texttt{zstd}~\cite{collet2021zstandard}) are often excellent in compression ratio, and they are the natural baseline against which matrix compressors are
compared~\cite{ferragina2022improving}; however, they ignore the two-dimensional structure and force the whole matrix (or block)
to be decompressed before any operation. In the experiments of \citet{ferragina2022improving}, a multithreaded
\texttt{gzip}-based product that decompressed each block at every iteration had the best memory footprint, together with the grammar-based
formats, but was by far the slowest, in some cases more than $40\times$ slower. Moreover, grammar compressors (and thus the Lempel--Ziv family, whose size is
within a logarithmic factor of grammars) are the class covered by the impossibility results of \cref{sec:limits} when
they are applied to the plain sequence of values.

For floating-point data, specialized lossless compressors exploit the
structure of the IEEE representation: fpzip~\cite{lindstrom2006fast} predicts values and encodes residuals; Gorilla~\cite{pelkonen2015gorilla} XORs consecutive
values, which suits time series. ZFP~\cite{lindstrom2014fixed} compresses small blocks of a multidimensional array (in a fixed-rate mode, which permits random access to blocks)
and is lossy or lossless depending on the mode. They reduce memory and disk, but the blocks are decoded before arithmetic is applied.

\subsection{Index and value coding for sparse formats}

Sparse matrix--vector multiplication on modern hardware is memory-bound, so reducing the number of bytes moved can speed it up even if
the number of arithmetic operations is unchanged~\cite{williams2009optimization}. This motivated compressing the index arrays of CSR:
CSR-DU~\cite{kourtis2008optimizing} replaces column indices by small \emph{deltas} with variable-length encodings, CSX~\cite{kourtis2011csx}
detects substructures (horizontal, vertical, diagonal runs and blocks), and \citet{willcock2006accelerating}
proposed delta-coded variants of CSR. These formats decode a few indices inside the multiplication loop. The time per product is still proportional to $\nnz(A)$ (and
may get higher due to decoding), and the gain comes from bandwidth.
More recently, \citet{schatzle2026fast} compress CSR with an entropy coder (a variant of asymmetric numeral systems, dtANS) designed to be decoded in parallel on a GPU \emph{inside} the multiplication; the result is up to $11.77\times$ smaller than the smallest cuSPARSE format and, on matrices with at least $2^{25}$ nonzeros, usually faster, with a best speedup of $3.48\times$. Decoding cost is hidden by the memory-bound nature of the product, but the work is still proportional to $\nnz(A)$.

The same idea underlies the compression of integer sequences: universal and Golomb-like codes~\cite{elias1975universal,golomb1966run} (the WebGraph framework also
uses them for gaps~\cite{boldi2005codes}), Elias--Fano~\cite{elias1974efficient,vigna2013quasi,ottaviano2014partitioned}, and SIMD-friendly frame-of-reference codes with patching~\cite{zukowski2006super,lemire2015decoding}. Graph processing systems run algorithms
over gap-coded adjacency lists by decoding on the fly, as Ligra+~\cite{shun2015smaller} does. All of these decode before using the data, and they are
orthogonal to, and combinable with, the computation-friendly formats: for example, the residual graph $A'$ of \cref{sec:reference} or the sequence $C$ of \cref{sec:grammar} can be further encoded with these methods, at the price of a decoding step.

\subsection{Blocked and structured formats}

A different way to save indices is to store \emph{blocks} of nonzeros: in blocked CSR (BCSR) a single index describes a small dense block (of, say, $2\times2$ or $4\times4$ entries), which allows register
blocking and SIMD, at the price of storing explicit zeros to complete the blocks~\cite{pinar1999improving,im2004sparsity}; variable block structures~\cite{vuduc2005fast}, compressed sparse blocks
(CSB), which makes products with $A$ and $A\T$ equally efficient~\cite{buluc2009parallel}, and doubly compressed formats for hypersparse matrices
(DCSC)~\cite{buluc2008representation} refine this idea. GPU formats such as ELLPACK, diagonal (DIA), and hybrids (HYB)~\cite{bell2009implementing} pad or regularize rows to match the hardware.
They are lossless but their ``compression'' is a layout choice rather than a reduction of the information: the work is proportional to the number of stored entries, including padding.
The $k^2$-tree formats of \cref{sec:quadtree} are in some sense a more principled version of blocking, since blocks are discovered hierarchically.

\subsection{Lightweight schemes from column stores, and CLA}

Column-oriented database systems use lightweight compression (run-length, dictionary, bit-packing) and, whenever possible, operate directly on the compressed
columns~\cite{abadi2006integrating}. The matrix analogue is Compressed Linear Algebra (CLA)~\cite{elgohary2018compressed,elgohary2019compressed}, which groups
correlated columns and uses offset lists, run-length encoding and dictionary coding per group to support matrix--vector products on the compressed form,
with a cost related to the number of distinct values and runs, and its successors such as AWARE~\cite{baunsgaard2023aware}, which chooses formats according to the workload. \citet{donenfeld2022unified} propose a
compiler that generates code for sparse computations over lossless-compressed formats. These systems straddle our border: simple schemes such as run-length encoding are computation-friendly,
whereas more complex ones need partial decoding; they are the main baseline for \cref{sec:grammar}.

\subsection{Lossy compression for machine learning}

In deep learning the matrices are weights, and it is common to \emph{change} them: pruning produces sparse matrices~\cite{hoefler2021sparsity}, quantization reduces the precision of the values~\cite{gholami2021survey}, low-rank
factorizations replace a matrix by the product of two thin ones, and the combination of pruning, quantization and Huffman coding of Deep Compression~\cite{han2016deep} achieves large reductions of size, with entropy coding followed by decoding at inference time~\cite{marino2023efficient}. Lossy techniques
may damage accuracy in a data- and model-dependent way, but they are orthogonal to lossless ones: as observed in~\cite{ferragina2022improving}, one can
apply lossy compression first and lossless compression later, which is the setting in which the formats of the main body are most attractive, since they remain exact.

\subsection{Sparse tensor frameworks and GraphBLAS}

Frameworks for sparse algebra provide formats and kernels over semirings, without compression beyond sparsity: GraphBLAS~\cite{kepner2016mathematical,davis2019algorithm,kepner2011graph} expresses graph
algorithms as sparse matrix operations over semirings, and tensor compilers~\cite{kjolstad2017tensor,chou2018format} generate kernels for combinations of per-dimension formats (dense, compressed, singleton, \dots), including
compressed sparse fibers~\cite{smith2015tensor}. Their algebraic model is the natural target for integrating the computation-friendly formats of this survey (\cref{sec:open}), since they already abstract the format of an operand from the kernel that uses it.
Matrix multiplication in these frameworks follows classical sparse algorithms~\cite{gustavson1978two,gilbert1992sparse} that are the uncompressed counterparts of the baseline of \cref{sec:algebra}.

\end{document}